\documentclass[printer,online]{jfp}

\usepackage[utf8]{inputenc}
\usepackage{graphicx}
\usepackage{amssymb}
\usepackage{stmaryrd}
\usepackage{rotating}
\usepackage{paralist}
\usepackage{epstopdf}
\usepackage{algpseudocode}
\algnotext{EndFor}
\algnotext{EndIf}
\algnotext{EndProcedure}
\usepackage{algorithm}
\usepackage{tikz}
\usepackage[skip=6pt]{caption, subcaption}
\usepackage{graphicx}
\usepackage{appendix}
\usepackage{soul}
\usepackage{tabularx}
\usepackage{booktabs}
\usepackage{mathrsfs}
\usepackage{amsbsy}
\usepackage{xspace}
\usepackage{harpoon}

\usepackage{listings}
\usepackage{tikz-cd} 
\usepackage{slashbox}
\usepackage{accents}
\usepackage{float}
\usepackage{xtab}

\usepackage{IEEEtrantools}
\usepackage{stmaryrd}

\newcommand\Dom{\text{Dom}}

\newcommand\newop[2]{\newcommand{#1}{\mathop{\mathrm{#2}}\nolimits}}
\newop{\plot}{plot}
\newop{\elem}{elem}
\newop{\ev}{ev}
\newop\rearrange{rearrange}
\newop\match{match}
\newop\const{const}
\newop\extensionality{ext}
\newop\Def{Def}

\newcommand{\IH}{\stackrel{I.H.}{=}}

\newcommand{\clog}[1]{\vconst{log}}
\newcommand{\op}{\text{op}}
\newcommand{\RR}{\mathbb{R}}

\newcommand{\logR}{\mathcal{R}}
\newcommand{\sem}[1]{\llbracket #1\rrbracket}
\newcommand{\supp}{\text{supp}}

\usepackage{ascii} 
\usepackage[T1]{fontenc}

\usepackage{booktabs}
\usepackage{balance}
\usepackage{array}
\newcolumntype{?}{!{\vrule width 0.8pt}}

\usepackage{mathpartir}
\usepackage{pgfplots}
\usetikzlibrary{arrows}
\usetikzlibrary{shapes}
\usetikzlibrary{matrix}
\usetikzlibrary{shapes.geometric}
\usetikzlibrary{backgrounds}
\usetikzlibrary{calc}

\definecolor{forestgreen}{rgb}{0.13, 0.55, 0.13}
\colorlet{myblue}{blue!70!black}
\colorlet{mygreen}{green!30!black}
\colorlet{mypurple}{purple!70!black}
\colorlet{myorange}{orange!70!black}
\colorlet{myred}{red!40!black}
\colorlet{myyellow}{yellow!60!black}

\def\allowspace{\hskip 0em plus 1em minus 0em\relax}

\newcommand{\cod}[1]{\texttt{#1}}
\newcommand{\codkw}[1]{{\color{myblue}\texttt{#1}}}

\lstdefinelanguage{fsmooth}%
{morekeywords={
  if,then,else,let,in,
  op2,op1,map,map2,foldl,
  scanl,scanr,map3,fun,reduce,
  shift1L, shift1R,
  op,var,J,
  sum,prod,dot,fst,
  zip,
  true,false,
  OnesLike,
  ZerosLike,
  pair,proj,cos,sin,exp
  },%
  sensitive,%
  morecomment=[l]//,%
  morecomment=[s]{/*}{*/},%
  morestring=[b]",%
  morestring=[b]',%
  showstringspaces=false,%
  morecomment=[s][\color{gray}]{@}{\ },%
    breaklines=true,%
  mathescape=true,%
showspaces=false,
showtabs=false,
showstringspaces=false,
breakatwhitespace=true,
  aboveskip=1pt,
  belowskip=1pt,
  lineskip=-0.2pt,
  numbersep=5pt,
  numberstyle=\tiny\ttfamily,
  basicstyle=\small\ttfamily,
  keywordstyle=\bfseries\color{blue!70!black},%
  columns=fullflexible,
  frame=single,
  escapeinside={(*@}{@*)},
  literate={->}{$\rightarrow\;$}{2}
           {<}{$\langle$}{1}
           {>}{$\rangle$}{1}
}[keywords,comments,strings]%

\lstMakeShortInline[columns=fixed]!

\definecolor{listingbg}{RGB}{240, 240, 240}
\definecolor{mygray}{RGB}{225, 225, 225}
\newcommand{\code}[1]{\lstinline[language=ML,columns=fixed,basicstyle=\ttfamily]|#1|}

\newcommand{\kw}[1]{\mathsf{#1}} 

\newtheorem{theorem}{Theorem} 
\newtheorem{fact}{Fact} [section]
\newtheorem{corollary}[theorem]{Corollary}
\newtheorem{proposition}[theorem]{Proposition}
\newtheorem{definition}[fact]{Definition}
\newtheorem{lemma}[theorem]{Lemma} 
 \newcommand{\true}{\ensuremath{\kw{true}}}
 \newcommand{\false}{\ensuremath{\kw{false}}}

\newcommand{\text}[1]{\textnormal{#1}}

\newcommand{\lett}{\codkw{let}}
\newcommand{\inn}{\codkw{in}}
\newcommand{\iif}{\codkw{if}}
\newcommand{\then}{\codkw{then}}
\newcommand{\elsee}{\codkw{else}}

\newcommand{\tab}{\text{$\;\;\;$}}

\newcommand{\lafsharp}{$\widetilde{\textsc{F}}$}
\newcommand{\fsmooth}{\lafsharp}

\newcommand{\ladsl}{$\widetilde{\textsc{M}}$}
\newcommand{\system}{$\text{d}\widetilde{\textsc{f}}$\xspace}

\newcommand{\viteratek}{\codkw{ifold}}
\newcommand{\vifoldk}{\viteratek}

\newcommand{\vbuildk}{\codkw{build}}
\newcommand{\varraygen}[1]{\cod{[\ensuremath{#1}]}}
\newcommand{\varray}[1]{\cod{[\ensuremath{#1_0,\ldots,#1_{n-1}}]}}
\newcommand{\vlengthk}{\codkw{length}}
\newcommand{\vgetk}{\codkw{get}}
\newcommand{\vtrue}{\codkw{true}}
\newcommand{\vfalse}{\codkw{false}}
\newcommand{\vpairk}{\codkw{pair}}
\newcommand{\vfstk}{\codkw{fst}}
\newcommand{\vsndk}{\codkw{snd}}

\newcommand{\expr}{\text{e}}

\newcommand{\genexprind}[2]{\text{$\text{#1}_{#2}$}}
\newcommand{\exprind}[1]{\genexprind{\expr}{#1}}
\newcommand{\val}{\text{v}}
\newcommand{\valind}[1]{\genexprind{\val}{#1}}

\newcommand{\vbuild}[2]{\text{\vbuildk{} #1 #2}}
\newcommand{\vget}[2]{#1[#2]}
\newcommand{\vlength}[1]{\vlengthk{}\,#1}
\newcommand{\vifthenelse}[3]{\text{\iif{} #1 \then{} #2 \elsee{} #3}}

\newcommand{\viterate}[3]{\text{\viteratek{} #1 #2 #3}}
\newcommand{\vabsk}{\codkw{fun}}
\newcommand{\vabs}[2]{\text{\vabsk{} #1 \cod{->} #2}}

\newcommand{\vapp}[2]{\text{#1 #2}}
\newcommand{\vletn}[2]{\text{\lett{} #1 = #2}}
\newcommand{\vlet}[3]{\text{\lett{} #1 = #2 \cod{in} #3}}

\newcommand{\vpair}[2]{\text{(#1, #2)}}

\newcommand{\vmore}[1]{\text{$\overline{\text{#1}}$}}

\newcommand{\vconst}[1]{\text{\cod{#1}}}

\newcommand{\typewrapper}[1]{\text{\cod{#1}}}
\newcommand{\typewrappernonterm}[1]{\textnormal{#1}}
\newcommand{\typemat}{\typewrappernonterm{M}}

\newcommand{\typepair}[2]{\text{#1 $\times$ #2}}
\newcommand{\typet}{\typewrappernonterm{T}}

\newcommand{\typeind}[1]{\text{$\typet{}_{#1}$}}

\newcommand{\funarrow}{\text{$\Rightarrow$}}
\newcommand{\typefun}[2]{\text{$\vmore{#1} \Rightarrow #2$}}
\newcommand{\typefunone}[2]{\text{#1 $\Rightarrow$ #2}}
\newcommand{\typeindex}{\typewrapper{Index}}

\newcommand{\typedouble}{\typewrapper{Double}}

\newcommand{\typebool}{\typewrapper{Bool}}
\newcommand{\typearray}[1]{\typewrapper{Array<#1>}}
\newcommand{\typevector}{\typewrapper{Vector}}
\newcommand{\typematrix}{\typewrapper{Matrix}}

\def\typenum{\typewrappernonterm{Num}}
\newcommand{\typedoubled}{\typewrapper{DoubleD}}
\newcommand{\typevectord}{\typewrapper{VectorD}}
\newcommand{\typematrixd}{\typewrapper{MatrixD}}

\newcommand{\evalsto}{\text{ $\leadsto$ }}

\def\tabt{\hspace*{0.35cm}}

\def\vcaddcard{\text{$+$}}
\def\vcsubcard{\text{$-$}}

\def\vcget{\text{\codkw{get}}}
\def\vclength{\text{\codkw{length}}}

\newcommand{\difftranstwo}[2]{\text{${\color{mygreen}\mathcal{D}\llbracket}$#1${\color{mygreen}\rrbracket}$}}
\newcommand{\difftrans}[1]{\difftranstwo{#1}{dx}}

\newcommand{\difftransalt}[1]{\text{${\color{mygreen}\mathcal{D}\llbracket} #1 {\color{mygreen}\rrbracket}$}}
\newcommand{\difftranstypealt}[1]{\text{${\color{mygreen}\mathcal{D_T}\llbracket} #1 {\color{mygreen}\rrbracket}$}}

\newcommand{\difftransp}[1]{\text{${\color{mygreen}\mathcal{P}\llbracket}$#1${\color{mygreen}\rrbracket}$}}
\newcommand{\difftranst}[1]{\text{${\color{mygreen}\mathcal{E}\llbracket}$#1${\color{mygreen}\rrbracket}$}}
\newcommand{\difftranstype}[1]{\text{${\color{mygreen}\mathcal{D_T}\llbracket}$#1${\color{mygreen}\rrbracket}$}}

\newcommand{\diffvarprefix}[1]{\text{$\dot{\text{#1}}$}}
\newcommand{\ddiffvarprefix}[1]{\diffvarprefix{\diffvarprefix{#1}}}

\newcommand{\pterm}[1]{\text{\vfstk{} (#1)}}
\newcommand{\dterm}[1]{\text{\vsndk{} (#1)}}

\newcommand{\adpair}[2]{\text{(#1, #2)}}
\newcommand{\diffk}{\cod{diff}}
\newcommand{\vdiffk}{\cod{vdiff}}
\newcommand{\mdiffk}{\cod{mdiff}}
\newcommand{\diff}[1]{\text{\diffk{}~#1}}
\newcommand{\gradk}{\cod{grad}}
\newcommand{\mgradk}{\cod{mgrad}}
\newcommand{\grad}[1]{\text{\gradk{}~#1}}
\newcommand{\jacobk}{\cod{jacob}}
\newcommand{\jacob}[1]{\text{\jacobk{}~#1}}
\newcommand{\derivk}{{\color{mygreen}\cod{\textbf{deriv}}}}
\newcommand{\deriv}[2]{\text{\derivk{} #1 #2}}

\newcommand{\codespace}{\vspace{0.2cm}}

\newcommand{\forwardvar}[1]{\overrightarrow{#1}}
\newcommand{\reversevar}[1]{\overleftarrow{#1}}

\newcommand{\samerule}{\vspace{-0.105cm}}
\newcommand{\nextfigure}{\vspace{0.5cm}}

\newcommand{\argt}[3]{\text{${\color{mygreen}\mathcal{A}_{\text{#1}}\llbracket}$#2${\color{mygreen}\rrbracket}_{\small\text{#3}}$}}

\newcommand{\demo}{\hfill $\triangle$}

\newsavebox{\mybox}
\newenvironment{fscode}
  {
\vspace{0.1cm}
\noindent
\begin{lrbox}{\mybox}
\begin{minipage}[h]{0.975\columnwidth}
\small
\begin{xtabular}{l}
  }
  {
\end{xtabular}
\end{minipage}
\end{lrbox}\fbox{\usebox{\mybox}}
\vspace{0.1cm}
  }

\newenvironment{fscodegray}
  {
\vspace{0.1cm}
\noindent
\begin{lrbox}{\mybox}
\begin{minipage}[h]{0.975\columnwidth}
\small
\begin{xtabular}{l}
  }
  {
\end{xtabular}
\end{minipage}
\end{lrbox}\fcolorbox{black}{mygray}{\color{black}\usebox{\mybox}}
\vspace{0.1cm}
  }

\begin{document}

\title{Efficient and Sound Differentiable Programming in a Functional Array-Processing Language}
\shorttitle{Efficient and Sound Diff. Programming in a Functional Array-Proc. Language}

\newcommand*\samethanks[1][\value{footnote}]{\footnotemark[#1]}

 \author[Amir Shaikhha et al]
        {AMIR SHAIKHHA\\
         University of Edinburgh\\
         \and{}
         MATHIEU HUOT\\
         University of Oxford\\
         \and{}
         SHABNAM GHASEMIRAD\\
         ETH Zurich\\
         \and{}
         ANDREW FITZGIBBON\thanks{Work done while at Microsoft Research.} \\
         Graphcore\\
         \and{}
         SIMON PEYTON JONES\samethanks\\
         Epic Games\\
         \and{}
         DIMITRIOS VYTINIOTIS\samethanks\\
         DeepMind
         \email{amir.shaikhha@ed.ac.uk, mathieu.huot@stx.ox.ac.uk, shabnam.ghasemirad@inf.ethz.ch, awf@graphcore.ai, simon.peytonjones@gmail.com, dvytin@google.com}
         }











\maketitle
 
 \begin{abstract}
Automatic differentiation (AD) is a technique for computing the derivative of a function represented by a program. This technique is considered as the de-facto standard for computing the differentiation in many machine learning and optimisation software tools. Despite the practicality of this technique, the performance of the differentiated programs, especially for functional languages and in the presence of vectors, is suboptimal. 

We present an AD system for a higher-order functional array-processing language. The core functional language underlying this system simultaneously supports both source-to-source forward-mode AD and global optimisations such as loop transformations. In combination, gradient computation with forward-mode AD can be as efficient as reverse mode, and the Jacobian matrices required for numerical algorithms such as Gauss-Newton and Levenberg-Marquardt can be efficiently computed. 
\end{abstract}

\section{Introduction}
Functional programming (FP) and automatic differentiation (AD) have been natural partners for sixty years, and major functional languages all have elegant automatic differentiation packages~\cite{elliott2009beautiful,baydin2015automatic,karczmarczuk1999functional}.
With the increasing importance of numerical engineering disciplines such as machine learning, speech processing, and computer vision, there has never been a greater need for systems which mitigate the tedious and error-prone process of manual coding of derivatives. 

Popular machine learning packages such as Tensorflow~\cite{abadi2016tensorflow} and Pytorch~\cite{paszke2017automatic} greatly benefit from AD for optimisation tasks. These systems implement embedded domain-specific languages (EDSLs) with a predefined set of efficient combinators for manipulating tensors.
Furthermore, they can use compilation (e.g., the XLA~\cite{leary2017xla} backend for Tensorflow and the Glow~\cite{rotem2018glow} 
backend for PyTorch) in order to perform further optimisations. 
However, these systems are quite restrictive in what constructs are efficiently supported; additional tensor operations
are not as efficient as the predefined set of tensor operators.

The AD systems for generic-purpose programming languages are implemented using two modes.
Forward-mode AD is relatively straightforward, both as a runtime technique using dual numbers,
or as a source-to-source program transformation.  However, forward-mode AD is usually considered wildly inefficient as 
a way to compute the gradient of a function, because it involves calling the forward-mode AD function $n$ times --- 
and $n$ may be very large (e.g. $n = 10^6$).  

This has led to a tremendous amount of work on reverse-mode AD. 
As a source-to-source transformation, reverse-mode AD is characterised by the necessity to maintain temporary 
variables holding partial results, to be consumed during a ``reverse pass'' of gradient computation.  
Modern systems devote considerable effort (e.g., checkpointing) to optimizing the recompute/store tradeoff 
of these temporaries.

Our key contribution is this: we start from the ``wildly inefficient'' loop
in which the forward-mode function is called~$n$ times, and demonstrate that it can be made
efficient simply by applying a collection of non-AD-specific compile-time optimising transformations.
In fact, our optimisations are sufficiently aggressive to generate code that computes the gradient with efficiency 
that is sometimes \emph{better than} standard reverse-mode techniques. Moreover, the entire pipeline is much simpler: there is 
no need to grapple with the complexities of reverse mode --- it simply falls out from the result of optimisation.

More specifically, we take a recently introduced F\# subset designed for efficient compilation of array-processing workloads, and augment it with vector AD primitives, yielding a functional AD tool that is competitive with the best C/C++ tools on micro and real-world computer vision and machine learning benchmarks, and considerably faster than others.
Furthermore, we show the correctness of the differentiation and optimisation rules by introducing the evaluation semantics of this language.

Our contributions are as follows:


\begin{itemize}
\item We introduce \lafsharp{}, an ML-like functional language with array processing constructs (Section~\ref{sec:overview}), and \ladsl{}, an embedded linear-algebra-based language.
\item We describe forward-mode AD (Section~\ref{sec_bg})
as a simple source-to-source program transformation for \lafsharp{} (Section~\ref{sec_diff}).
\item We introduce a collection of optimising transformations, none of them AD-specific, for our language (Section~\ref{sec:fsmooth_trans}).
\item We demonstrate the correctness of both source-to-source AD transformations and optimising transformations (Section~\ref{sec_sem}).
\item For some small but realistic benchmarks, we show that our transformations suffice to generate extremely efficient code, 
starting from a na\"ive loop that repeatedly calls the forward derivative of the function.  We compare our results with 
those of other AD systems (Section~\ref{sec:exp}).
\end{itemize}

\codespace{}
\noindent We discuss related work in Section~\ref{sec_related}.

\section{Background}
\label{sec_bg}
Many applications in machine learning (such as training artificial neural networks, natural language processing, and computer vision) require computing the derivative of an objective function. 
There are four techniques for computing the derivative: 
1) manual derivation of analytical derivatives,
2) numeric differentiation,
3) symbolic differentiation, and
4) automatic differentiation.

Among these techniques, manual derivation is error prone and time consuming, whereas, numerical differentiation (using finite difference approximations), although easy to implement, can lead to precision errors and is really slow for functions with vector inputs. 
Symbolic differentiation involves transforming expressions into their derivative forms.
Symbolically differentiated expressions
can have an exponential overhead over the original program if one is not careful on sharing the computation. 
Furthermore, this method has no general solution for computing the derivative of programs that include iterations and control-flow constructs~\cite{baydin2015automatic,karczmarczuk1999functional}. 

Automatic Differentiation (AD)
systematically applies the chain rule, and evaluates the derivatives for the primitive arithmetic 
operations (such as addition, multiplication, etc.). 
One of the key properties of AD is the constant-time overhead of the differentiated program with 
respect to the original code; not being careful about sharing during the process of differentiation, 
can lead to code explosion~\cite{baydin2015automatic}.

There are two main modes for implementing AD. 
Forward-mode computes the derivative part (tangent part) alongside the original computation while making a forward 
pass over the program. Reverse-mode makes a forward pass to compute the original part of the program, followed by a backward pass for computing the derivative part (adjoint part).  Consider a program in ANF~\cite{anf}:

\begin{fscode}
$f\big(x_1, ..., x_n\big) = $\\
\tab \lett \,$v_1 = e_1$ \\
\tab ... \\
\tab \lett \,$v_n = e_n$ \\
\tab $v_n$
\end{fscode}

\codespace{}

\noindent To compute the derivative of this function using the forward-mode AD, we associate a \textit{tangent} variable to each variable $v_i$, denoted by $\forwardvar{v_i}$.  Each tangent variable is computed as
$\forwardvar{v_i} = \forwardvar{x_1} \times \frac{\partial v_i}{\partial x_1} + ... + \forwardvar{x_n} \times \frac{\partial v_i}{\partial x_n}$.
In order to compute the partial derivative of $f$ with respect to $x_i$, we assign $\forwardvar{x_i} = 1$ and $\forwardvar{x_j} = 0$ for $i \ne j$, and call the transformed code.

Reverse-mode AD computes the $n$ partial derivatives simultaneously, rather than in two different iterations.
To compute the derivative of this function using this approach, we associate an \textit{adjoint} variable to each variable $v_i$, denoted by $\reversevar{v_i}$, which is computed as $\reversevar{v_i} = \frac{\partial y}{\partial v_i}$. 
As a result, if we are interested in computing the partial derivative of function $f$ with respect to $x_1$, we have to compute the value of $\reversevar{x_1}$.
Similarly, if we are interested in the partial derivative of this function with respect to $x_2$, we have to compute the value of $\reversevar{x_2}$.  To do so, we have to apply the chain rule in the reverse order.


Generally speaking, forward and reverse-mode compute a column and a row, respectively, of the full Jacobian matrix $\mathbf{J}$ at each invocation.
More precisely, for a function with an input vector of size $n$ and an output vector of size $m$, the forward mode approach computes a column vector of size $m$, and the reverse mode computes a row vector of size $n$
(see Figure~\ref{jacobian_mat}).


\begin{figure}[t!]
        \centering
        \definecolor{orange}{rgb}{1,0.5,0}
  \begin{equation*}
  f : \mathbb{R}^n \rightarrow \mathbb{R}^m\,\,\, \hspace{0.5cm} \mathbf{J}_f = 
  \frac{\partial f}{\partial x} = \begin{tikzpicture}[baseline={(m.center)}]
            \matrix [matrix of math nodes,left delimiter={[},right delimiter={]}] (m)
            {
                \frac{\partial f_1}{\partial x_1} & \cdots & \frac{\partial f_1}{\partial x_n} \\
                \vdots & \ddots & \vdots \\
                \frac{\partial f_m}{\partial x_1} & \cdots & \frac{\partial f_m}{\partial x_n} \\
            };
            \draw[color=blue] (m-1-1.north west) rectangle (m-3-1.south east-|m-1-1.north east) node[pos=1.1,xshift=-0.6cm] {Forward Mode};
            \draw[color=red] (m-1-1.north west) rectangle (m-1-3.south east-|m-1-3.north east) node[pos=1.4,xshift=0.5cm,yshift=0.7cm] {Reverse Mode};
            \draw[color=violet] (-0.9,1) rectangle (-0.4,-1);
            \draw[color=violet] (-0.3,1) rectangle (-0.2,-1);
            \draw[color=violet] (-0.1,1) rectangle (0,-1);
            \draw[color=violet] (0.1,1) rectangle (0.2,-1);
            \draw[color=violet] (0.4,1) rectangle (0.9,-1) node[pos=1.1] {\system};
            \end{tikzpicture}
\end{equation*}
        \caption{The Jacobian Matrix of a Function. Forward-mode AD computes a column of this matrix, whereas the reverse-mode AD computes a row of this matrix. \system computes the full Jacobian matrix using a vectorized variant of the forward-mode AD.}
        \label{jacobian_mat}
\end{figure}






For a class of
optimisation problems, such as various computer vision problems using the Levenberg-Marquardt algorithm~\cite{marquardt1963algorithm,levenberg1944method,more1978levenberg},
one is required to compute the \emph{full} Jacobian matrix.
In such cases, neither of the two techniques perform efficiently.
To compute the full Jacobian matrix, both forward and reverse-mode techniques must iterate over either the columns or the rows of the Jacobian matrix, respectively.
Given that both approaches have a constant overhead over the original computation, the forward mode technique is more efficient for computing the full Jacobian matrix when $m \gg n$, whereas the reverse mode AD is more efficient when $n \gg m$.
However, when $n$ and $m$ are in the same range, it is not clear which approach performs better.  
Moreover:
\begin{itemize}
\item
By carefully examining the body of the loops needed for computing the full Jacobian matrix, one can observe that many computations are loop-invariant and are unnecessarily performed multiple times.
Thus, there is a lost opportunity for \textit{loop-invariant code motion} for hoisting such expressions outside the loop, thus asymptotically improving the performance (cf. the NNMF and Bundle Adjustment experiments in Section~\ref{sec:exp}).
\item
Furthermore, while the result of automatic differentiation is known to have only a constant factor more arithmetic operations than the original program, the constant can be significant; this overhead can have a dramatic impact on the run-time performance in practice. 
More specifically, in applications involving the manipulation of vectors, many intermediate vectors are allocated that can be removed. 
The optimisation for eliminating such intermediate vectors is known as \textit{deforestation}~\cite{deforestation,foldr-fusion-1,Svenningsson:2002:SFA:581478.581491,Coutts07streamfusion} or loop fusion in the functional programming community.
This optimisation opens the door for many other optimisations such as normalising loops that are iterating over sparse 
vectors with a single non-zero element into a single statement (cf. Example 5 in Section~\ref{sec:fsmooth_trans}).
\end{itemize}
  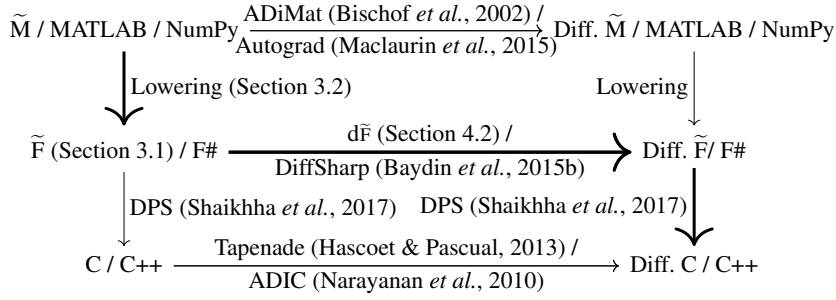
\begin{figure}[t!]
\begin{equation*}
\begin{tikzcd}[column sep=3.9cm, row sep=1.0cm]
\text{\ladsl{} / MATLAB / NumPy} \arrow[d,line width=0.4mm,"\text{Lowering (Section \ref{sec:ladsl})}"] \arrow[r,"\text{ADiMat~\cite{bischof2002combining} / }"] \arrow[r,swap,"\text{Autograd~\cite{maclaurin2015autograd}}"] & \text{Diff. \ladsl{} / MATLAB / NumPy} \arrow[d,swap,"\text{Lowering}"] \\
\text{\fsmooth~(Section~\ref{sec:fsmooth}) / F\#} \arrow[r,line width=0.4mm,"\text{\system{} (Section~\ref{sec:fsmooth_ad}) /}"]
\arrow[r,swap,line width=0.4mm,"\text{DiffSharp~\cite{baydin2015diffsharp}}"]  \arrow[d,"\text{DPS~\cite{dps_fhpc}}"]  & \text{Diff. \fsmooth / F\#} \arrow[d,line width=0.4mm,swap,"\text{DPS~\cite{dps_fhpc}}"]\\
\text{C / C++}  \arrow[r,"\text{Tapenade~\cite{tapenade} /}"]  \arrow[r,swap,"\text{ADIC~\cite{narayanan2010adic2}}"]  & \text{Diff. C / C++} \\
\end{tikzcd}
\end{equation*}
\vspace{-1.2cm}
\caption{Compilation process in different AD systems. The solid arrows correspond to the pipeline used in \system.}
\label{fig:comp_cd}
\end{figure}

\section{Overview}
\label{sec:overview}

In this section, we start with an overview of the compilation process in \system, which is shown in Figure~\ref{fig:comp_cd}. This figure demonstrates the position of \system with respect to existing AD tools.

\system starts from a program written in a high-level linear algebra library, called \ladsl{} (Section~\ref{sec:ladsl}).
This program is lowered into its implementation in a higher-order functional language with array support, called \fsmooth{} (Section~\ref{sec:fsmooth}).
If a part of the program requires computing differentiation (which are specified by using high-level differentiation API exposed by \system, as mentioned in Section~\ref{sec:diff_api})
\system uses AD transformation rules (Section~\ref{sec:fsmooth_ad}) for transforming the involved expressions into their differentiated form.

Finally, after applying several simplifications such as loop fusion, partial evaluation, data layout transformation, etc. (Section~\ref{sec:fsmooth_trans}) the differentiated program is transformed into low-level C code.
The generated C code uses efficient stack-discipline memory management by using the destination-passing style (DPS) technique~\cite{dps_fhpc}. 
Alternatively, one can use other array programming languages such as Futhark~\cite{henriksen2017futhark} and SAC~\cite{Grelck2006} as the target language for 
differentiated \fsmooth{} programs.

Next, we present the core functional language used in \system, on top of which we define source-to-source AD transformation and simplification rules.

\subsection{\fsmooth}
\label{sec:fsmooth}

\lafsharp{} 
(pronounced as F smooth) 
is a subset of F\#, an ML-like higher-order functional programming language.  It is designed to be \emph{expressive enough} to make it easy to write array-processing workloads, while simultaneously being \emph{restricted enough} 
(e.g., avoiding partially applied functions and returning functions from lambdas which are enforced by the 
(T-App) and (T-Abs) typing rules, respectively)
in order to compile it to code that is as efficient as hand-written C, with very simple and efficient memory management~\cite{dps_fhpc}.

\begin{figure}[tb]
  \parbox{\textwidth}{
\raggedright
\def\comment{ & -- }
\textbf{Syntax:} \\
\begin{tabular*}{0.965\columnwidth}{r c p{0.4\columnwidth} r}
\expr & ::= & \expr{} \vmore{\expr{}} | \vabs{\vmore{\text{x}}}{\expr} | \text{x} \comment Application, Abstraction, and Variable Access\\
& | & \text{r} | \text{i} | \vtrue{} | \vfalse{} \comment Scalar, Index, and Boolean Values\\
& | & \text{c} \comment Constants (see below)\\
& | & !let x = !e! in !e \comment(Non-Recursive) Let Binding\\
& | & !if !e! then !e! else !e \comment Conditional\\
& | & \varraygen{\expr, ..., \expr} \comment Array Construction\\
\typet{} & ::= & \typemat{} \comment (Non-Functional) Expression Type\\
& | & \typefunone{\vmore{\typet}}{\typemat} \comment Function Types (No Currying)\\
\typemat{} & ::= & \typenum{} \comment Numeric Type\\
& | & \typearray{\typemat{}} \comment Vector, Matrix, ... Type\\
& | & \typepair{\typemat}{\typemat} \comment Pair Type \\
& | & \typebool \comment Boolean Type\\
\typenum & ::= & \typedouble{} | \typeindex{} \comment Scalar and Index Type\\
\end{tabular*}
\\[8pt]
\textbf{Typing Rules:}\\
\fbox{
\begin{minipage}{0.965\columnwidth}
\centering
(T-App) $\infer{\exprind{0}: \typefun{\typet}{\typemat} \tab \tab \vmore{\expr} : \vmore{\typet}}{\exprind{0}\ \vmore{\expr}: \typemat}$
\tab
(T-Abs) $\infer{ \Gamma  \cup  \vmore{\text{x}}: \vmore{\typet} \vdash \expr{}: \typemat }{\Gamma \vdash \vabs{\vmore{\text{x}}}{\expr}: \typefun{\typet}{\typemat}}$
\tab
(T-Var) $\infer{\text{x}: \typet \in \Gamma }{\Gamma \vdash \text{x}: \typet}$
\\
(T-Let) $\infer{ \Gamma \vdash \exprind{1}: \typeind{1} \tab \tab \Gamma, \text{x}: \typeind{1} \vdash \exprind{2}: \typeind{2}}{\Gamma \vdash \text{\lett{} x = \exprind{1} \inn{} \exprind{2}: \typeind{2}}}$ 
\tab (T-If) $\infer{\exprind{1}: \typebool{} \tab \tab \exprind{2}: \typemat \tab \tab \exprind{3}: \typemat }{\vifthenelse{\exprind{1}}{\exprind{2}}{\exprind{3}}: \typemat}$ 
\end{minipage}
}\\[8pt]
\textbf{Scalar Function Constants:}\\
\fbox{
\begin{minipage}{0.533\columnwidth}
\begin{tabular}{l c l}
\vconst{+} | \vconst{-} | \vconst{*} | \vconst{/} | \vconst{**} &:& \typefunone{\typenum{} \funarrow{} \typenum}{\typenum}\\
\vconst{-} &:& \typefunone{\typenum}{\typenum}\\
\vconst{sin} | \vconst{cos} | \vconst{tan} | \vconst{log} | \vconst{exp} &:& \typefunone{\typedouble}{\typedouble}\\
\end{tabular}
\end{minipage}
}\fbox{
\begin{minipage}{0.40\columnwidth}
\begin{tabular}{l c l}
\vconst{>} | \vconst{<} | \vconst{==} | \vconst{<>} &:& \typefunone{\typenum{} \funarrow{} \typenum}{\typebool} \\
\vconst{\&\&} | \vconst{||} &:& \typefunone{\typebool{} \funarrow{} \typebool}{\typebool} \\
\vconst{not} &:& \typefunone{\typebool}{\typebool} \\
\end{tabular}
\end{minipage}
}\\[8pt]
\textbf{Vector Function Constants:}\\
\fbox{
\begin{minipage}{0.533\columnwidth}
\setlength\tabcolsep{1.5pt} 
\begin{tabular}{l c l}
\vbuildk{} &:& 
\typefunone{
\typeindex{}
 \funarrow{}
 (\typefunone{\typeindex}{\typemat{}})
 }{
 \typearray{\typemat}
}\\
\viteratek{} &:& 
\typefunone{%
(\typefunone{%
$\typemat$
\funarrow{}
\typeindex
}{$\typemat$})
\funarrow{}
$\typemat$
\funarrow{}
\typeindex
}{$\typemat$
}\\
\end{tabular}
\end{minipage}
}\fbox{
\begin{minipage}{0.40\columnwidth}
\setlength\tabcolsep{1.5pt} 
\begin{tabular}{l c l}
\vcget{} &:& 
\typefunone{%
\typearray{\typemat}%
\funarrow{}%
\typeindex%
}{%
\typemat%
}\\
\vclength{} &:& 
\typearray{\typemat}%
\funarrow{}%
\typeindex\\
\end{tabular}
\end{minipage}
}\\[8pt]
\textbf{Pair Function Constants:}\\
\fbox{
\begin{minipage}{0.32\columnwidth}
\vpairk{} : \typefunone{%
$\typemat{}_1$%
}{%
\typefunone{%
$\typemat{}_2$%
}{\typepair{$\typemat{}_1$}{$\typemat{}_2$}}
}
\end{minipage}
}\fbox{
\begin{minipage}{0.29\columnwidth}
\vfstk{} : \typefunone{%
\typepair{$\typemat{}_1$}{$\typemat{}_2$}%
}{%
$\typemat{}_1$%
}
\end{minipage}
}\fbox{
\begin{minipage}{0.293\columnwidth}
\vsndk{} : \typefunone{%
\typepair{$\typemat{}_1$}{$\typemat{}_2$}%
}{%
$\typemat{}_2$%
}
\end{minipage}
}
\\[8pt]
\textbf{Syntactic Sugar:} \\
\fbox{
\begin{minipage}{0.533\columnwidth}
\begin{tabular}{l c l}
\typematrix{} & = & \typearray{\typearray{\typedouble}} \\
\typedoubled{} & = & \typepair{\typedouble}{\typedouble} \\
\typevectord{} & = & \typearray{\typepair{\typedouble}{\typedouble}} \\
\typematrixd{} & = & \typearray{\typearray{\typepair{\typedouble}{\typedouble}}}
\end{tabular}
\end{minipage}
}\fbox{
\begin{minipage}{0.40\columnwidth}
\begin{tabular}{l c l}
\vget{\exprind{0}}{\exprind{1}} & = & \vcget{} \exprind{0} \exprind{1}
\\
\vpair{\exprind{0}}{\exprind{1}} & = & \vpairk{} \exprind{0} \exprind{1} \\
\exprind{1} $bop$ \exprind{2} & = & $bop$ \exprind{1} \exprind{2} \\
\typevector{} & = & \typearray{\typedouble}
\end{tabular}
\end{minipage}
}
\caption{The syntax, type system, and function constants of the core \fsmooth{}.}
\label{fig:fsmooth_core_syntax}}
\end{figure}

Figure~\ref{fig:fsmooth_core_syntax} shows the abstract syntax (parentheses can be used as necessary), type system, and several built-in functions of \fsmooth. 
\vmore{x} and \vmore{e} denote one or more variables and expressions, respectively, which are separated
by spaces, whereas, \vmore{T} represents one or more types which are separated by \funarrow{}.
In addition to the usual $\lambda$-calculus constructs (abstraction, application, and variable access), \fsmooth{} supports let binding and conditionals.

\fsmooth{} supports array programming by defining the following built-in functions: \vbuildk{} for producing arrays;  \viteratek{} for iteration for a particular number of times (from \cod{0} to \cod{n-1}) while maintaining a state across iterations; \vclength{} to get the size of an array; and \vcget{} to index an array. 

One of the key features of \fsmooth{} is its support for both source-to-source automatic differentiation and global optimisations such as loop-invariant code motion and loop fusion.
The transformations required for automatic differentiation are presented in Section~\ref{sec:fsmooth_ad}, and the ones for optimisation and simplification are shown in Section~\ref{sec:fsmooth_trans}.

Next, we show how a Linear Algebra library can be defined on top of \fsmooth{}.

\subsection{\ladsl{}}
\label{sec:ladsl}
\ladsl{} is a functional Linear Algebra library, mainly inspired by MATLAB and R, programming languages which are heavily used by data analysts.
By providing high-level vector and matrix operations, \ladsl{} frees the users from low-level details and enables them to focus on the algorithmic aspects of the problem in hand. 

\begin{table}[t!]
\setlength{\tabcolsep}{6pt}
\small
\centering
\caption{Equivalent operations in Matlab, R, NumPy, and \ladsl{}.}
\label{table_ops}
\vspace{-0.2cm}
\begin{tabular}{ c c c c }
\toprule
  {\bf Matlab} & {\bf R} & {\bf NumPy} & {\bf \ladsl{}}  \\ \toprule
  A * B & A \%*\% B & A.dot(B) & matrixMult A B \\ \midrule
  A + B & A + B & A + B & matrixAdd A B \\ \midrule
  A' & t(A) & A.T & matrixTranspose A \\ \midrule
  ones(n, m) & matrix(1, n, m) & ones((n, m)) & matrixOnes n m \\ \midrule
  zeros(n, m) & matrix(0, n, m) & zeros((n, m)) & matrixZeros n m \\ \midrule
  eye(n) & diag(n) & eye(n) & matrixEye n \\ \midrule
\end{tabular}
\end{table}

\ladsl{} is simply a \fsmooth{} library, but it can also be thought of
as an \textit{embedded domain-specific language} (EDSL)~\cite{hudak-dsl}.
Figure~\ref{fig:ladsl_ops} demonstrates a subset of \ladsl{} operations which are defined as functions in \fsmooth{}. 
This library is expressive enough for constructing vectors and matrices, element-wise operations, accessing a slice of elements, reduction-based operations (computing the sum of vector elements), matrix transpose, and matrix multiplication.\footnote{We have duplicated identical implementations for functions such as \code{vectorMap} and \code{matrixMap} (and similarly for \code{vectorMap2} and \code{matrixMap2}) because of the restrictions imposed by \fsmooth{}~\cite{dps_fhpc}. More specifically, the C code generation process does not handle polymorphic functions. Hence, we need to specify two different monomorphic functions with the same implementation for such functions.}
Supporting more sophisticated operations such as matrix determinant and matrix decomposition is beyond the scope of the current paper, and we leave it for the future.
As discussed before, \ladsl{} is inspired by MATLAB and R. 
As a result, there is a mapping among the constructs of \ladsl{} and these matrix-based languages. 
Hence, it is easily possible to translate a program written in one of these languages to \ladsl{}.
Table~\ref{table_ops} demonstrates the mapping among a subset of the constructs of MATLAB, R, NumPy and \ladsl{}.

\begin{figure*}[t]
\setlength{\tabcolsep}{3pt}
\aboverulesep=0ex
\belowrulesep=0ex
\newcommand{\nextfun}{\midrule}
\begin{minipage}{.475\textwidth}\raggedright
\begin{tabular}{|l|}
\toprule
\lett{} vectorRange = \vabs{n}{}
\\ \tabt 
\vbuild{n}{(\vabs{i}{i})}
\\ \nextfun
\lett{} vectorFill = \vabs{n e}{} 
\\ \tabt
\vbuild{n}{(\vabs{i}{e})}
\\ \nextfun
\lett{} vectorHot = \vabs{n i}{} 
\\ \tabt
\vbuild{n}{(\vabs{j}{\iif{} i = j \then{} 1 \elsee{} 0})}
\\ \nextfun
\lett{} vectorMap = \vabs{v f}{}
\\ \tabt
\vbuild{(\vlength{v})}{(\vabs{i}{f \vget{v}{i}})}
\\ \nextfun
\lett{} vectorMap2 = \vabs{v1 v2 f}{}
\\
\tabt 
\vbuild{(\vlength{v1})}{(\vabs{i}{f \vget{v1}{i} \vget{v2}{i}})}
\\ \nextfun
\lett{} vectorZip = \vabs{v1 v2}{}
\\
\tabt 
vectorMap2 v1 v2 (\vpairk{})
\\ \nextfun
\lett{} vectorAdd = \vabs{v1 v2}{}
\\
\tabt 
vectorMap2 v1 v2 (+)
\\ \nextfun
\lett{} vectorEMul = \vabs{v1 v2}{}
\\
\tabt 
vectorMap2 v1 v2 ($\times$)
\\ \nextfun
\lett{} vectorSMul = \vabs{v s}{}
\\
\tabt 
vectorMap v (\vabs{a}{a $\times$ s})
\\ \nextfun
\lett{} vectorSum = \vabs{v}{}\\
\tabt \viterate{(\vabs{s i}{s + \vget{v}{i}})}{0}{(\vlength{v})}
\\ \nextfun
\lett{} vectorDot = \vabs{v1 v2}{}\\
\tabt
vectorSum (vectorEMul v1 v2)
\\ \nextfun
\lett{} vectorNorm = \vabs{v}{}
\\
\tabt 
sqrt (vectorDot v v)
\\ \nextfun
\lett{} vectorSlice = \vabs{v s e}{} \\
\tabt \vbuild{(e \vcsubcard{} s \vcaddcard{} 1)}{(\vabs{i}{\vget{v}{i + s}})}
\\ \nextfun
\lett{} vectorToMatrix = \vabs{v}{}
\\
\tabt \vbuild{1}{(\vabs{i}{v})} 
\\ \nextfun
\lett{} vectorOutProd = \vabs{v1 v2}{}
\\
\tabt \lett{} m1 = vectorToMatrix v1 
\\
\tabt \lett{} m2 = vectorToMatrix v2 
\\
\tabt \lett{} m2T = matrixTranspose m2
\\
\tabt matrixMul m1 m2T
\\
\bottomrule
\end{tabular}
\end{minipage}
\begin{minipage}{.48\textwidth}\raggedright
\begin{tabular}{|l|}
\toprule
\lett{} matrixRows = \vabs{m}{} \vlength{m} 
\\ \nextfun
\lett{} matrixCols = \vabs{m}{} \vlength{(\vget{m}{0})} 
\\ \nextfun
\lett{} matrixZeros = \vabs{r c}{} 
\\ \tabt \vbuild{r}{(\vabs{i}{vectorFill c 0})}
\\ \nextfun
\lett{} matrixOnes = \vabs{r c}{} 
\\ \tabt \vbuild{r}{(\vabs{i}{vectorFill c 1})}
\\ \nextfun
\lett{} matrixEye = \vabs{n}{} 
\\ \tabt \vbuild{n}{(\vabs{i}{vectorHot n i})}
\\ \nextfun
\lett{} matrixHot = \vabs{n m r c}{} 
\\ \tabt
\vbuild{n}{(\vabs{i}{}} 
\\ \tabt \tabt
\vbuild{m}{(\vabs{j}{}} 
\\ \tabt \tabt \tabt \iif{} (i = r \cod{\&\&} j = c) \then{} 1 \elsee{} 0
\\
\tabt ) )
\\ \nextfun
\lett{} matrixMap = \vabs{m f}{}
\\ \tabt
\vbuild{(\vlength{m})}{(\vabs{i}{f \vget{m}{i}})}
\\ \nextfun
\lett{} matrixMap2 = \vabs{m1 m2 f}{}
\\
\tabt 
\vbuild{(\vlength{m1})}{}(\vabs{i}{}f\\
\tabt \tabt \vget{m1}{i} \vget{m2}{i})
\\ \nextfun
\lett{} matrixAdd = \vabs{m1 m2}{}
\\
\tabt
matrixMap2 m1 m2 vectorAdd
\\ \nextfun
\lett{} matrixTranspose = \vabs{m}{} 
\\
\tabt
\vbuild{(matrixCols m)}{}(\vabs{i}{}
\\
\tabt \tabt 
\vbuild{(matrixRows m)}{}(\vabs{j}{}
\\
\tabt \tabt \tabt 
\vget{\vget{m}{j}}{i}
\\
\tabt ) )
\\ \nextfun
\lett{} matrixMul = \vabs{m1 m2}{} 
\\
\tabt \lett{} m2T = matrixTranspose m2
\\
\tabt
\vbuild{(matrixRows m1)}{}(\vabs{i}{}
\\
\tabt \tabt \vbuild{(matrixCols m2)}{}(\vabs{j}{}
\\
\tabt \tabt \tabt vectorDot (\vget{m1}{i}) (\vget{m2T}{j})
\\
\tabt ) )
\\ \nextfun
\lett{} matrixTrace = \vabs{m}{}\\
\tabt{}\vifoldk{}\,(\vabs{s i}{}%
s+\vget{\vget{m}{i}}{i}) 0 (\vlengthk{}\,m)\\
\bottomrule
\end{tabular}
\end{minipage}\hfill
\caption{A subset of \ladsl{} constructs defined in \lafsharp{}.}
\label{fig:ladsl_ops}
\end{figure*}

\noindent \textbf{Example 1.} Assume that we have a matrix $M$ and two vectors $u$ and $v$ (which are represented as column matrices and are independent of $M$).
Based on matrix calculus one can prove that $\frac{\partial \big(uMv^T\big)}{\partial M}=u^Tv$. 
However, computing the differentiated version of this function using forward-mode AD tools requires multiple iterations over the differentiated program for every element in the matrix $M$. 
By using the reverse-mode AD, one can invoke the differentiated function only once, and the adjoint parts of the input matrix $M$ will be filled in.
We will show in Section~\ref{sec:fsmooth_trans} that \system{} derives the gradient of this expression with respect to $M$, resulting in an expression equivalent to $u^Tv$.
This optimises away multiple iterations over the differentiated program for each element of matrix $M$, in contrast to the existing AD tools based on the forward-mode AD technique.

For the moment, we only show how the matrix expression $uMv^T$ is expressed in \ladsl{}:

\begin{fscode}
\lett{} f = \vabs{u M v}{}\\
\tabt{}\lett{} um = vectorToMatrix u\\
\tabt{}\lett{} vt = matrixTranspose (vectorToMatrix v)\\
\tabt{}\lett{} m = 
matrixMult um (matrixMult M
vt)\\
\tabt{}\vget{\vget{m}{0}}{0}
\end{fscode}

\codespace{}
\noindent The last expression is for accessing the single scalar element of a $1\times 1$ matrix.

\demo

  \section{Differentiation}
\label{sec_diff}

In this section, we show the differentiation process in \system. First, we start by the high-level API exposed by \system to the end users.
Then, we show how \system uses automatic differentiation behind the scenes for computing derivatives.
Finally, we present the optimisations offered by \system, and we demonstrate how \system can use these optimisations to deduce several matrix calculus identities.

\subsection{High-Level API}
\label{sec:diff_api}
For computing the derivative of an arbitrary function, \system provides the \derivk{} construct. 
This construct can be better thought of as a macro, 
which is expanded during compilation time.
The expanded expression includes the expression of the original computation, which is given as the first argument (and can be an arbitrary scalar, vector, or matrix expression), and the derivative of this expression with respect to the variable given as the second argument, referred to as the \textit{independent variable}.
Note that one can easily compute the derivative of an expression with respect to a list of free variables by multiple 
invocation of the \derivk{} construct.

Figure~\ref{fig:deriv} shows the implementation of the \derivk{} construct, which is a compile-time
transformation, corresponding to computing the derivative of e with respect to x.
First, \derivk{} constructs a lambda function which has the free variables of the given expression as its input 
parameters.\footnote{\derivk{} only handles expressions of scalar, vector, or matrix type,
and cannot be used for function types. The same restriction applies for the extracted free variables, but 
there is no restriction on the type of other sub-expressions.}
The produced lambda function is given as input to source-to-source automatic differentiation 
(denoted by \difftrans{}), which can handle expressions of arbitrary type as explained later in Section~\ref{sec:fsmooth_ad}.
The differentiated function is applied to the dual number encoding of all the free variables (using the \argt{$i$}{}{}
compile-time transformation,
corresponding to the dual number encoding of a tensor expression of rank $i$). 
Based on the type of the independent variable, the result of this applied function will be the result scalar value,
or an element of the result vector or matrix.

If the free variable is different than the input variable with respect to which we are differentiating (i.e., the independent variable), the derivative part is a zero scalar, vector, or matrix. 
Otherwise, the derivative part is a one-hot encoding scalar, vector, or matrix.
If the independent variable has a scalar type, \derivk{} returns the applied function. 
However, if the independent variable has a vector type, \derivk{} constructs a vector with the same number of elements as the independent variable. 
For computing the $r^{\text{th}}$ element of the result vector, the corresponding input vector is a one-hot encoding with a single one at the $r^{\text{th}}$ position.
The situation is similar for an independent variable with a matrix type; the corresponding one-hot encoding matrix has a single one at the $r^{\text{th}}$ row and $c^{\text{th}}$ column.

\noindent \textbf{Example 2.} Let us assume that we would like to compute the derivative of a program computing the cosine function with respect to its input:

\begin{fscode}
\vconst{cos} a
\end{fscode}

\codespace{}
\noindent The derivative of this program at point $a$ is represented as follows:

\begin{fscode}
\dterm{\deriv{(\vconst{cos} a)}{a}}
\end{fscode}

\codespace{}
\noindent This expression is transformed into the following expression after expanding the \derivk{} macro:

\begin{fscode}
\dterm{(\difftrans{\vabs{a}{}\vconst{cos} a}) \adpair{a}{1}}
\end{fscode}

\demo

\begin{figure*}[t!]
\centering
\setlength{\tabcolsep}{2pt}
\fbox{
\begin{tabular}{l c l}
\derivk{} e (x: \typedouble{}) & = & (\difftrans{\vabs{v}{e}})
\argt{0}{v}{x}\\
\derivk{} e (x: \typevector{}) & = & \vbuildk{} (\vlength{x}) (\vabs{r}{} (\difftrans{\vabs{v}{e}}) 
\argt{1}{v}{x,r})\\
\derivk{} e (x: \typematrix{}) & = & \vbuildk{} (matrixRows x) (\vabs{r}{} \\
& & \tabt \vbuildk{} (matrixCols x) (\vabs{c}{} (\difftrans{\vabs{v}{e}}) 
\argt{2}{v}{x,r,c}))  \\
\\
\textbf{where}&&v = fvs(e)\\
\\
\argt{0}{ x: \typedouble{} }{x} & = & \adpair{x}{1}\\
\argt{0}{ v: \typedouble{} }{x} & = & \adpair{v}{0}\\
\argt{0}{ v: \vmore{\typedouble{}} }{x} & = & \argt{0}{\genexprind{v}{1}}{x} ... \argt{0}{\genexprind{v}{n}}{x} \quad \textbf{where} \quad v = \genexprind{v}{1} ... \genexprind{v}{n}\\
\samerule \\
\argt{1}{ x: \typevector{} }{x,r} & = & vectorZip x (vectorHot (\vlength{x}) r)\\
\argt{1}{ v: \typevector{} }{x,r} & = & vectorZip v (vectorZeros (\vlength{v}))\\ 
\argt{1}{ v: \vmore{\typevector{}} }{x,r} & = &  \argt{1}{\genexprind{v}{1}}{x,r} ... \argt{1}{\genexprind{v}{n}}{x,r} \quad \textbf{where} \quad v = \genexprind{v}{1} ... \genexprind{v}{n}\\ \samerule \\
\argt{2}{ x: \typematrix{} }{x,r,c} & = & matrixZip x (matrixHot (matrixRows x) (matrixCols x) r c)\\
\argt{2}{ v: \typematrix{} }{x,r,c} & = & matrixZip v (matrixZeros (matrixRows v) (matrixCols v))\\
\argt{2}{ v: \vmore{\typematrix{}} }{x,r,c} & = &  \argt{2}{\genexprind{v}{1}}{x,r,c} ... \argt{2}{\genexprind{v}{n}}{x,r,c} \quad \textbf{where} \quad v = \genexprind{v}{1} ... \genexprind{v}{n}\\
\end{tabular}
}
\caption{Implementation of the \derivk{} construct as a source-to-source transformation pass.}
\label{fig:deriv}
\end{figure*}

Furthermore, \system provides three additional differentiation constructs, inspired by AD tools such as DiffSharp~\cite{baydin2015diffsharp}: 1) \diffk{} computes the derivative of a function, from a real number to a real number, with respect to its input, 
2) \gradk{} computes the gradient of a function, from a vector of real numbers to a real number, with respect to its input vector,
and 3) \jacobk{} computes the Jacobian matrix of a vector-valued function, a function from a vector of real numbers to a vector of real numbers, with respect to its input vector.
Figure~\ref{fig:diff_trans_api} demonstrates how these high-level differentiation constructs are defined in terms of the source-to-source AD transformation construct $\mathcal{D}$.

\begin{figure}[t]
\setlength{\tabcolsep}{6pt}
\small
\centering
\begin{tabular}{l l l}
{\bf Operator} & {\bf Type} & {\bf Definition} \\ \toprule
\diff{} & (\typedouble \funarrow \typedouble) \funarrow &
\vabs{f x}{}\difftrans{f} \adpair{x}{1}\\
& \tab \typedouble \funarrow \typedoubled \\ \midrule
\grad{} & (\typevector \funarrow \typedouble)  & 
\vabs{f v}{}
\\ & \tab \funarrow \typevector \funarrow \typevectord &  
\tab \vbuildk{} (\vlengthk~v) (\vabs{i}{}\\ \cline{1-2}
\jacob{} & (\typevector \funarrow \typevector)  &
\tab \tab \difftrans{f} (vectorZip v (vectorHot (\vlengthk~v) i))
\\ & \tab \funarrow \typevector \funarrow \typematrixd &  
\tab )\\
\end{tabular}
\caption{High-Level Differentiation API for \lafsharp{}.}
\label{fig:diff_trans_api}
\end{figure}

\noindent \textbf{Example 2 (Continued).} 
For the previous example, if we would like to use the \diff{} construct, first we have to define the following function:

\begin{fscode}
g = \vabs{x}{}\vconst{cos}(x)
\end{fscode}

\codespace{}
\noindent The derivative of this function at point $a$ is represented as follows:

\begin{fscode}
\dterm{(\diff{g})~a}
\end{fscode}

\codespace{}
\noindent which is expanded to the following program:

\begin{fscode}
\dterm{\difftrans{g} \adpair{a}{1}}
\end{fscode}

\demo

\noindent Table~\ref{tbl:matrix_derivative} summarizes different matrix derivatives, and how they can be computed using our high-level API. Note that the definition of \vdiffk{} and \mdiffk{} is similar to \diffk, and the definition of \mgradk{} is similar to \gradk{} and \jacobk{} (cf. Figure~\ref{fig:diff_trans_api}). 
Note that the \derivk{} construct subsumes all these operators.   

\begin{table}
\setlength{\tabcolsep}{6pt}
\aboverulesep=0ex
\belowrulesep=0ex
\small
\caption{Different types of matrix derivatives.}
\label{tbl:matrix_derivative}
\begin{tabular}{c?c c c}
\backslashbox{Input Type}{Output Type} & {\bf Scalar}  & {\bf Vector} & {\bf Matrix} \\ \toprule
{\bf Scalar} & \diffk{} & \vdiffk{} & \mdiffk{}  \\ 
{\bf Vector} & \gradk{} & \jacobk{} & -- \\ 
{\bf Matrix} & \mgradk{} & -- & --
\end{tabular}
\end{table}

One key advantage of defining different matrix derivatives in terms of automatic differentiation is that one no longer needs to define the matrix calculus derivative rules for all different combinations shown in Table~\ref{tbl:matrix_derivative}.
Instead these rules can be deduced automatically from the automatic differentiation rules defined for scalar values. 
Moreover, even the algebraic identities for matrix derivative can be deduced by using the simplification rules presented in Section~\ref{sec:fsmooth_trans}.

Next, we present the source code transformation required for applying automatic differentiation rules.
\subsection{Source-to-Source Automatic Differentiation}
\label{sec:fsmooth_ad}

\system relies on source-to-source translation for implementing forward-mode automatic differentiation.
Each expression is converted into an expression containing both the original computation, together with the derivative computation, a.k.a. the dual number technique.
The scalar expressions are transformed into a pair of values, the original computation and the derivative computation.
The vector expressions are transformed into vectors containing tuple expressions, instead of scalar expressions.
The situation is similar for higher-rank tensors such as matrices.

The rules for automatic differentiation are demonstrated in Figure~\ref{fig:diff_trans}. 
\difftrans{e} specifies the AD translation for expression e. 
A variable y is translated to \diffvarprefix{y}, which, as opposed to the notation used in physics, keeps both the original and derivative part (D-Abs, D-Var, and D-Let).

The differentiation of an array of elements is an array of the differentiation of its elements (D-Vector).
Constructing an array is differentiated as an array with the same size, however, the way that each element of the array is constructed is differentiated (D-Build).
Differentiating an iteration results in an iteration with the same number of iterations, and with the initial state and the next state function both differentiated (D-IFold).
The differentiation of the length and indexing an array, is the same as the length and indexing the differentiated array, respectively (D-Length and D-Get).

For scalar expressions, we only consider the differentiation rules for real values (i.e.,
expressions of type \typedouble). For
expressions of type \typebool{} and \typeindex{} the derivative part is \false{} and 0, respectively (cf. Section~\ref{sec_sem}).
\difftransp{\expr} is a shorthand for extracting the original computation from the translated term \difftrans{\expr}, while \difftranst{\expr} is a shorthand for accessing the derivative part.
Note that in order to avoid redundant computation and code explosion, especially for the differentiation rules such as the product rule (D-Mult),
the arguments are bound to a new variable.

Differentiating a pair of elements results in the pair of differentiated elements (D-Pair). Similarly, differentiating the projection of a pair, is the projection of the differentiated pair (D-Fst, D-Snd).
For other scalar-valued functions, the differentiation rules are
similar to the corresponding rules in mathematics.

\noindent \textbf{Example 2 (Continued).} In the previous example, the differentiation transformation rules translate the program as follows:

\begin{fscode}
\diffvarprefix{g} = \vabs{\diffvarprefix{x}}{}-\dterm{\diffvarprefix{x}} * \vconst{sin}(\pterm{\diffvarprefix{x}})
\end{fscode}

\codespace{}
\noindent Based on the definition of the \diff{} construct, we have to use the AD version of the function (i.e., \diffvarprefix{g}) and assign 1 to the derivative part of the input. So the value of $\cos'$ for the input a is computed as follows:

\begin{fscode}
\dterm{(\diff{g})~a} \tab \evalsto \tab  
\dterm{\difftrans{g} \adpair{a}{1}} \tab \evalsto \tab  
\dterm{\diffvarprefix{g} \adpair{a}{1}} \tab \evalsto 
\\
\text{-\dterm{\adpair{a}{1}} * \vconst{sin}(\pterm{\adpair{a}{1}})} \tab \evalsto \tab 
\text{-1 * \vconst{sin}(a)} \tab \evalsto \tab  
\text{-\vconst{sin}(a)}
\end{fscode}

\demo

\noindent Similarly, we can compute the partial derivatives of a given function, by setting the desired derivative part to one, and the rest of derivatives to zero. This process is illustrated in the next example.

\noindent \textbf{Example 3.} Assume that we would like to compute the partial derivative of the expression a \vconst{*} b with respect to a, which is represented as follows in \fsmooth{}:

\begin{fscode}
\dterm{\deriv{(a \vconst{*} b)}{a}}
\end{fscode}

\codespace{}
\noindent This expression is expanded as follows:

\begin{fscode}
\dterm{\difftrans{\vabs{a b}{a \vconst{*} b}} \adpair{a}{1} \adpair{b}{0}}
\end{fscode}

\codespace{}
\noindent Note that the derivative part of the second input is set to 0. Similar to the previous example, the result is as follows:

\begin{fscode}
\dterm{(\vabs{\diffvarprefix{a}~\diffvarprefix{b}}{}
\adpair{\pterm{\diffvarprefix{a}}\vconst{*}\pterm{\diffvarprefix{b}}}{\pterm{\diffvarprefix{a}}\vconst{*}\dterm{\diffvarprefix{b}} +
\dterm{\diffvarprefix{a}}\vconst{*}\pterm{\diffvarprefix{b}}}) \adpair{a}{1} \adpair{b}{0}}
\end{fscode}

\codespace{}
\noindent which is evaluated as follows:

\begin{fscode}
\dterm{\adpair{a \vconst{*}  b}{1 \vconst{*}  b + a \vconst{*}  0}} \tab \evalsto \tab
1 \vconst{*} b + a \vconst{*} 0 \tab \evalsto \tab
b
\end{fscode}

\demo

\begin{figure*}[t!]
\centering
\setlength{\tabcolsep}{1pt}
\aboverulesep=0ex
\belowrulesep=0ex
\begin{tabular}{|l r c l|}
\toprule
\bf{$\lambda$-calculus Rules:}&&&\\\toprule
(D-App) & 
\difftrans{\vapp{\exprind{0}}{\exprind{1}}}&=&
(\difftrans{\exprind{0}})\ (\difftrans{\exprind{1}})\\ \midrule
(D-Abs) &
\difftrans{\vabs{\text{x}}{\expr{}}}&=&
\vabs{\diffvarprefix{\text{x}}}{ \difftrans{\expr{}} } \\ \midrule
(D-Var) & 
\difftrans{\text{y}}&=&
\diffvarprefix{\text{y}}\\ \midrule
(D-Let) &
\difftrans{\vlet{\text{x}}{\exprind{1}}{\exprind{2}}}&=&
\vlet{\diffvarprefix{\text{x}}}{\difftrans{\exprind{1}}}{} \difftrans{\exprind{2}}
\\ \midrule
(D-If) & 
\difftrans{\vifthenelse{\exprind{1}}{\exprind{2}}{\exprind{3}}}&=&
\vifthenelse{(\vfstk{} \difftrans{\exprind{1}})}{\difftrans{\exprind{2}}}{\difftrans{\exprind{3}}}\\
\toprule
\bf{Vector Rules:}&&&\\ \bottomrule
(D-Vector) & \difftrans{\varraygen{\exprind{0}, ..., \exprind{n}}} &=&
\varraygen{\difftrans{\exprind{0}}, ..., \difftrans{\exprind{n}}}\\\midrule
(D-Build) & \difftrans{\vbuild{\exprind{0}}{\exprind{1}}} &=&
\vbuild{(\vfstk{} \difftrans{\exprind{0}})}{(\vabs{i}{(\difftrans{\exprind{1}}) (i, 0))}}
\\ \midrule
(D-IFold) & \difftrans{\viterate{\exprind{0}}{\exprind{1}}{\exprind{2}}} &=&
\viteratek{} (\vabs{x i}{} 
\\
& & & \tabt 
(\difftrans{\exprind{0}}) x (i, 0)%
) \difftrans{\exprind{1}} (\vfstk{} \difftrans{\exprind{2}})\\ \midrule
(D-Get) & \difftrans{\vget{\exprind{0}}{\exprind{1}}} &=&
\vget{(\difftrans{\exprind{0}})}{\vfstk{} \difftrans{\exprind{1}}}
\\ \midrule
(D-Length) & \difftrans{\vlength{\exprind{0}}} &=&
(\vlength{\difftrans{\exprind{0}}}, 0)
\\
\toprule
\bf{Pair Rules:}&&&\\ \bottomrule
(D-Pair) & \difftrans{\vpair{\exprind{0}}{\exprind{1}}} &=&
\vpair{\difftrans{\exprind{0}}}{\difftrans{\exprind{1}}}
\\ \midrule
(D-Fst) & \difftrans{\vfstk{} \exprind{0}} &=&
\vfstk{} (\difftrans{\exprind{0}})
\\ \midrule
(D-Snd) & \difftrans{\vsndk{} \exprind{0}} &=&
\vsndk{} (\difftrans{\exprind{0}})
\\
\toprule
\bf{Scalar Rules:}&&&\\ \bottomrule
 & \difftrans{\expr} &=&
(\difftransp{\expr}, \difftranst{\expr})
\\ &&&
\\ \midrule
(D-Add) & \difftranst{\exprind{1} + \exprind{2}} &=&
\difftranst{\exprind{1}} + \difftranst{\exprind{2}}
\\ \midrule
(D-Mult) & \difftranst{\exprind{1} * \exprind{2}} &=&
\difftranst{\exprind{1}} * \difftransp{\exprind{2}} + \difftransp{\exprind{1}} * \difftranst{\exprind{2}}
\\ \midrule
(D-Div) & \difftranst{\exprind{1} / \exprind{2}} &=&
(\difftranst{\exprind{1}} * \difftransp{\exprind{2}} - \difftransp{\exprind{1}} * \difftranst{\exprind{2}}) / (\difftransp{\exprind{2}}**2)
\\ \midrule
(D-Neg) & \difftranst{-\exprind{1}} &=&
-\difftranst{\exprind{1}}
\\ \midrule
(D-Pow) & \difftranst{\exprind{1} ** \exprind{2}} &=&
(\difftransp{\exprind{2}} * \difftranst{\exprind{1}} / \difftransp{\exprind{1}} +  \\ 
& & & \tabt
\vconst{log} (\difftransp{\exprind{1}}) * \difftranst{\exprind{2}}) * (\difftransp{\exprind{1}} ** \difftransp{\exprind{2}})
\\ \midrule
(D-Sin) & \difftranst{\vconst{sin} \exprind{1}} &=&
\difftranst{\exprind{1}} * (\vconst{cos} \difftransp{\exprind{1}})
\\ \midrule
(D-Cos) & \difftranst{\vconst{cos} \exprind{1}} &=&
-\difftranst{\exprind{1}} * (\vconst{sin} \difftransp{\exprind{1}})
\\ \midrule
(D-Tan) & \difftranst{\vconst{tan} \exprind{1}} &=&
\difftranst{\exprind{1}} / ((\vconst{cos} \difftransp{\exprind{1}}) ** 2)
\\ \midrule
(D-Log) & \difftranst{\vconst{log} \exprind{1}} &=&
\difftranst{\exprind{1}} / \difftransp{\exprind{1}}
\\ \midrule
(D-Exp) & \difftranst{\vconst{exp} \exprind{1}} &=&
\difftranst{\exprind{1}} * (\vconst{exp} \difftransp{\exprind{1}})
\\
\toprule
\bf{Typing Rules:}&&&\\ \bottomrule
(DT-Fun) &
\difftranstype{\typefunone{\typeind{1}}{\typeind{2}}} &=&
\typefunone{\difftranstype{\typeind{1}}}{\difftranstype{\typeind{2}}} \\ \midrule
(DT-Num) &
\difftranstype{\typenum} &=&
\typenum{} $\times$ \typenum \\ \midrule
(DT-Bool) &
\difftranstype{\typebool} &=&
\typebool{} $\times$ \typebool \\ \midrule
(DT-Arr) &
\difftranstype{\typearray{\typemat{}}} &=&
\typearray{\difftranstype{\typemat{}}}
\\ \midrule
(DT-Pair) &
\difftranstype{\typepair{$\typemat{}_1$}{$\typemat{}_2$}} &=&
\typepair{\difftranstype{$\typemat{}_1$}}{\difftranstype{$\typemat{}_2$}}
\\ \bottomrule
\end{tabular}
\caption{Automatic Differentiation Rules for \lafsharp{} Expressions.}
\label{fig:diff_trans}
\end{figure*}


\noindent It is important to note that \system performs many of the evaluation steps shown for the previous examples during compilation time, i.e., performs partial evaluation. 
Section~\ref{sec:fsmooth_trans} gives more details on the optimisations and simplifications offered by \system.

\subsection{Perturbation Confusion and Nested Differentiation}
\label{sec:nesteddiff}

In several problems such as computing the Hessian matrix, one requires to compute the differentiation of a differentiated program.
In such cases, one should be careful dealing with tangent parts.
We demonstrate this problem in the next example.

\noindent
\textbf{Example 4.}
Here is the classical example showing the perturbation confusion problem:

\codespace{}

$$\frac{\partial}{\partial x}\big(x \frac{\partial x + y}{\partial y}\big)$$

\codespace{}

\noindent
This expression should be evaluated to $1$ at every point. 
However, an AD tool can mistakenly evaluate this expression to $2$.
This is because of confusing the tangent part (perturbation) of the free variable x, while computing the inner derivative.
This is known as the \textit{perturbation confusion} problem in the AD literature.
\demo

If one uses the differentiation API of Figure~\ref{fig:diff_trans_api}, the perturbation confusion problem
appears. In order to avoid this problem, the \derivk{} macro needs to be used. The macro expansion of the \derivk{} 
operator can be thought of as a 
preprocessing step that binds each of the perturbations to a different variable~\cite{siskind2005perturbation}. 
Then, it is the responsibility of the \fsmooth{} programming language implementation 
(e.g., using alpha renaming as mentioned in~\cite{siskind2008nesting}) to avoid the perturbation confusion 
problem. We demonstrate this fact on our running example.

\noindent
\textbf{Example 4 (Continued).}
The previous expression is implemented as follows in the \fsmooth{} language:

\codespace{}

\begin{fscode}
\vabs{x y}{} \\
\tab
\vsndk{} (\\
\tab \tab 
\derivk{} (x * (\vsndk{} (\\
\tab \tab \tab \deriv{(x + y)}{y}\\
\tab \tab
))) x
\\
\tab )
\end{fscode}

\codespace{}

\noindent
After expanding the inner \derivk{} macro, the following expression is derived:



\codespace{}

\begin{fscode}
\vabs{x y}{} \\
\tab
\vsndk{} (\\
\tab \tab 
\derivk{} (\\
\tab \tab \tab 
\vletn{$t_1$}{\vsndk{} (}\\
\tab \tab \tab \tab (\vabs{\diffvarprefix{x} \diffvarprefix{y}}{\adpair{\pterm{\diffvarprefix{x}} + \pterm{\diffvarprefix{y}}}{\dterm{\diffvarprefix{x}} + \dterm{\diffvarprefix{y}}}}) \adpair{x}{0} \adpair{y}{1}\\
\tab \tab \tab )\\
\tab \tab \tab
x * $t_1$) x
\\
\tab )
\end{fscode}
















\noindent
Note that the variable $t_1$ is created in order to avoid code explosion that can 
result from the product rule (cf. Section~\ref{sec:fsmooth_ad}). Expanding the outer \derivk{} macro results in the following expression:

\codespace{}

\begin{fscode}
\vabs{x y}{} \\
\tab
\vsndk{} (\\
\tab \tab 
(\vabs{\diffvarprefix{x} \diffvarprefix{y}}{} \\
\tab \tab \tab 
\vletn{\diffvarprefix{$t_1$}}{\vsndk{} (}\\
\tab \tab \tab \tab (\vabs{\ddiffvarprefix{x} \ddiffvarprefix{y}}{}\\
\tab \tab \tab \tab \tab
(
	\adpair{
		\pterm{\pterm{\ddiffvarprefix{x}}} + \pterm{\pterm{\ddiffvarprefix{y}}}
	}{
		\dterm{\pterm{\ddiffvarprefix{x}}} + \dterm{\pterm{\ddiffvarprefix{y}}}
	}
, \\
\tab \tab \tab \tab \tab \tab
	\adpair{
		\pterm{\dterm{\ddiffvarprefix{x}}} + \pterm{\dterm{\ddiffvarprefix{y}}}
	}{
		\dterm{\dterm{\ddiffvarprefix{x}}} + \dterm{\dterm{\ddiffvarprefix{y}}}
	}
)\\
\tab \tab \tab \tab
) \adpair{\diffvarprefix{x}}{\adpair{0}{0}} \adpair{\diffvarprefix{y}}{\adpair{1}{0}}\\
\tab \tab \tab)\\
\tab \tab \tab 
\adpair{\pterm{\diffvarprefix{x}} * \pterm{\diffvarprefix{$t_1$}}}{\dterm{\diffvarprefix{x}} * \pterm{\diffvarprefix{$t_1$}} + \pterm{\diffvarprefix{x}} * \dterm{\diffvarprefix{$t_1$}}}\\
\tab \tab) \adpair{x}{1} \adpair{y}{0}
\\
\tab )
\end{fscode}

\codespace{}

\noindent
Note that this macro expansion results in the inner perturbation variables \ddiffvarprefix{x} and
\ddiffvarprefix{y} which are different from the outer variables \diffvarprefix{x} and \diffvarprefix{y}.
This different naming results in avoiding the perturbation confusion problem.
Finally, partially evaluating the inner expression results in the following expression:





\codespace{}

\begin{fscode}
\vabs{x y}{} 
1
\end{fscode}

\demo

  \section{Efficient Differentiation}
\label{sec:fsmooth_trans}

In this section, we show how \system achieves efficient differentiable programming. 
First, we show several transformation rules applicable on \fsmooth{} expressions. 
We show how these transformation rules are used to derive matrix-algebraic identities, 
in the level of \fsmooth{} expressions.
Then, we show how we generate C code from \fsmooth{} expressions for more efficient memory management.

\subsection{Transformation Rules}
There are various algebraic identities that one can define for \fsmooth{}. 
Based on these identities, vector and matrix-based programs, as well as differentiated programs can be 
heavily optimised. 
Figure~\ref{fig:fsmooth_opts} shows a set of optimisations defined for \fsmooth{}.
Through examples, we show how these rewrite rules can discover vector and matrix-level 
algebraic equalities.  

\begin{figure}
\aboverulesep=0ex
\belowrulesep=0ex
\setlength{\tabcolsep}{2pt}
\newcommand{\rewriteif}[1]{\textit{#1}}
\begin{minipage}{0.53\columnwidth}
\centering
\begin{tabular}{|p{0.37\columnwidth} p{0.095\columnwidth} p{0.38\columnwidth}|}
\toprule
(\vabs{x}{} \exprind{0}) \exprind{1}
&\evalsto& 
\lett{} x = \exprind{1} \inn{} \exprind{0} \\ \midrule
\lett{} x = \exprind{0} \inn{} \exprind{1}
&\evalsto& 
\exprind{1}[x $\mapsto$ \exprind{0}] \\ \midrule
\lett{} x = \exprind{0} \inn{} \exprind{1}
&\evalsto& 
\exprind{1} \\
\multicolumn{3}{|c|}{\rewriteif{(x $\not\in$ fvs(\exprind{1}))}} \\ \midrule
\lett{} x = & & \lett{} y = \exprind{0} \inn{}\\
\tab \lett{} y = \exprind{0} \inn{} \exprind{1} &\evalsto& \lett{} x = \exprind{1} \\
\inn{} \exprind{2}
& & 
\inn{} \exprind{2} \\ \midrule
\lett{} x = \exprind{0} \inn{} & & \lett{} x = \exprind{0} \inn{} \\
\lett{} y = \exprind{0} \inn{} &\evalsto& \lett{} y = x \inn{} \\
\exprind{1} & & \exprind{1} \\ \midrule
\lett{} x = \exprind{0} \inn{} & & \lett{} y = \exprind{1} \inn{} \\
\lett{} y = \exprind{1} \inn{} &\evalsto& \lett{} x = \exprind{0} \inn{} \\
\exprind{2} & & \exprind{2} \\
\multicolumn{3}{|c|}{\rewriteif{(x $\not\in$ fvs(\exprind{1}))}} \\ \midrule
f(\lett{} x = \exprind{0} \inn{} \exprind{1})
&\evalsto& 
\lett{} x = \exprind{0} \inn{} f(\exprind{1}) \\
\bottomrule
\end{tabular}
\subcaption{$\lambda$-Calculus Rules}
\label{fig:fsmooth_opt_lambda}
\nextfigure
\begin{tabular}{|p{0.37\columnwidth} p{0.095\columnwidth} p{0.38\columnwidth}|}
\toprule
\expr{} \vconst{+} 0 = 0 \vconst{+} \expr{} 
&\evalsto& 
\expr{} \\ \midrule
\expr{} \vconst{*} 1 = 1 \vconst{*} \expr{} 
&\evalsto& 
\expr{} \\ \midrule
\expr{} \vconst{*} 0 = 0 \vconst{*} \expr{} 
&\evalsto& 
0 \\ \midrule
\expr{} \vconst{+} -\expr{} = \expr{} \vconst{-} \expr{} 
&\evalsto& 
0 \\ \midrule
\exprind{0} \vconst{*} \exprind{1} \vconst{+}
\exprind{0} \vconst{*} \exprind{2}
&\evalsto& 
\exprind{0} \vconst{*} (\exprind{1} \vconst{+} \exprind{2}) \\ \bottomrule
\end{tabular}
\subcaption{Ring-Structure Rules}
\label{fig:fsmooth_opt_ring}
\end{minipage}
\begin{minipage}{0.45\columnwidth}
\centering
\begin{tabular}{|p{0.49\columnwidth} p{0.095\columnwidth} p{0.31\columnwidth}|}
\toprule
\vget{(\vbuild{\exprind{0}}{\exprind{1}})}{\exprind{2}} 
&\evalsto& 
\exprind{1}\  \exprind{2} \\
\multicolumn{3}{|c|}{\rewriteif{(if \exprind{0} $>$ \exprind{2})}} 
\\ \midrule
\vlength{(\vbuild{\exprind{0}}{\exprind{1}})} 
&\evalsto&
\exprind{0} \\ \bottomrule
\end{tabular}
\subcaption{Loop Fusion Rules}
\label{fig:fsmooth_opt_fusion}
\centering
\begin{tabular}{|p{0.30\columnwidth} p{0.095\columnwidth} p{0.50\columnwidth}|}
\toprule
\iif{} \vtrue{} &&\\
\tab \then{} \exprind{1} &\evalsto&  
\exprind{1}
\\
\tab \elsee{} \exprind{2} && \\ \midrule
\iif{} \vfalse{}&&\\
\tab \then{} \exprind{1}&\evalsto&
\exprind{2}\\
\tab \elsee{} \exprind{2} &&
\\ \midrule
\iif{} \exprind{0} &&\\
\tab \then{} \exprind{1}&\evalsto&  
\exprind{1}\\
\tab  \elsee{} \exprind{1} &&
\\ \midrule
\iif{} \exprind{0} && \iif{} \exprind{0} \\
\tab \then{} \exprind{1} 
&\evalsto&  \tab \then{} \exprind{1}[\exprind{0} $\mapsto$ \vtrue{}] \\
\tab \elsee{} \exprind{2} & & \tab \elsee{} \exprind{2}[\exprind{0} $\mapsto$ \vfalse{}]
\\ \midrule
f (\iif{} \exprind{0} && \iif{} \exprind{0}\\
\tab \then{} \exprind{1}&\evalsto& \tab \then{} f (\exprind{1}) \\
\tab \elsee{} \exprind{2}) & & \tab \elsee{} f (\exprind{2})\\
\bottomrule
\end{tabular}
\subcaption{Conditional Rules}
\label{fig:fsmooth_opt_if}
\begin{tabular}{|p{0.30\columnwidth} p{0.095\columnwidth} p{0.50\columnwidth}|}
\toprule
\vfstk{} \adpair{\exprind{0}}{\exprind{1}}
&\evalsto& 
\exprind{0} 
\\ \midrule
\vsndk{} \adpair{\exprind{0}}{\exprind{1}}
&\evalsto& 
\exprind{1} \\ \bottomrule
\end{tabular}
\subcaption{Tuple Normalisation Rules}
\label{fig:fsmooth_opt_tuple}
\end{minipage}
\begin{minipage}{\columnwidth}
\centering
\begin{tabular}{|p{0.50\columnwidth} p{0.035\columnwidth} p{0.38\columnwidth}|}
\toprule
\vifoldk{} f z 0
&\evalsto& 
z \\ \midrule
\vifoldk{} f z (n + 1)
&\evalsto& 
\vifoldk{} (\vabs{a i}{} f a (i+1)) (f z 0) n \\ \midrule
\vifoldk{} (\vabs{a i}{} a) z n
&\evalsto& 
z \\ \midrule
\vifoldk{} (\vabs{a i}{} \iif(i = \exprind{0}) \then{} \exprind{1} \elsee{} a) z n 
&\evalsto& 
\lett{} a = z \inn{} \lett{} i = \exprind{0} in \exprind{1} \\
\multicolumn{3}{|c|}{\rewriteif{(if 0 $\leq$ \exprind{0} $<$ n and \exprind{0} does not mention a or i)}} \\ \bottomrule
\end{tabular}
\subcaption{Loop Normalisation Rules}
\label{fig:fsmooth_opt_iteration}
\centering
\begin{tabular}{|p{0.41\columnwidth} p{0.095\columnwidth} p{0.41\columnwidth}|}
\toprule
\vifoldk{} (\vabs{a i}{}&&
\\
\tab 
\adpair{\genexprind{f}{0} (\vfstk{} a) i}{\genexprind{f}{1} (\vsndk{} a) i}
&\evalsto& 
(\vifoldk{} \genexprind{f}{0} \genexprind{z}{0} n,\vifoldk{} \genexprind{f}{1} \genexprind{z}{1} n)
\\
) \adpair{\genexprind{z}{0}}{\genexprind{z}{1}} n & &
\\
\bottomrule
\end{tabular}
\vspace{-0.15cm}
\subcaption{Loop Fission Rule}
\label{fig:fsmooth_opt_fission}
\end{minipage}
\vspace{-0.2cm}
\caption{Transformation Rules for \fsmooth{}. Even though none of these rules are 
AD-specific, the rules of Figure~\ref{fig:fsmooth_opt_tuple} and Figure~\ref{fig:fsmooth_opt_fission} are more useful in the AD context.}
\label{fig:fsmooth_opts}
\end{figure}

There are various optimisations defined for scalar operations based on the ring structure of addition and multiplication, which are shown in Figure~\ref{fig:fsmooth_opt_ring}.
Note that other ring-based algebraic identities, such as associativity and commutativity, do not appear directly in the list of rules that \system applies. 
This is because they do not necessarily improve the performance, unless they are combined with other rewrite rules. 

As \fsmooth{} is based on $\lambda$-calculus, all partial evaluation rules for this calculus come for free.
Furthermore, the optimisations defined in the literature for let-binding can also be used.
Figure~\ref{fig:fsmooth_opt_lambda} shows this set of rules. 
More specifically, the last rule subsumes the loop-invariant code motion rule, which can
have a massive impact on performance. Consider the following C program for constructing an array:

\begin{lstlisting}
for(int i=0; i<n; i++) {
  double x = g(y);
  res[i] = x * vec[i];
}
\end{lstlisting}

\noindent By carefully examining this program, one realises that the computation of \code{x} is invariant to the loop, and can be hoisted outside:

\begin{lstlisting}
double x = g(y);
for(int i=0; i<n; i++) {
  res[i] = x * vec[i];
}
\end{lstlisting}

\noindent The corresponding \fsmooth{} program before optimisation is as follows:

\begin{fscode}
\vbuild{n}{}(\vabs{i}{}\\
\tab \tab \lett{} x = g y \inn{}\\
\tab \tab x * \vget{vec}{i}\\
)
\end{fscode}

\noindent Applying the last rule of Figure~\ref{fig:fsmooth_opt_lambda} results in the following program:

\begin{fscode}
\lett{} x = g y \inn{} \\
\vbuild{n}{}(\vabs{i}{} x * \vget{vec}{i})
\end{fscode}

Another source of performance issue for vector programs is the unnecessary intermediate vectors. Consider the following program which represents adding three vectors:

\begin{lstlisting}
for(int i=0; i<n; i++) {
  tmp[i] = vec1[i] + vec2[i];
}
for(int i=0; i<n; i++) {
  res[i] = tmp[i] + vec3[i];
}
\end{lstlisting}

Applying loop fusion results in removing the unnecessary vector \code{tmp}, and having a single loop instead of two:

\begin{lstlisting}
for(int i=0; i<n; i++) {
  res[i] = vec1[i] + vec2[i] + vec3[i];
}
\end{lstlisting}

As the vector constructs of \fsmooth{} are based on pull arrays, one can use the pull-array fusion rules for removing 
unnecessary intermediate vectors and matrices. The two fusion rules for pull-arrays are shown in Figure~\ref{fig:fsmooth_opt_fusion}.
Apart from fusing a pipeline of vector/matrix operations, this optimisation can also be used to derive several matrix-algebra identities.

\noindent \textbf{Example 5.} It is known that for a matrix $M$, the following equality holds ${(M^T)}^T=M$.
We show how we can derive the same equality in \system{}. 
In other words, we show that:

\codespace{}
\begin{fscode}
matrixTranspose (matrixTranspose M) = M
\end{fscode}

\codespace{}
\noindent After let binding the inner expression, and inlining the definition of matrixTranpose and the functions
inside it, the following program is produced:

\codespace{}
\begin{fscode}
\lett{} MT = 
\\
\tabt
\vbuild{(\vlength{\vget{M}{0}})}{}(\vabs{i}{}
\\
\tabt \tabt 
\vbuild{(\vlength{M})}{}(\vabs{j}{}
\\
\tabt \tabt \tabt 
\vget{\vget{M}{j}}{i} ) ) in
\\
\vbuild{(\vlength{\vget{MT}{0}})}{}(\vabs{i}{}
\\
\tabt 
\vbuild{(\vlength{MT})}{}(\vabs{j}{}
\\
\tabt \tabt
\vget{\vget{MT}{j}}{i} ) )
\end{fscode}

\noindent Now, by applying the loop fusion rules (cf. Figure~\ref{fig:fsmooth_opt_fusion}) and performing further partial evaluation, the following
expression is derived:

\codespace{}
\begin{fscode}
\vbuild{(\vlength{M})}{}(\vabs{i}{}
\\
\tabt 
\vbuild{(\vlength{\vget{M}{0}})}{}(\vabs{j}{}
\\
\tabt \tabt
\vget{\vget{M}{i}}{j} ) )
\end{fscode}

\codespace{}
\noindent This is the same expression as M.

\demo

Figure~\ref{fig:fsmooth_opt_if} shows the rewrite rules for conditional expressions. 
The first two rules partially evaluate a conditional expression when its condition is statically known.
The next rule removes a conditional expression when both the branches are the same expression.
The fourth rewrite rule propagates the result of evaluating the condition into both branches.
Finally, the last rewrite rule pushes a function applied to a conditional expression into both branches.
This results in duplicating that function, which can lead to explosion in the size of expressions.

Figure~\ref{fig:fsmooth_opt_iteration} corresponds to normalisation rules for the \vifoldk{} construct.
The first two rewrite rules are quite well-known; they unfold a loop the size of which is statically known.
The last two rewrite rules are more interesting and can result in asymptotic performance improvements.
The third rule turns a loop that does not modify its state into a single statement corresponding to its initial state.
The last rule turns a loop that modifies its state only in one of its iterations into a single statement.
These two rules are especially useful in the context of dealing with sparse vectors and matrices 
(e.g., one-hot encoding vectors which can be the result of gradient computations, as well as identity matrices) 
as we can see in the next example.

\noindent \textbf{Example 6.} It is known that for a vector $v$, the following equality holds: $v \times I=v$,
where $I$ is an identity matrix of the same dimension as $v$.
We show how \system{} derives the same algebraic identity.
More specifically, we show that:

\codespace{}
\begin{fscode}
\lett{} I = matrixEye (\vlength{v}) \inn{}\\
\vbuild{(\vlength{v})}{}(\vabs{i}{}\\
\tabt \vifoldk{} (\vabs{a j}{}\\
\tabt \tabt a + \vget{v}{j} * \vget{\vget{I}{j}}{i} ) 0 (\vlength{v})
\end{fscode}

\codespace{}
\noindent is equivalent to v. Inlining and fusing this expression results in:

\codespace{}
\begin{fscode}
\vbuild{(\vlength{v})}{}(\vabs{i}{}\\
\tabt \vifoldk{} (\vabs{a j}{}\\
\tabt \tabt a + \vget{v}{j} * (\iif{} (i=j) \then{} 1 \elsee{} 0) ) 0 (\vlength{v})
\end{fscode}

\codespace{}
\noindent Applying conditional rules (cf. Figure~\ref{fig:fsmooth_opt_if}) and ring-structure rules (cf. 
Figure~\ref{fig:fsmooth_opt_ring}) results in:

\codespace{}
\begin{fscode}
\vbuild{(\vlength{v})}{}(\vabs{i}{}\\
\tabt \vifoldk{} (\vabs{a j}{}\\
\tabt \tabt \iif{} (i=j) \then{} a + \vget{v}{j} \elsee{} a) ) 0 (\vlength{v})
\end{fscode}

\codespace{}
\noindent After applying the loop normalisation rules (cf. Figure~\ref{fig:fsmooth_opt_iteration}),
the following expression is derived:

\codespace{}
\begin{fscode}
\vbuild{(\vlength{v})}{}(\vabs{i}{}\lett{} a = 0 \inn{} \lett{} j = i \inn{} a + \vget{v}{j})
\end{fscode}

\codespace{}
\noindent Finally, performing partial evaluation and simplifications results in the following expression:

\codespace{}
\begin{fscode}
\vbuild{(\vlength{v})}{}(\vabs{i}{}\vget{v}{i})
\end{fscode}

\codespace{}
\noindent This is the same expression as v.

\demo

Let us focus on transformations which are more related to differentiation (but not specific to them).
Many intermediate tuples resulting from the dual number technique of 
AD can be removed by using partial evaluation. Figure~\ref{fig:fsmooth_opt_tuple} shows the partial evaluation 
rules for removing the intermediate tuples which are followed by a projection.

Partially evaluating the tuples across the boundary of a loop  requires a sophisticated analysis of the body of the loop.
To simplify this task, we perform loop fission for the loops that return a tuple of values. 
This is possible only when different elements of the tuple are computed independently in different iterations of the loop.

Consider the following C program which computes two aggregates:

\begin{lstlisting}
for(int i=0; i<n; i++) {
  res1 = res1 + vec[i];
  res2 = res2 * vec[i];
}
\end{lstlisting}

\noindent As these two aggregates are independent of each other, loop fission can separate their loops as follows:

\begin{lstlisting}
for(int i=0; i<n; i++) {
  res1 = res1 + vec[i];
}
for(int i=0; i<n; i++) {
  res2 = res2 * vec[i];
}
\end{lstlisting}

\noindent Figure~\ref{fig:fsmooth_opt_fission} shows how loop fission turns an iteration creating a pair of elements into a pair of two iterations constructing independently the elements of that pair.
After performing this optimisation, if we are interested only in a particular element of the result tuple, other loops corresponding to irrelevant elements are removed by partial evaluation.

Based on these rewrite rules, \system derives well-known matrix calculus rules, without requiring to add a rewrite rule in the level of matrices (i.e., \ladsl{}). 
In the next section, we suggest a strategy for the ordering of the presented rewrite rules which supports all the examples and benchmarks reported in this article.

\subsection{Transformation Strategy}
Figure~\ref{fig:trans_pipeline} shows the transformation pipeline of \system.
The phases shown in the first two rows focus on the differentiated subexpression (specified
by the highlighted gray area), starting by 
expanding the \derivk{} macro, followed by applying source-to-source AD rules on 
the expanded expression, and then applying optimisations on the differentiated subexpression.
In the first optimisation, the body of functions are inlined in order to open opportunities for loop fusion (Figure~\ref{fig:fsmooth_opt_fusion}) and further partial evaluation (Figure~\ref{fig:fsmooth_opt_lambda}).
Then, we apply loop fission (Figure~\ref{fig:fsmooth_opt_fission}) and tuple partial evaluation (Figure~\ref{fig:fsmooth_opt_tuple}).
Afterwards, we perform various simplification rules consisting of algebraic identities available in the 
scalar arithmetic level (e.g., associativity, commutativity, and distributivity as shown in Figure~\ref{fig:fsmooth_opt_ring}),
as well as $\lambda$-calculus and its extended constructs rewrite rules (as shown in Figures~\ref{fig:fsmooth_opt_lambda},~\ref{fig:fsmooth_opt_tuple}, 
and~\ref{fig:fsmooth_opt_if}).
As last optimisation for the differentiated subexpression, we perform induction-based loop normalisation rules of \code{ifold} (Figure~\ref{fig:fsmooth_opt_iteration}) in order to turn loops that do not change the state in to a single statement.

Afterwards, we focus on the whole program and again apply loop fusion and partial evaluation
in order to fuse the differentiated subexpression with the rest of the expressions.
Then, we again apply the same sequence of optimisations as before, followed by turning 
the program into ANF and applying common-subexpression elimination (CSE).
The ANF makes it easier to perform loop-invariant code motion (the last rule of 
Figure~\ref{fig:fsmooth_opt_lambda}) and dead-code elimination (DCE, which is the third rule of the same figure).
Finally, after performing simplification rules, we rely on the restrictions
imposed by the \fsmooth{} language, such as only allowing a limited form of recursion 
(using \vifoldk{} and \vbuildk{}, thus benefiting from their associated optimisations),
and produce efficient low-level C code using DPS~\cite{dps_fhpc}.

This transformation pipeline works well in practice as we will see in Section~\ref{sec:exp}. 
However, we do not have any guarantee that for all programs these sets of rewrite 
rules with the given sequence are sufficient to make forward-mode AD as efficient as reverse-mode AD.
Also, we leave the use of alternative transformation pipelines and using search strategies for automated rewriting (e.g., using Monte-Carlo tree search~\cite{de2009bandit}) as 
future work.

The next example shows how \system can derive a well-known algebraic identity for the derivative of matrices using the presented pipeline.

\begin{figure}[t]
\tikzstyle{startstop} = [rectangle, rounded corners, minimum width=1.5cm, minimum height=1cm,text centered, draw=black, fill=white!30]

\tikzstyle{arrow} = [thick,->,>=stealth]
\begin{tikzpicture}[node distance=2.2cm]

\tikzstyle{mybox} = [rectangle, minimum width=1.5cm, minimum height=1.1cm,text centered, draw=black, fill=white!30, scale=0.9]
\tikzstyle{mybox2} = [rectangle, minimum width=1.2cm, minimum height=0.3cm,text centered, draw=black, fill=mygray, scale=0.9]
\tikzstyle{opt} = [rectangle, rounded corners, minimum width=1.5cm, minimum height=0.9cm,text centered, draw=black, fill=white!30, scale=0.9]

\node (input) [mybox, align=center] {\footnotesize{\ladsl{} Expr.}};
\node (macroexp) [opt, right of=input, align=center] {Macro\\ Expansion};
\node (sifaq) [mybox, right of=macroexp, align=center] {\ladsl{} Expr.\\\vspace{0.3cm}};
\node (sifaq2) [mybox2, right of=macroexp, align=center,yshift=-0.2cm] {\difftrans{\expr}};
\node (ccode) [opt, right of=sifaq, align=center,yshift=-0.2cm] {Differentiation};
\node (df1) [mybox, right of=ccode, align=center,yshift=0.2cm] {\ladsl{} Expr.\\\vspace{0.3cm}};
\node (df2) [mybox2, right of=ccode, align=center] {\tiny{Diff. Subexpr.}};
\node (opt3) [opt, right of=df1, align=center,yshift=-0.2cm] {Loop Fusion \&\\Partial Evaluation};
\node (df3) [mybox, below of=input, align=center,yshift=0.8cm] {\ladsl{} Expr.\\\vspace{0.3cm}};
\node (df4) [mybox2, below of=input, align=center,yshift=0.6cm] {\tiny{Diff. Subexpr.}};
\node (opt3_1) [opt, right of=df3, align=center,yshift=-0.2cm] {Loop Fission \&\\Tuple PE};
\node (df3_1) [mybox, right of=opt3_1, align=center,yshift=0.2cm] {\ladsl{} Expr.\\\vspace{0.3cm}};
\node (df4_1) [mybox2, right of=opt3_1, align=center,yshift=0cm] {\tiny{Diff. Subexpr.}};
\node (opt3_2) [opt, right of=df3_1, align=center,yshift=-0.2cm] {Simplification};
\node (df3_2) [mybox, right of=opt3_2, align=center,yshift=0.2cm] {\ladsl{} Expr.\\\vspace{0.3cm}};
\node (df4_2) [mybox2, right of=opt3_2, align=center,yshift=0cm] {\tiny{Diff. Subexpr.}};
\node (opt3_3) [opt, right of=df3_2, align=center,yshift=-0.2cm] {Loop \\ Normalisation};
\node (df3_n) [mybox, below of=df4, align=center,yshift=1cm] {\ladsl{} Expr.\\\vspace{0.3cm}};
\node (df4_n) [mybox2, below of=df4, align=center,yshift=0.8cm] {\tiny{Diff. Subexpr.}};
\node (opt4) [opt, right of=df3_n, align=center] {Loop Fusion \& \\Partial Evaluation};
\node (df5) [mybox, right of=opt4, align=center] {Diff.\\\fsmooth{} Expr.};
\node (opt5) [opt, right of=df5, align=center] {Loop Fission \& \\Tuple PE};
\node (df6) [mybox, right of=opt5, align=center] {Diff.\\\fsmooth{} Expr.};
\node (opt6) [opt, right of=df6, align=center] {Simplification};
\node (df7) [mybox, below of=df3_n, align=center,yshift=0.8cm] {Diff.\\\fsmooth{} Expr.};
\node (opt7) [opt, right of=df7, align=center] {Loop\\ Normalisation};
\node (df8) [mybox, right of=opt7, align=center] {Diff.\\\fsmooth{} Expr.};
\node (opt8) [opt, right of=df8, align=center] {ANF \& \\  CSE};
\node (df9) [mybox, right of=opt8, align=center] {Diff.\\\fsmooth{} Expr.};
\node (opt9) [opt, right of=df9, align=center] {Code Motion\\\& DCE};
\node (df10) [mybox, below of=df7, align=center,yshift=0.8cm] {Diff.\\\fsmooth{} Expr.};
\node (opt10) [opt, right of=df10, align=center] {Simplification};
\node (df11) [mybox, right of=opt10, align=center] {Diff.\\\fsmooth{} Expr.};
\node (opt11) [opt, right of=df11, align=center] {DPS};
\node (df12) [mybox, right of=opt11, align=center] {Diff.\\C Expr.};

\draw [arrow] (input) -- node[align=center, anchor=center] {} (macroexp);
\draw [arrow] (macroexp) -- node[align=center, anchor=east] {} (sifaq);
\draw [arrow] (sifaq2) -- node[align=center, anchor=east] {} (ccode);
\draw [arrow] (ccode) -- node[align=center, anchor=east] {} (df2);
\draw [arrow] (df2) -- node[align=center, anchor=east] {} (opt3);
\draw [arrow] (opt3) -- ++(1.4,0) -- ++(0,-0.47) -- ++(-12.2,0) -- ++(0,-0.8) -- node[align=center, anchor=east] {} (df4);
\draw [arrow] (df4) -- node[align=center, anchor=east] {} (opt3_1);
\draw [arrow] (opt3_1) -- node[align=center, anchor=east] {} (df4_1);
\draw [arrow] (df4_1) -- node[align=center, anchor=east] {} (opt3_2);
\draw [arrow] (opt3_2) -- node[align=center, anchor=east] {} (df4_2);
\draw [arrow] (df4_2) -- node[align=center, anchor=east] {} (opt3_3);
\draw [arrow] (opt3_3) -- ++(1.4,0) -- ++(0,-0.47) -- ++(-12.2,0) -- ++(0,-0.8) -- node[align=center, anchor=east] {} (df4_n);
\draw [arrow] (df3_n) -- node[align=center, anchor=east] {} (opt4);
\draw [arrow] (opt4) -- node[align=center, anchor=east] {} (df5);
\draw [arrow] (df5) -- node[align=center, anchor=east] {} (opt5);
\draw [arrow] (opt5) -- node[align=center, anchor=east] {} (df6);
\draw [arrow] (df6) -- node[align=center, anchor=east] {} (opt6);
\draw [arrow] (opt6) -- ++(1.4,0) -- ++(0,-0.67) -- ++(-12.2,0) -- ++(0,-0.6) -- node[align=center, anchor=east] {} (df7);
\draw [arrow] (df7) -- node[align=center, anchor=east] {} (opt7);
\draw [arrow] (opt7) -- node[align=center, anchor=east] {} (df8);
\draw [arrow] (df8) -- node[align=center, anchor=east] {} (opt8);
\draw [arrow] (opt8) -- node[align=center, anchor=east] {} (df9);
\draw [arrow] (df9) -- node[align=center, anchor=east] {} (opt9);
\draw [arrow] (opt9) -- ++(1.4,0) -- ++(0,-0.67) -- ++(-12.2,0) -- ++(0,-0.6) -- node[align=center, anchor=east] {} (df10);
\draw [arrow] (df10) -- node[align=center, anchor=east] {} (opt10);
\draw [arrow] (opt10) -- node[align=center, anchor=east] {} (df11);
\draw [arrow] (df11) -- node[align=center, anchor=east] {} (opt11);
\draw [arrow] (opt11) -- node[align=center, anchor=east] {} (df12);



\end{tikzpicture}
\caption{Transformation pipeline of \system.}
\label{fig:trans_pipeline}
\end{figure}
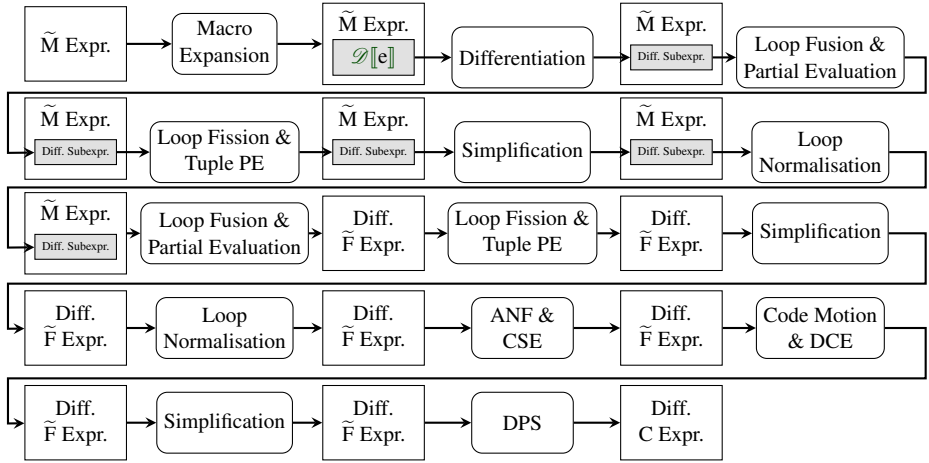

\noindent \textbf{Example 7.} Based on matrix calculus derivative rules, it is known that $\frac{\partial\big(v_1\cdot v_2\big)}{\partial v_1}=v_2$, where $\cdot$ is the vector dot product operator. 
We would like to show how \system{} can deduce the same algebraic identity. 
In other words, we show that the following \fsmooth{} program:

\codespace{}
\begin{fscode}
vectorMap (\derivk{} (vectorDot v1 v2) v1) \vsndk{})
\end{fscode}

\noindent is the equivalent to v2.

\noindent \textit{Macro Expansion.}
After expanding the \derivk{} macro, \system produces the following program:

\codespace{}
\begin{fscode}
vectorMap (\\
\tabt{}\vbuildk{}\,(\vlengthk{}\,v1) (\vabs{i}{}\\
\tabt{}\tabt{}\fcolorbox{black}{mygray}{\difftrans{\vabs{v1 v2}{}vectorDot v1 v2}} \\
\tabt{}\tabt{}\tabt{}(vectorZip v1 (vectorHot (\vlengthk{}\,v1) i)) \\
\tabt{}\tabt{}\tabt{}(vectorZip v2 (vectorZeros (\vlengthk{}\,v2))))\\
) \vsndk{}
\end{fscode}

\noindent \textit{Differentiation.} We first focus on the highlighted expression. After applying AD transformation rules (cf. Figure~\ref{fig:diff_trans}), 
the following expression is derived:

\begin{fscodegray}
\vabs{\diffvarprefix{v1} \diffvarprefix{v2}}{}\diffvarprefix{vectorDot} \diffvarprefix{v1} \diffvarprefix{v2}
\end{fscodegray}

\noindent \textit{Loop Fusion and Partial Evaluation.} Inlining the definition of \diffvarprefix{vectorDot} (which is derived by applying the AD transformation rules 
over the \lafsharp{} definitions given in Figure~\ref{fig:ladsl_ops}) results in:

\begin{fscodegray}
\vabs{\diffvarprefix{v1} \diffvarprefix{v2}}{}\\
\tabt \lett{} v = \vbuild{(\vfstk{} \adpair{\vlength{\diffvarprefix{v1}}}{0})}{} (\vabs{i}{}\\
\tabt\tabt \lett{} idx = \adpair{i}{0}\\
\tabt\tabt \lett{} x1 = \vget{\diffvarprefix{v1}}{\vfstk{} idx}\\
\tabt\tabt \lett{} x2 = \vget{\diffvarprefix{v2}}{\vfstk{} idx}\\
\tabt\tabt \lett{} x3 = \vfstk{} x1\\
\tabt\tabt \lett{} x4 = \vsndk{} x1\\
\tabt\tabt \lett{} x5 = \vfstk{} x2\\
\tabt\tabt \lett{} x6 = \vsndk{} x2\\
\tabt\tabt \adpair{x3*x5}{x4*x5+x3*x6}\\
\tabt)\\
\tabt \lett{} z = \adpair{0}{0}\\
\tabt \lett{} n = \adpair{\vlength{\diffvarprefix{v1}}}{0}\\
\tabt \viteratek{} (\vabs{s j}{} \\
\tabt\tabt \lett{} idx = \adpair{j}{0}\\
\tabt\tabt \lett{} x1 = \vget{v}{\vfstk{} idx}\\
\tabt\tabt \lett{} x2 = \vfstk{} s\\
\tabt\tabt \lett{} x3 = \vsndk{} s\\
\tabt\tabt \lett{} x4 = \vfstk{} x1\\
\tabt\tabt \lett{} x5 = \vsndk{} x1\\
\tabt\tabt(x2+x4, x3+x5)
\\
\tabt)
z (\vfstk{} n)
\end{fscodegray}

\noindent Applying the fusion and partial evaluation rules
(cf. Figure~\ref{fig:fsmooth_opts}) results in:

\begin{fscodegray}
\vabs{\diffvarprefix{v1} \diffvarprefix{v2}}{}\\
\tabt \viteratek{} (\vabs{s j}{} \\
\tabt\tabt(
(\vfstk{} s) \vconst{+} (\vfstk{} \vget{\diffvarprefix{v1}}{j}) \vconst{*} (\vfstk{} \vget{\diffvarprefix{v2}}{j})
,\\
\tabt\tabt\tabt
(\vsndk{} s) \vconst{+} 
(\vsndk{} \vget{\diffvarprefix{v1}}{j}) \vconst{*} (\vfstk{} \vget{\diffvarprefix{v2}}{j}) \vconst{+}
(\vfstk{} \vget{\diffvarprefix{v1}}{j}) \vconst{*} (\vsndk{} \vget{\diffvarprefix{v2}}{j})
)
\\
\tabt)
\adpair{0}{0} (\vlength{\diffvarprefix{v1}})
\end{fscodegray}

\noindent Injecting the transformed expression into the initial program results in:

\codespace{}
\begin{fscode}
vectorMap (\\
\tabt{}\vbuildk{}\,(\vlengthk{}\,v1) (\vabs{i}{}\\
\tabt{}\tabt{}\fcolorbox{black}{mygray}{\begin{tabular}{l}(\vabs{\diffvarprefix{v1} \diffvarprefix{v2}}{}\\
\tabt \viteratek{} (\vabs{s j}{} \\
\tabt\tabt(
(\vfstk{} s) \vconst{+} (\vfstk{} \vget{\diffvarprefix{v1}}{j}) \vconst{*} (\vfstk{} \vget{\diffvarprefix{v2}}{j})
,\\
\tabt\tabt\tabt
(\vsndk{} s) \vconst{+} 
(\vsndk{} \vget{\diffvarprefix{v1}}{j}) \vconst{*} (\vfstk{} \vget{\diffvarprefix{v2}}{j}) \vconst{+}
(\vfstk{} \vget{\diffvarprefix{v1}}{j}) \vconst{*} (\vsndk{} \vget{\diffvarprefix{v2}}{j})
)
\\
\tabt)
\adpair{0}{0} (\vlength{\diffvarprefix{v1}}))\end{tabular}} \\
\tabt{}\tabt{}\tabt{}(vectorZip v1 (vectorHot (\vlengthk{}\,v1) i)) \\
\tabt{}\tabt{}\tabt{}(vectorZip v2 (vectorZeros (\vlengthk{}\,v2))))\\
) \vsndk{}
\end{fscode}

\codespace{}
\noindent Inlining the definition of vectorMap, vectorZip, vectorHot, and vectorZeros, and applying the loop fusion, $\beta$-reduction, and partial evaluation rules produces the following expression:

\begin{fscode}
\vbuildk{}\,(\vlengthk{}\,v1) (\vabs{i}{}\\
\tabt \lett{} s = \viteratek{} (\vabs{s j}{} \\
\tabt\tabt\tabt\tabt\tabt
\lett{} x1 = \adpair{\vget{v1}{j}}{\iif{} (i=j) \then{} 1 \elsee{} 0}\\
\tabt\tabt\tabt\tabt\tabt
\lett{} x2 = \adpair{\vget{v2}{j}}{0}\\
\tabt\tabt\tabt\tabt\tabt
\lett{} x3 = \vfstk{} x1\\
\tabt\tabt\tabt\tabt\tabt
\lett{} x4 = \vfstk{} x2\\
\tabt\tabt\tabt\tabt\tabt
(
(\vfstk{} s) \vconst{+} x3 \vconst{*} x4
,\\
\tabt\tabt\tabt\tabt\tabt\tabt
(\vsndk{} s) \vconst{+} 
\vsndk{} x1 \vconst{*} x4 \vconst{+}
x3 \vconst{*} \vsndk{} x2
)
\\
\tabt\tabt\tabt\tabt{})
\adpair{0}{0} (\vlength{v1})\\
\tabt \vsndk{} s) \\
\end{fscode}

\noindent After applying tuple partial evaluation (cf. Figure~\ref{fig:fsmooth_opt_tuple}) the following program is generated:

\begin{fscode}
\vbuildk{}\,(\vlengthk{}\,v1) (\vabs{i}{}\\
\tabt{}\vsndk{} (\viteratek{} (\vabs{s j}{} \\
\tabt\tabt\tabt\tabt\tabt
(
(\vfstk{} s) \vconst{+} \vget{v1}{j} \vconst{*} \vget{v2}{j}
,\\
\tabt\tabt\tabt\tabt\tabt\tabt
(\vsndk{} s) \vconst{+} 
\vget{v1}{j} \vconst{*} 0 \vconst{+}
(\iif{} (i=j) \then{} 1 \elsee{} 0) \vconst{*} \vget{v2}{j}
)
\\
\tabt\tabt\tabt\tabt{})
\adpair{0}{0} (\vlength{v1})) \\
\end{fscode}

\noindent \textit{Loop Fission and Tuple PE.} Now we apply loop fission (cf. Figure~\ref{fig:fsmooth_opt_fission}):

\begin{fscode}
\vbuildk{}\,(\vlengthk{}\,v1) (\vabs{i}{}\\
\tabt{}\vsndk{} (\\
\tabt \tabt \viteratek{} (\vabs{s j}{} s \vconst{+} \vget{v1}{j} \vconst{*} \vget{v2}{j}) 0 (\vlength{v1}),\\
\tabt\tabt  \viteratek{} (\vabs{s j}{} s \vconst{+} 
\vget{v1}{j} \vconst{*} 0 \vconst{+}
(\iif{} (i=j) \then{} 1 \elsee{} 0) \vconst{*} \vget{v2}{j}) 0 (\vlength{v1})
\\
\tabt)\\
)
\end{fscode}

\noindent Note that applying the loop fission rule, does not necessarily improve the performance; 
it is only after performing tuple partial evaluation rules that the iteration responsible for the original computation is removed and the performance is improved. 
Thus, the strategy for applying rewrite rules can become tricky.
For this particular rewrite rule, we only apply it, when subsequently we can use partial evaluation to further simplify the program. 
To do so, we define a compound rewrite rule that either applies these rules together, or does not do anything.
This has a similar effect to the fold-fusion law, which can be found in the FP literature~\cite{fissiongibbons,hutton_1999}. 
After applying the tuple partial evaluation rule, the following program is derived:

\begin{fscode}
\vbuildk{}\,(\vlengthk{}\,v1) (\vabs{i}{}\\
\tabt \viteratek{} (\vabs{s j}{} s \vconst{+} 
\vget{v1}{j} \vconst{*} 0 \vconst{+}
(\iif{} (i=j) \then{} 1 \elsee{} 0) \vconst{*} \vget{v2}{j}) 0 (\vlength{v1})
\\
)
\end{fscode}

\noindent \textit{Simplification.} Applying conditional rules (cf. Figure~\ref{fig:fsmooth_opt_if}) and ring-based arithmetic simplification rules  (cf. Figure~\ref{fig:fsmooth_opt_ring}) result in:

\codespace{}
\begin{fscode}
\vbuildk{}\,(\vlengthk{}\,v1) (\vabs{i}{}\\
\tabt{}(\vifoldk{}\,(\vabs{s j}{}\\
\tabt{}\tabt{}\iif (i = j) \then\\
\tabt{}\tabt{}\tabt{}(s + \vget{v2}{j})\\
\tabt{}\tabt{}\elsee\\
\tabt{}\tabt{}\tabt{}s) 0 (\vlengthk{}\,v1)))
\end{fscode}

\noindent \textit{Loop Normalisation.} By using the optimisation that turns single access iterations into a single statement (cf. Figure~\ref{fig:fsmooth_opt_iteration}), \system produces the following program:

\begin{fscode}
\vbuildk{}\,(\vlengthk{}\,v1) (\vabs{i}{}
\vget{v2}{i})
\end{fscode}

\codespace{}
\noindent This program is equivalent to $v2$ if the size of the two input vectors are the same (i.e., \vlengthk{} $v1$ = \vlengthk{} $v2$). 
Otherwise, the input program is ill-formed.

\demo

Based on the same set of transformation rules, \system derives other matrix-calculus identities for the gradient
of matrices such as $\frac{\partial tr\big(M\big)}{\partial M}=I$, which states that the derivative of the
trace of a matrix with respect to that matrix, is an identity matrix. More generally,
\system can automatically discover the following algebraic identity if $A$ is independent of $M$:  
$\frac{\partial tr\big(MA\big)}{\partial M}=A^T$.

Now we return to the example shown in the beginning of this paper.

\codespace{}
\noindent \textbf{Example 1 (Continued).} If we have a matrix $M$ and two vectors $u$ and $v$ (which are represented as row matrices and are independent of $M$), using matrix calculus one can prove that $\frac{\partial \big(uMv^T\big)}{\partial M}=u^Tv$. 
First, we start by a partially inlined representation of this program in \fsmooth:

\begin{fscode}
\vabs{u M v}{}\\
\tabt{} vectorMap (\derivk{} (\\
\tabt{}\tabt{}\tabt{}\lett{} m =\\
\tabt{}\tabt{}\tabt \tabt matrixMult \\
\tabt{}\tabt{}\tabt \tabt \tabt (\vbuildk{}\,1 (\vabs{i}{} u)) \\
\tabt{}\tabt{}\tabt \tabt \tabt (matrixMult M \\
\tabt{}\tabt{}\tabt \tabt \tabt \tabt (matrixTranspose (\vbuildk{}\,1 (\vabs{i}{} v))))\\
\tabt{}\tabt ) M) \vsndk{}
\end{fscode}

This expression is expanded as follows:

\codespace{}
\begin{fscode}
\lett{} f = \vabs{u M v}{}\\
\tabt{}\lett{} m =\\
\tabt \tabt matrixMult \\
\tabt \tabt \tabt (\vbuildk{}\,1 (\vabs{i}{} u)) \\
\tabt \tabt \tabt (matrixMult M \\
\tabt \tabt \tabt \tabt (matrixTranspose (\vbuildk{}\,1 (\vabs{i}{} v))))\\
\tabt{}\vget{\vget{m}{0}}{0}\\
\vabs{u M v}{}\\
\tabt{}(\vbuildk{}\,(\vlengthk{}\,M) (\vabs{i}{}\\
\tabt{}\tabt{}(\vbuildk{}\,(\vlengthk{}\,\vget{M}{0}) (\vabs{j}{}\\
\tabt{}\tabt{}\tabt{}(\vsndk{}\,(\difftrans{f}\\
\tabt{}\tabt{}\tabt{}\tabt{}(vectorZip u (vectorZeros (\vlengthk{}\,u)))\\
\tabt{}\tabt{}\tabt{}\tabt{}(matrixZip M (matrixHot (\vlengthk{}\,M) (\vlengthk{}\,\vget{M}{0}) i j))\\
\tabt{}\tabt{}\tabt{}\tabt{}(vectorZip v (vectorZeros (\vlengthk{}\,v)))))))))
\end{fscode}

\codespace{}
\noindent Note that the function f is returning the only scalar element of the 1-by-1 matrix $uMv^T$.
After performing loop fusion, loop fission and partial evaluation the following program is derived:

\codespace{}
\begin{fscode}
\vabs{u M v}{}\\
\tabt{}\vbuildk{}\,(\vlengthk{}\,M) (\vabs{i}{}\\
\tabt{}\tabt{}\vbuildk{}\,(\vlengthk{}\,\vget{M}{0}) (\vabs{j}{}\\
\tabt{}\tabt{}\tabt{}\vget{u}{i} \vconst{*} \vget{v}{j}))
\end{fscode}

\codespace{}
\noindent This program is equivalent to $u^Tv$ if the input program is well formed, i.e., the number of rows and columns of $M$ are the same as the length of $u$ and $v$, respectively.

\demo

\subsection{Code Generation}
After applying the optimisations mentioned in the previous section,
one can further improve the efficiency by generating programs in a low-level language with manual memory management. 
This way, the overhead of garbage collection can be removed. 
Furthermore, by using stack-discipline memory management techniques such as Destination-Passing Style (DPS)~\cite{dps_fhpc}, one can benefit from efficient bump memory allocation instead of using the expensive \texttt{malloc} and \texttt{free} calls.

\codespace{}
\noindent 
\textbf{Example 1 (Continued).} The generated C code for the optimised differentiated program is as follows:

\begin{lstlisting}
matrix uMv_d(storage s, vector u, matrix M, vector v) {
  matrix res = (matrix)s;
  for(int r = 0; r < M->rows; r++) {
    for(int c = 0; c < M->cols; c++) {
      res->elems[r][c] = u->elems[r] * v->elems[c];
    }
  }
  return res;
}
\end{lstlisting}

\noindent The parameter \texttt{s} is the storage area allocated for storing the result matrix.

\demo

\noindent Up to now, we have only seen the cases where only the derivative part of the program was of interest. 
If we are interested in the original part of the program as well (e.g., the intermediate vectors cannot be fused), we need to store both the original and derivative parts.
In such cases, the differentiated vectors, which are represented as arrays of tuples, can be transformed into a more efficient data layout.
The well-known array of structs (AoS) to struct of arrays (SoA) transformation represents differentiated vectors as a tuple of two numeric arrays. Further partial evaluation can remove the unnecessary decoupled numeric arrays.

  \section{Semantics}
\label{sec_sem}

\subsection{Operational Semantics}

We give a standard call-by-value operational semantics to the language from Figure \ref{fig:fsmooth_core_syntax}. For the purpose of this section, we turn unary functions like $\vconst{sin}:\typedouble{}\funarrow{}\typedouble{}$ into typing rules, e.g.
\[(\vconst{sin}) \infer{ \Gamma  \vdash \expr{}: \typedouble}{\Gamma \vdash \vconst{sin}(\expr{}): \typedouble}\]

Similarly, for a binary function like $\vconst{*}:\typedouble{}\funarrow{}\typedouble{}\funarrow{}\typedouble{}$, the associated typing rule is:
\[(\vconst{*}) \infer{ \Gamma  \vdash \exprind{1}: \typedouble \tab \tab \Gamma  \vdash \exprind{2}: \typedouble }{\Gamma \vdash \exprind{1}~\vconst{*}~\exprind{2}: \typedouble}\] 

More generally, we denote by $\op_1$ any unary operator and by $\op_2$ any binary one. For functions with multiple arguments, the evaluation order is from left to right. We introduce the term $\bot$ at every type for divergence.

Values are given as follows:
\[\val:=  \text{r} ~|~ \text{i} ~|~ \vtrue{} ~|~ \vfalse{} ~|~\vpair{\valind{0}}{\valind{1}} ~|~  \vabs{\vmore{\text{x}}}{\expr}  ~|~ \varray{\val} \]

To deal with constants like \vconst{tan} which are not defined everywhere, let $\Dom$ be the domain function which assigns to a constant function its standard domain of definition. For instance, $\Dom(\vconst{tan}):=\{\val\in\RR~|~ \vconst{cos}(\val)\neq 0\}$. No construct is defined on $\bot$, hence if any argument reduces to $\bot$, the construct will also reduce to $\bot$. 
Constants of ground type will be underlined in the operational semantics for clarity. For example, the rule \[\inferrule{\val\in \Dom(\op_1)
  }{
   \op_1(\underline{\val}) \to \underline{\op_1(\val)}
  }\]
when $\op_1=\vconst{tan}$ and $\val=7$, leads to 
\[\inferrule{7\in \Dom(\vconst{tan})
  }{
   \vconst{tan}(\underline{7}) \to \underline{b}
  }\] 
  where $b$ is the real $tan(7)$. 

The reductions rules for the call-by-value small-step operational semantics are given in Figure \ref{fig:CBV}. 

\begin{figure}[tb]
\fbox{
  \parbox{0.96\textwidth}{
\input{reductions}
\caption{Small-step operational semantics}\label{fig:CBV}}}
\end{figure}

Given that the operational semantics is standard, we immediately get the following:

\begin{lemma}[Subject reduction]
\label{lem:SR1}
	If $\Gamma\vdash e:A$ and $e\to e'$ then $\Gamma\vdash e':A$.
\end{lemma}

\begin{lemma}
	Every closed term of ground type $\vdash \expr{}:\typemat{} $ reduces to a value $\val{}$ or diverges (i.e. reduces to $\bot$).
\end{lemma}

\subsection{Correctness of optimisation rules} 
\label{sub:correctness_of_optimisation_rules}

We prove correct the optimisations of Figure \ref{fig:fsmooth_opts}.
As is usual, we denote by $\to^*$ the transitive reflexive closure of $\to$.

\begin{theorem}
	All the rules $\leadsto$ from figure \ref{fig:fsmooth_opts} are sound. 
	That is, if $\vdash e:A$ and $e\leadsto e'$, then $e\to^*v$ iff $e'\to^*v$.
\end{theorem}

\begin{proof}
	By case analysis on the rule. This is straightforward for most rules so we only present the interesting cases. $z,z_0,z_1$ will reduce to values first in all the following cases, so without loss of generality we assume that they are values. 

\underline{Rule (c).1}: \vget{(\vbuild{\exprind{0}}{\exprind{1}})}{\exprind{2}} \evalsto \exprind{1} \exprind{2}.\\
If $\exprind{2}\to^* \val$ and \exprind{1}$\to^*$ \vabs{x} \exprind{1}' then the RHS reduces to $\exprind{1}'[\val/x]$. If $\exprind{0}\to^*\val'$ then the LHS reduces to \vgetk{} (\vbuildk{} \val' (\vabs{x} \exprind{1}')) \exprind{2}. It further reduces to \vgetk{} [\exprind{1}'[0/x], $\ldots$, \exprind{1}'[\val'-1/x]] \exprind{2}. Without loss of generality, assume that \exprind{1}' is a value. Then the LHS further reduces to \vgetk{} [\exprind{1}'[0/x],$\ldots$,\exprind{1}'[\val'-1/x]] \val. By assumption \val<\val' so the arguments are in the domain of \vgetk{} and the LHS finally reduces to \exprind{1}'[\val/x].

\underline{Rule (f).2}: \vifoldk{} f z n+1\evalsto \vifoldk{} (\vabs{a i}{} f a (i+1)) (f z 0) n.\\
By induction on $n$. If $n=0$ then the LHS reduces to f (\vifoldk{} f z 0) 0 which reduces to f z 0. The RHS immediately reduces to f z 0. If $n>0$ then the LHS reduces to f (\vifoldk{} f z 0) n. The RHS reduces in a few steps to f ( \vifoldk{} (\vabs{a i}{} f a (i+1)) (f z 0) n-1) (n-1+1). By I.H. the first argument of f in both cases reduces to the same value and we are done.  

\underline{Rule (f).4}: 
\vifoldk{} (\vabs{a i}{} \iif(i = \exprind{0}) \then{} \exprind{1} \elsee{} a) z n \evalsto \lett{} a = z \inn{} \lett{} i = \exprind{0} in \exprind{1} \textit{(if \exprind{0} does not mention a or i, and $0\leq \exprind{0} < n$)}.\\
By induction on $n> \exprind{0}$. If $n=\exprind{0}+1$ then the LHS reduces in a few steps to (\vabs{a}{} \iif{} n-1 = \exprind{0}) \then{} \exprind{1} \elsee{} a)(LHS(n-1)). As $n-2<\exprind{0}$ a straightforward induction shows that LHS(n-1) reduces to z. Hence the LHS reduces to \exprind{1}[z/a,\exprind{0}/i], and so does the RHS. For $n>\exprind{0}+1$ the LHS will reduce to (\vabs{a}{} \iif{} n-1 = \exprind{0}) \then{} \exprind{1} \elsee{} a)(LHS(n-1)). As $n>\exprind{0}+1$, $n-1 \neq \exprind{0}$ and hence the else branch will be chosen. By I.H. LHS(n-1) reduces to the same thing as the RHS, and hence the LHS reduces to the same thing as LHS(n-1) which reduces to the same thing as the RHS.

\underline{Rule (g)}: 
\vifoldk{} (\vabs{a i}{}
\adpair{\genexprind{f}{0} (\vfstk{} a) i}{\genexprind{f}{1} (\vsndk{} a) i}) \adpair{\genexprind{z}{0}}{\genexprind{z}{1}} n
\evalsto
 (\vifoldk{} \genexprind{f}{0} \genexprind{z}{0} n,  \vifoldk{} \genexprind{f}{1} \genexprind{z}{1} n). \\
 By induction on $n$. If $n=0$, both reduce to \adpair{\genexprind{z}{0}}{\genexprind{z}{1}}.  If $n>0$, then the LHS reduces in several steps to
\adpair{\genexprind{f}{0} (\vfstk{} LHS(n-1)) (n-1)}{\genexprind{f}{1} (\vfstk{} LHS(n-1)) (n-1)}, where LHS(n-1) denotes the term on the LHS for n-1 instead of n. By I.H. \vifoldk{} (\vabs{a i}{}
\adpair{\genexprind{f}{0} (\vfstk{} a) i}{\genexprind{f}{1} (\vsndk{} a) i}) \adpair{\genexprind{z}{0}}{\genexprind{z}{1}} n $\to^*$ \adpair{\valind{0}}{\valind{1}} iff the RHS at $n-1$ reduces to \adpair{\valind{0}}{\valind{1}}. So the RHS reduces to \adpair{\genexprind{f}{0} \valind{0} n-1}{\genexprind{f}{1} \valind{1} n-1}, and so does the RHS.

\underline{Rule (a).7}: f(\lett{} x = \exprind{0} \inn{} \exprind{1})\evalsto 
\lett{} x = \exprind{0} \inn{} f(\exprind{1}).\\
By case analysis on f. If f is a constant, the result is clear. Assume f is a unary construct. Then if  \exprind{0}$\to^*\val$ then \lett{} x = \exprind{0} \inn{} \exprind{1}$\to^*$ \exprind{1}[\val/x]. If \exprind{1}[\val/x]$\to^*$\val' then the LHS reduces to f(\val'). The RHS reduces to f(\exprind{1})[\val/x]=f(\exprind{1}[\val/x]) which reduces to f(\val'). If f is a $n$-ary operator for $n>1$, the proof is similar but the evaluation order of the arguments comes into play making the argument a bit more cumbersome.
\end{proof}


\subsection{Type Soundness of \difftransalt{}}

We now prove that \difftransalt{} is a well-defined macro on the language. That is, it is compatible with typing and substitution. 

\begin{lemma}[Well-typedness of $\difftransalt{}$]
\label{lem:well_typedness_D}
	If $\Gamma\vdash e:A$ then $\difftranstypealt{\Gamma}\vdash \difftransalt{e}:\difftranstypealt{A}$.
\end{lemma}

\begin{proof}
	By induction on the typing tree of $e$ and case analysis. All cases are very similar so we only show one:
	$\Gamma\vdash \exprind{1}\exprind{2}:A$. There exists a type $B$ such that $\Gamma\vdash \exprind{1}:\typefunone{B}{A}$ and $\Gamma\vdash \exprind{2}:B$ so by I.H. $\difftranstypealt{\Gamma}\vdash \difftransalt{\exprind{1}}:\typefunone{\difftransalt{B}}{\difftranstypealt{A}}$ and $\difftranstypealt{\Gamma}\vdash \difftransalt{\exprind{2}}:\difftranstypealt{B}$. So $\difftranstypealt{\Gamma}\vdash \difftransalt{\exprind{1}\exprind{2}}=\difftransalt{\exprind{1}}\difftransalt{\exprind{2}}:\difftranstypealt{A}$.
\end{proof}

\begin{proposition}[Compatibility of $\difftransalt{}$ with substitution]
\label{lem:functoriality_D}
If $\Gamma,y:B\vdash \exprind{1}:A$ and $\Gamma\vdash \exprind{2}:B$ then
	$\difftranstypealt{\Gamma}\vdash \difftransalt{\exprind{1}[\exprind{2}/y]}=\difftransalt{\exprind{1}}[\difftransalt{\exprind{2}}/y]:\difftranstypealt{A}$.
\end{proposition}

\begin{proof}
	By induction and case analysis on $\exprind{1}$. All cases are again very similar so we only present a few cases.
	\begin{itemize}
		\item $\exprind{1}:= \vabs{\vmore{\text{z}}}{\exprind{3}}$:  
		\begin{eqnarray*}
			\difftransalt{\exprind{1}[\exprind{2}/y]}
			&=\difftransalt{(\vabs{\vmore{\text{z}}}{\exprind{3}})[\exprind{2}/y]} \\
			&=\difftransalt{\vabs{\vmore{\text{z}}}{\exprind{3}[\exprind{2}/y]}}\\
			&=\vabs{\vmore{\difftransalt{\text{z}}}}{\difftransalt{\exprind{3}[\exprind{2}/y]}}\\
			&\IH \vabs{\vmore{\difftransalt{\text{z}}}}{(\difftransalt{\exprind{3}}[\difftransalt{\exprind{2}}/y])}\\
			&=\vabs{\vmore{\difftransalt{\text{z}}}}{(\difftransalt{\exprind{3}})[\difftransalt{\exprind{2}}/y]}\\
			&=\difftransalt{\vabs{\vmore{\text{z}}}{\exprind{3}}}[\difftransalt{\exprind{2}}/y]
		\end{eqnarray*}
		\item $\exprind{1}:= \exprind{3}\exprind{4}$: \begin{eqnarray*}
			\difftransalt{\exprind{1}[\exprind{2}/y]}
			&=\difftransalt{(\exprind{3}\exprind{4})[\exprind{2}/y]}\\
			&=\difftransalt{\exprind{3}[\exprind{2}/y]\exprind{4}[\exprind{2}/y]}\\
			&=\difftransalt{\exprind{3}[\exprind{2}/y]}\difftransalt{\exprind{4}[\exprind{2}/y]}\\
			&\IH \difftransalt{\exprind{3}}[\difftransalt{\exprind{2}}/y]\difftransalt{\exprind{4}}[\difftransalt{\exprind{2}}/y]\\
			&=(\difftransalt{\exprind{3}}\difftransalt{\exprind{4}})[\difftransalt{\exprind{2}}/y]\\
			&=\difftransalt{\exprind{3}\exprind{4}}[\difftransalt{\exprind{2}}/y]
		\end{eqnarray*}
			\item $\exprind{1}:= \vifthenelse{\exprind{0}}{\exprind{3}}{\exprind{4}}$:
			\begin{eqnarray*}
			&\difftrans{(\vifthenelse{\exprind{0}}{\exprind{3}}{\exprind{4}})[\exprind{2}/y]}\\
			&= \difftrans{(\vifthenelse{\exprind{0}[\exprind{2}/y]}{\exprind{3}[\exprind{2}/y]}{\exprind{4}[\exprind{2}/y]})} \\
			&= \vifthenelse{(\vfstk{} \difftrans{\exprind{0}[\exprind{2}/y]})}{\difftrans{\exprind{3}[\exprind{2}/y]}}{\difftrans{\exprind{4}[\exprind{2}/y]}} \\
			&\IH  \vifthenelse{(\vfstk{} \difftrans{\exprind{0}}[\difftrans{\exprind{2}}/y])}{\difftrans{\exprind{3}}[\difftrans{\exprind{2}}/y]}{\difftrans{\exprind{4}}[\difftrans{\exprind{2}}/y]} \\
			&= (\vifthenelse{(\vfstk{} \difftrans{\exprind{0}})}{\difftrans{\exprind{3}}}{\difftrans{\exprind{4}}})[\difftrans{\exprind{2}}/y] \\
			&= \difftrans{(\vifthenelse{\exprind{0}}{\exprind{3}}{\exprind{4}})}[\difftrans{\exprind{2}}/y]
			\end{eqnarray*}
	\end{itemize}
\end{proof}
\subsection{Denotational semantics} 
\label{sub:denotational_semantics}

Differentiation is about functions, not values, and the operational semantics is ill-suited for stating and proving correctness of \difftransalt{}. The way to state correctness is roughly that \difftransalt{\expr{}}, seen as a function, computes the partial derivatives of \expr{}, also seen as a function. This is stated formally using denotational semantics. 

Similarly to \cite{barthe2020versatility} where a standard denotational semantics in terms of sets and functions is given to the language, we give a standard denotational semantics $\sem{-}$ in terms of sets and partial functions. 

A type $\tau$ and a context $\Gamma$ are denoted by sets $\sem{\tau},\sem{\Gamma}$, and a term $\Gamma\vdash \expr{}:\tau$ by a partial function $\sem{\Gamma}\to\sem{\tau}$. The domain of the function depends on \expr{} and is given by structural induction on \expr{}. Every construction of the language is total and the only partiality comes from the constants following the $\Dom$ function.

In particular, the denotation of \typedouble{} will be the reals $\RR$ and the denotation of \typebool{} is a set of two elements. We can see this space as the disjoint union of two Cartesian spaces, each consisting of one element. By doing so, a partial smooth function $\sem{\typedouble}\to\sem{\typebool}$ is exactly a function constant on each connected component of its domain. This allows us, following \cite{abadi2019simple}, to denote $\vconst{<}:\typedouble{},\typedouble{}\funarrow{} \typedouble{}$ by a partial smooth function $\RR\times\RR\to \sem{\typebool}$ undefined on the diagonal $\{(x,x)~|~x\in\RR\}\subset \RR\times\RR$. 
Similarly for the operators \vconst{>}, \vconst{==}, \vconst{<>}. This is for the purpose of stating and proving correctness of \difftransalt{} only. 

We are only interested here in proving correctness of \difftransalt{} for partial smooth functions for the usual notion of differentiation, so the rules for differentiation from Figure \ref{fig:diff_trans} for \vconst{+}, \vconst{-}, \vconst{*}, \vconst{/}, \vconst{**} only apply to these operators for the type \typedouble. 

Furthermore, if the language had sum types, we would have $\difftransalt{\tau+\tau'}=\difftransalt{\tau}+\difftransalt{\tau'}$. In particular, based on the fact that $\typebool=1+1$, where $1$ is the unit type, one might expect $\difftransalt{\typebool}=\difftransalt{1+1}=\difftransalt{1}+\difftransalt{1}=1\times 1 + 1\times 1$. The latter equal is coming from the fact that $1=\RR^0$ can also be seen as a Cartesian space. Note that $1\times 1 + 1\times 1$ is not $\typebool\times\typebool$, however it is equivalent to $(1+1)\times 1$ which can be injected into $\typebool\times\typebool$ by sending $\true$~ to $(\true,\false)$ and $\false$~ to $(\false,\false)$. A similar story happens with the natural numbers \typeindex, which is equivalent to a countable sum $1+1+\ldots 1 +1$. In this case $\difftransalt{1+1+\ldots 1 +1}$ which should be $1\times 1 +\ldots 1\times1$ can be embedded into $\typeindex \times \typeindex$ by sending \text{n} to $(\text{n},0)$. By doing so, we recover exactly our choice for \difftransalt{} from Figure \ref{fig:diff_trans}.
 
Finally, our denotational semantics is standard and compatible with the operational semantics. That is

\begin{proposition}[Soundness of denotational semantics]
\label{lem:op_denot_sem}
	If $\Gamma\vdash \expr{}\to\expr{}':A$ then $\sem{\expr{}}=\sem{\expr{}'}$.
\end{proposition}

\subsection{Correctness of syntactic differentiation}

In \cite{barthe2020versatility}, the authors proved correct a fragment of our language. We extend their logical relations argument to our whole language. In particular, we now have to deal with conditionals, partial operations and array-related constructs such as \viteratek{}. We recall and extend their framework for our purpose.

Let $\text{dual}_\text{y}(\text{x}):=(\text{y},1)$ if $\text{x}\equiv \text{y}$ and $(\text{x},0)$ otherwise. We denote by $\partial_y$ the partial derivative with respect to the variable $y$.

\begin{definition}[Open logical relation]
	Let $\Theta=\text{x}_1:\typedouble{},\ldots,\text{x}_n:\typedouble{}$ be a fixed environment. Define the family of relations $(\logR^{\Theta}_\tau)_{\Theta,\tau}$ by induction on types $\tau$ as follows:
		\[
		\small
\begin{array}{lllll}
	\exprind{1}~\logR^{\Theta}_\typedouble{}~\exprind{2} \Longleftrightarrow \left\{\begin{array}{lllll}
	\Theta\vdash \exprind{1}:\typedouble \wedge \difftranstypealt{\Theta}\vdash \exprind{2}:\difftranstypealt{\typedouble} \\
	\forall \text{y}:\typedouble.\sem{\Theta\vdash\vfstk{}(\exprind{2})[ \text{dual}_\text{y}(\text{x}_i)/\difftransalt{\text{x}_i}]_{\text{1}\leq \text{i}\leq \text{n}}:\typedouble}\\
	\qquad\qquad\,\,= \sem{\Theta\vdash \exprind{1}:\typedouble} \\
	\forall \text{y}:\typedouble.\sem{\Theta\vdash\vsndk{}(\exprind{2})[ \text{dual}_\text{y}(\text{x}_i)/\difftransalt{\text{x}_i}]_{\text{1}\leq \text{i}\leq \text{n}}:\typedouble}\\
	\qquad\qquad\,\,= \partial_{\text{y}}\sem{\Theta\vdash \exprind{1}:\typedouble}
	\end{array} \right.\medskip\\
	\exprind{1}~\logR^{\Theta}_{\typefunone{$\tau_1$}{$\tau_2$}}~\exprind{2} \Longleftrightarrow \left\{\begin{array}{ll}
	\Theta\vdash \exprind{1}:\typefunone{$\tau_1$}{$\tau_2$} \wedge \difftranstypealt{\Theta}\vdash \exprind{2}:\difftranstypealt{\typefunone{$\tau_1$}{$\tau_2$}}  \\
	\forall \exprind{3},\exprind{4}.~ \exprind{3}~\logR^{\Theta}_{\tau_1}~\exprind{3}\Rightarrow \exprind{1}\exprind{3}~\logR^{\Theta}_{\tau_2}~\exprind{2}\exprind{4}
	\end{array} \right. 	\medskip\\
	\exprind{1}~\logR^{\Theta}_{\tau_1\times\tau_2}~\exprind{2} \Longleftrightarrow \left\{\begin{array}{ll}
	\Theta\vdash \exprind{1}:\tau_1\times\tau_2 \wedge \difftranstypealt{\Theta}\vdash \exprind{2}:\difftranstypealt{\tau_1\times\tau_2}  \\
	\vfstk{}(\exprind{1})~ \logR^\Theta_{\tau_1} ~\vfstk{}(\exprind{2}) \wedge \vsndk{}(\exprind{1}) ~\logR^\Theta_{\tau_2}~ \vsndk{}(\exprind{2})
	\end{array} \right. 	\medskip\\
	\exprind{1}~\logR^{\Theta}_{\typebool}~\exprind{2} \Longleftrightarrow \left\{\begin{array}{lll}
	\Theta\vdash \exprind{1}:\typebool \wedge \difftranstypealt{\Theta}\vdash \exprind{2}:\difftranstypealt{\typebool} \\
	\forall \text{y}:\typedouble.\sem{\Theta\vdash\vfstk{}(\exprind{2})[ \text{dual}_\text{y}(\text{x}_i)/\difftransalt{\text{x}_i}]_{\text{1}\leq \text{i}\leq \text{n}}:\typebool}\\
	\qquad\qquad\,\,= \sem{\Theta\vdash \exprind{1}:\typebool} \\
	\forall \text{y}:\typedouble.\sem{\Theta\vdash\vsndk{}(\exprind{2})[ \text{dual}_\text{y}(\text{x}_i)/\difftransalt{\text{x}_i}]_{\text{1}\leq \text{i}\leq \text{n}}:\typebool}\\
	\qquad\qquad\,\,= \sem{\Theta\vdash \false :\typebool}
	\end{array} \right. 	\medskip\\
	\exprind{1}~\logR^{\Theta}_{\typearray{$\tau$}}~\exprind{2} \Longleftrightarrow \left\{\begin{array}{ll}
	\Theta\vdash \exprind{1}:\typearray{$\tau$} \wedge \difftranstypealt{\Theta}\vdash \exprind{2}:\difftranstypealt{\typearray{$\tau$}} \\
	\forall \text{i},  \exprind{1}[\text{i}]~ \logR^\Theta_\tau ~\exprind{2}[\text{i}]
	\end{array} \right.
\end{array}
	\]
\end{definition}

 We now extend the lemma showing that the relation is compatible with the denotational semantics:

\begin{lemma}
	Let  $\Theta=\text{x}_1:\typedouble{},\ldots,\text{x}_n:\typedouble{}$. Then, the following hold:
	\[\begin{array}{ll}
		\exprind{1}~\logR^{\Theta}_\tau~\exprind{2} \wedge \sem{\Theta\vdash \exprind{1}:\tau}=\sem{\Theta\vdash \exprind{0}:\tau}\Rightarrow \exprind{0}~\logR^{\Theta}_\tau~\exprind{2} \\
		\exprind{1}~\logR^{\Theta}_\tau~\exprind{2} \wedge \sem{\difftranstypealt{\Theta}\vdash \exprind{2}:\difftranstypealt{\tau}}=\sem{\difftranstypealt{\Theta}\vdash \exprind{3}:\difftranstypealt{\tau}}\Rightarrow \exprind{1}~\logR^{\Theta}_\tau~\exprind{3}
	\end{array}\]
\end{lemma}

Again following \cite{barthe2020versatility}, we extend $(\logR^{\Theta}_\tau)$ to the family $(\logR^{\Gamma,\Theta}_\tau)_{\Gamma,\Theta,\tau}$ where $\Gamma$ ranges over arbitrary environments, as follows:
\[\exprind{0}~\logR^{\Gamma,\Theta}_\tau~\exprind{1}\Longleftrightarrow\]
\[ (\Gamma,\Theta\vdash \exprind{0}:\tau)\wedge (\difftranstypealt{\Gamma},\difftranstypealt{\Theta}\vdash \exprind{1}:\difftranstypealt{\tau})\wedge (\forall \gamma,\gamma'. \gamma~\logR^\Gamma_\Theta~ \gamma')\Rightarrow \exprind{1}\gamma ~\logR^\Theta_\tau~\exprind{1}\gamma'\]
$\gamma,\gamma'$ ranges overs substitution, and if we denote by $\supp(f)$ the support of the function f: 
\[ \gamma~\logR^\Gamma_\Theta~ \gamma' \Longleftrightarrow (supp(\gamma)=\Gamma)\wedge (\supp(\gamma')=\difftranstypealt{\Gamma})\wedge (\forall x:\tau\in\Gamma,\gamma(x)~\logR^\Theta_\tau ~\gamma'(\difftrans{x}))\]

We now prove a fundamental lemma:

\begin{lemma}[Fundamental lemma]
	For all environments $\Gamma,\Theta$ and for any $\Gamma,\Theta\vdash \expr{}:\tau$, we have $\expr{}~\logR^{\Gamma,\Theta}_\tau~ \difftransalt{\expr{}}$.
\end{lemma}

\begin{proof}[Proof sketch]
The proof is by case analysis on \expr{} and is similar to the one from \cite{barthe2020versatility}. We only present the main new cases. We also only sketch the argument and skip writing all the substitutions $\gamma,\gamma'$. It is straightforward to complete our sketch of proof and add them.
	\begin{itemize}
		\item $ \expr{}:=[\expr{}_\text{1},\ldots,\expr{}_\text{n-1}]$: by I.H. we have that for all \text{i}, $\expr{}_\text{i}~\logR^{\Gamma,\Theta}_\tau~\difftransalt{\expr{}_\text{i}}$, hence we are done by definition of $\logR$ on array types.
		\item $\expr{}:=\vbuildk{}~\text{n}~\text{f}$: then $\sem{\expr{}}=\sem{[\text{f}(0),\ldots,\text{f}(n-1)]}$ by Lemma \ref{lem:op_denot_sem}. We conclude by the I.H. on \text{f}, and again the definition of $\logR$ on array types.
		\item $\expr{}:=\exprind{1}\vconst{<} \exprind{2}$: then it is immediate by the definitions of \difftransalt{} and $\logR$ for $\typebool$.
		\item $\expr{}:=\vifthenelse{\exprind{0}}{\exprind{1}}{\exprind{2}}$: this expression has a ground type $\typemat$. As we have seen with the case of arrays above, products and arrays come essentially for free, so w.l.o.g. assume that $\typemat=\typedouble$. Then, what we have to show is an equality in the semantics. It suffices to show this for a covering set of the domain. We obtain one by considering connected components of $A:=\sem{\exprind{1}}^{-1}(\sem{\true})$ and $B:=\sem{\exprind{1}}^{-1}(\sem{\false})$. On each such connected component $D\subset A$, $\sem{\exprind{1}}=\sem{\true}$ and so $\sem{\expr{}}_{|D}=\sem{\exprind{1}}_{|D}$ and hence by I.H. on \exprind{1} we have that $\sem{\vsndk{}\difftransalt{\expr{}}}_{|D}=\sem{\difftransalt{\exprind{1}}}_{|D}=\partial_y\sem{\exprind{1}}_{|D}=\partial_y\sem{\expr{}}_{|D}$. This is fine by Lemma \ref{lem:op_denot_sem} and our choice of domain of definition of $\vconst{<}$. We get the result for every such $D$ and, as they cover the domain of definition of \expr{} we are done.
		\item $\expr{}:=\viteratek{}(\exprind{1})(\exprind{0})(\text{n})$: we prove the result by induction on $\text{n}$. Assume that \exprind{0},\exprind{1} are values for clarity. If $\text{n}=0$ then both sides reduces to respectively $\exprind{0}$ and $\difftransalt{\exprind{0}}$. We conclude by Lemma \ref{lem:op_denot_sem} and the I.H. on \exprind{0}. If $\text{n}>0$ then both sides reduce to respectively \exprind{1} (LHS(n-1)) (n-1) and \difftransalt{\exprind{1}} (RHS(n-1)) (n-1,0). By I.H. we have $LHS(n-1)~\logR~RHS(n-1)$. We also have $(n-1)~\logR~ (n-1,0)$. We can thus conclude using the I.H. on \exprind{1} and Lemma \ref{lem:op_denot_sem}.  
 	\end{itemize}
\end{proof}

As always with logical relations arguments, we obtain the correctness theorem as a corollary of the fundamental lemma:

\begin{theorem}[Correctness of the AD transform, limited]
	For any term $\Theta\vdash \expr{}:\typedouble$ as above, the term $\difftranstypealt{\Theta}\vdash \difftransalt{\expr{}}:\difftranstypealt{\typedouble}$ computes the partial derivatives of $\expr{}$, in the sense that for any variable $\text{y}$ we have:
\[\partial_\text{y}\sem{\Theta\vdash \expr{}:\typedouble}=\]
\[\sem{\Theta\vdash \vsndk{} (\difftransalt{\expr{}})[ \text{dual}_\text{y}(\text{x}_1)/\difftransalt{\text{x}_1},\ldots,\text{dual}_\text{y}(\text{x}_n)/\difftransalt{\text{x}_n}]:\difftranstypealt{\typedouble}}\]
\end{theorem}

From the theorem it is easy to deduce correctness for products and arrays, by post composing with a projection or !get!, and using the previous theorem.

\begin{corollary}[Correctness of the AD transform, full]
	\begin{itemize}
		\item $\sem{\diff{} f~ \text{x}}=(\partial_\text{x} \sem{f})(\text{x})$
		\item $\sem{\grad{} f~ \val}=\nabla_\val \sem{f}$ where $\nabla_\val$ is the gradient at $\val$
		\item $\sem{\jacob{} f~ \val}=J_\val\sem{f}$ where $J_\val$ is the jacobian at $\val$.
	\end{itemize}
\end{corollary}

 \section{Experimental Results}
\label{sec:exp}
In this section, we show how \system performs in practice.
We show the performance of the differentiated code for micro benchmarks as well as two real-world machine learning and computer vision applications.

\noindent \textbf{Experimental Setup.} We have performed the experiments using an iMac machine with
an Intel Core i5 CPU running at 2.7GHz, 32GB of DDR3 RAM at
1333Mhz.
The operating system is OS X 10.13.1. We use CLang 900.0.39.2 for compiling the generated C code, and Python 2.7.12 
for running the Python code, and Mono 5.8.1 for compiling and running F\# programs.
Furthermore, we use DiffSharp 0.6.3, Tapenade 3.14 (its web interface), Theano 0.9.0, TensorFlow 1.13, and Futhark 0.10.2.
For all the experiments, we compute the mean time of ten runs, and the time out is set to twenty minutes.
For the TensorFlow experiments, we do not use the XLA backend and we do not consider the run time of the first round in 
order to remove the overhead of its graph compilation.

Throughout this section, we compare the performance of the following alternatives:
\def\apara#1{\item #1}
\begin{itemize}
    \apara{DiffSharp (R):} The reverse-mode AD of the DiffSharp library~\cite{baydin2015diffsharp}. 
    \apara{DiffSharp (F):} The forward-mode AD of the DiffSharp library. 
    \apara{Tapenade (R):} The reverse-mode AD of the Tapenade framework~\cite{tapenade}.
    \apara{Tapenade (F):} The forward-mode AD of the Tapenade framework.
    \apara{Theano:} The AD offered by the Theano framework~\cite{bergstra2010theano}, which is a combination of the reverse-mode
    AD with symbolic differentiation.
    \apara{TensorFlow:} The AD offered by the Tensorflow framework~\cite{henriksen2017futhark}, which is based on the reverse-mode
    AD.
    \apara{Futhark:} A forward-mode AD implementation on top of the Futhark programming language~\cite{henriksen2017futhark}.
    \apara{\system:} The code generated by \system with different sets of optimisations 
    (e.g., loop fusion (\textit{LF}), loop-invariant code motion (\textit{LICM}), loop normalisation (\textit{LN}), and Destination-Passing Style (\textit{DPS})~\cite{dps_fhpc} 
    for stack-discipline memory management).
\end{itemize}

\noindent \textbf{Micro Benchmarks} which consist of the following vector expressions: 1) gradient of dot product of two vectors with respect to the first vector (which is a Jacobian matrix with a single row), 2) gradient of the maximum value of a vector with respect to the input vector (which is a Jacobian matrix with a single row), 3) gradient of addition of two vectors with respect to the first vector (which is a Jacobian matrix), and 4) gradient of the multiplication of a vector with a scalar value with respect to the scalar value (which is a Jacobian matrix with a single column). 

\begin{figure}[t!]
        \centering
    	\includegraphics[width=0.48\textwidth]{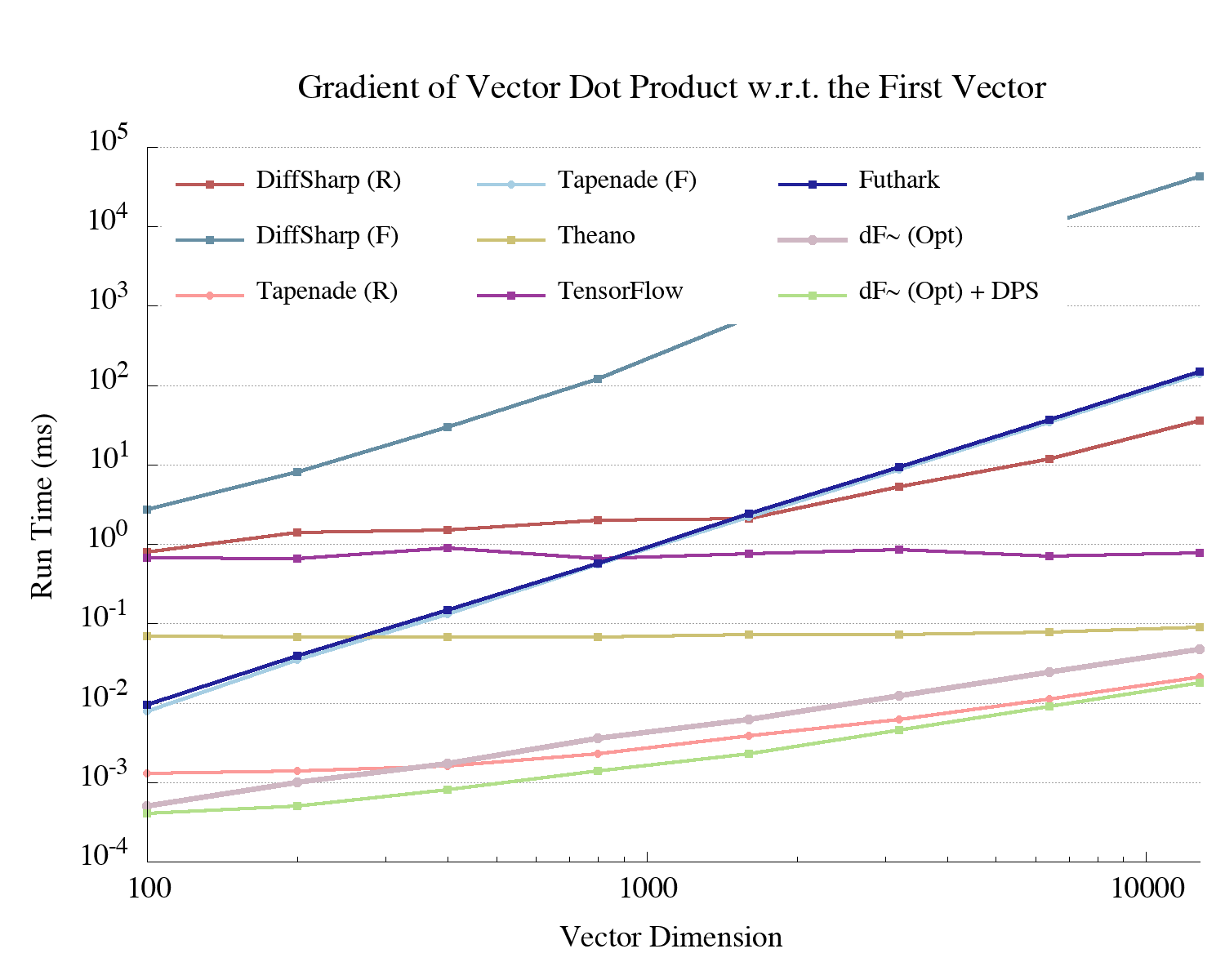}~\includegraphics[width=0.48\textwidth]{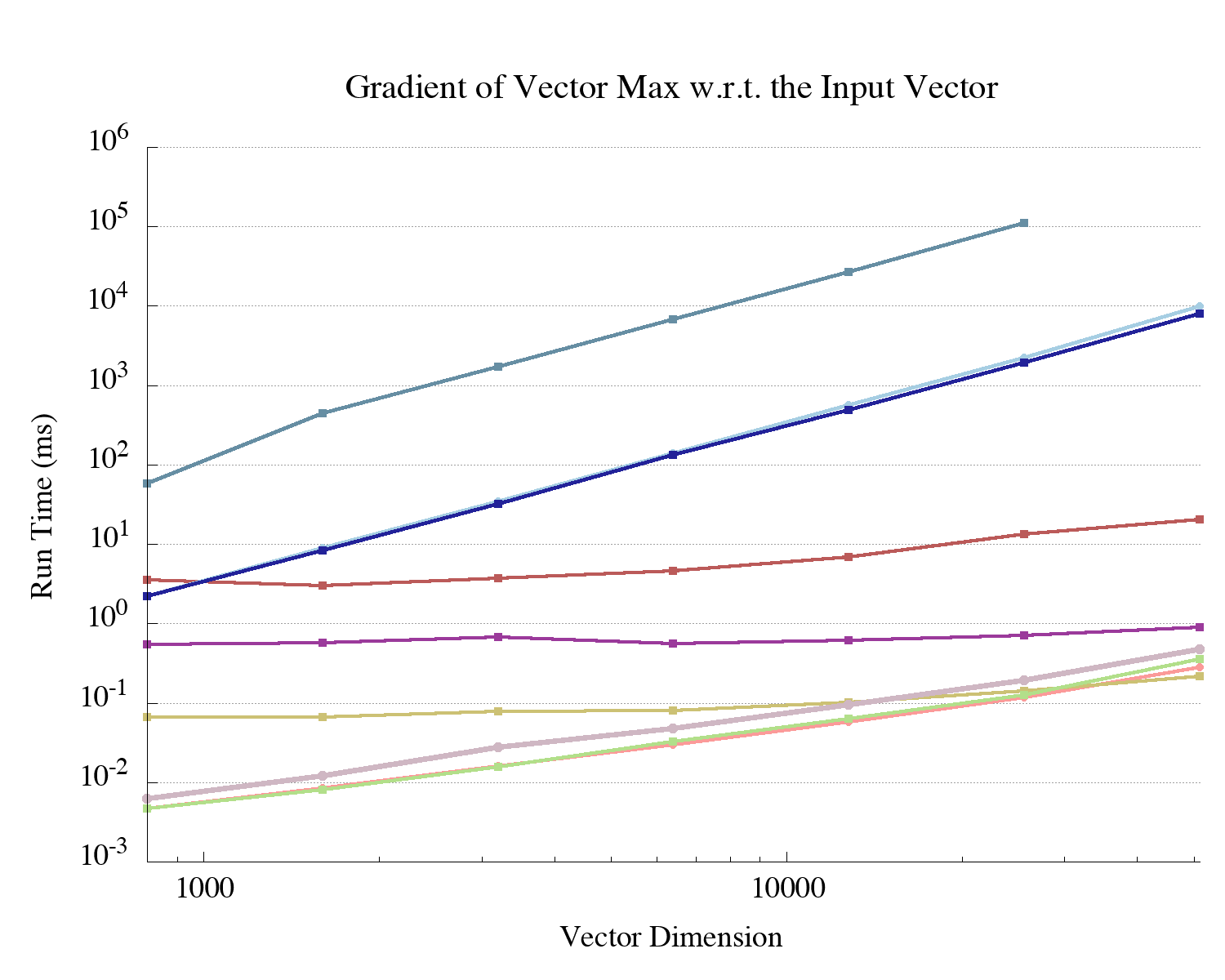}\\
        \includegraphics[width=0.48\textwidth]{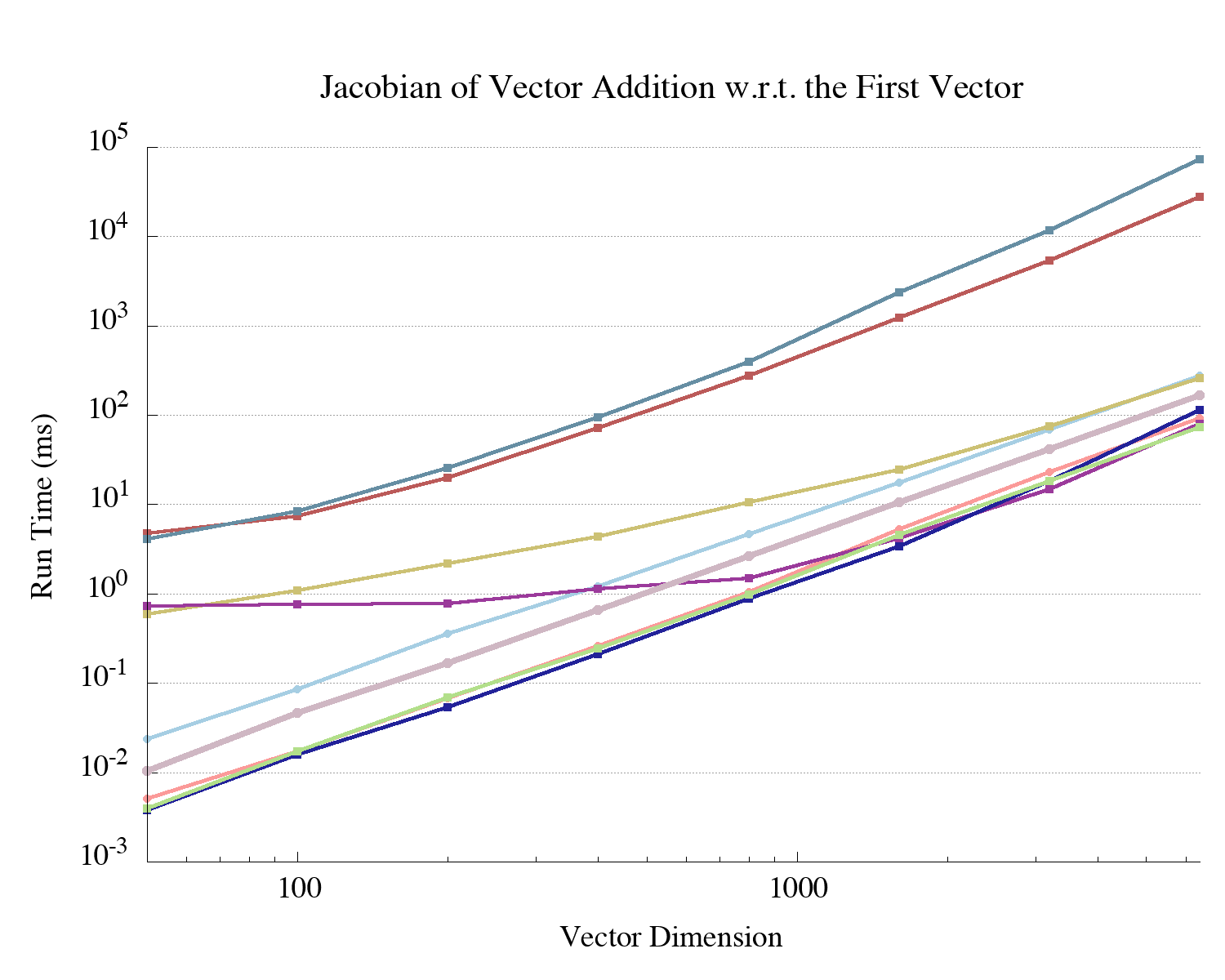}~\includegraphics[width=0.48\textwidth]{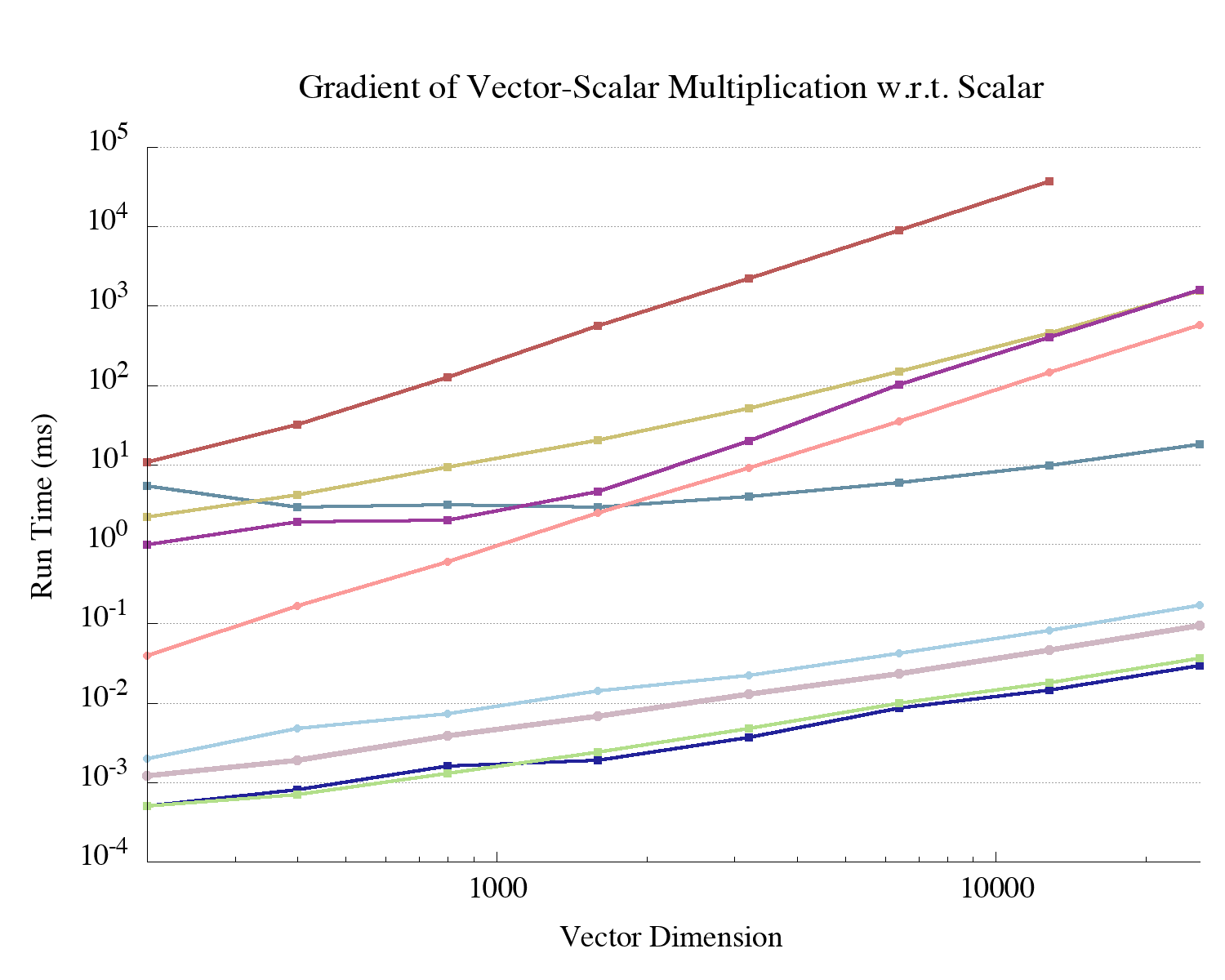}
        \caption{Performance results for Micro Benchmarks.}
        \label{micro_perf}
\end{figure}

Figure~\ref{micro_perf} shows the performance results for the mentioned micro benchmarks. 
In all cases, DiffSharp, Theano, and TensorFlow have performance overhead, which is reduced 
for larger data sizes thanks to the smaller interpretation overhead. 
The Futhark compiler generates C code after applying various types of optimisations such as loop fusion, 
code motion, dead-code elimination, double buffering (for loops), and standard optimisations like copy-propagation, 
constant-folding, and CSE among many other optimisations.
In all the experiments, we have applied the best of possible optimisations to the programs generated
by \system{} (denoted by \system{}~(Opt)). 
The performance of these programs is improved further when the generated C code uses DPS for stack-discipline memory 
management (denoted by \system{}~(Opt) + DPS). 

As in the first two cases the Jacobian matrix is a row vector, reverse-mode AD computes the whole Jacobian matrix in a single backward pass. 
However, forward-mode AD needs to iterate over each column to compute the corresponding derivative value.
Hence, Tapenade (F) and DiffSharp (F) have asymptotic performance difference.
Even though, the Futhark compiler implements many optimisations mentioned above, the forward-mode AD is asymptotically worse than the reverse-mode AD. 
On the other hand, DiffSharp (R), Theano, and TensorFlow show performance overhead which is reduced for larger data sizes.
In addition, Tapenade (R) and the code generated by \system{} (\system{}~(Opt)) show similar performance. This shows that
the optimisations explained in Section~\ref{sec:fsmooth_trans} have succesfully made the forward-mode AD code of \system{} as efficient
as the reverse-mode AD code. 
Finally, the performance of \system{} is improved further when the generated C code uses DPS for stack-discipline memory management. 

For the case of the addition of two vectors, as the Jacobian matrix is a square matrix, reverse-mode AD and forward-mode
AD show comparable performance. Both AD modes of DiffSharp are from two to three orders of magnitude slower than \system{}.
Both Theano and TensorFlow have performance overhead, which again reduces for bigger data sizes.
Specifically, TensorFlow becomes as good as Tapenade (R), Futhark, and \system{}. 

Finally, for the last case, as the Jacobian matrix is a column vector, the forward mode AD computes the whole Jacobian matrix in a single forward pass.
However, the reverse mode AD requires traversing over each row to compute the corresponding partial derivative values.
Hence, as opposed to computing the gradient for the dot product and the maximum element of a vector,
the systems based on the forward-mode AD show asymptotically better performance. 

\noindent \textbf{Non-Negative Matrix Factorization (NNMF)} is a useful tool which has many applications in various fields ranging from document clustering, recommendation systems, signal processing, to computer vision. 
For instance, in~\cite{liu2010distributed}, the authors study the NNMF of Web dyadic data represented as the matrix $A$. 
Dyadic data contains rich information about the interactions between the two participating sets. It is useful for a broad range of practical applications including Web search, Internet monetization, and social media content~\cite{liu2010distributed}. 
For example the (query, clicked URL) data is used in query clustering~\cite{2002:QCU:503104.503108}, query suggestions~\cite{Baeza-Yates:2004:QRU:2146449.2146527} and improving search relevance~\cite{Agichtein:2006:IWS:1148170.1148177}. 
Matrix factorization is a commonly used approach to understanding the latent structure of the observed matrix for various applications~\cite{Berry06algorithmsand, Sra06nonnegativematrix}. The authors present a probabilistic NNMF framework for a variety of Web dyadic data that conforms to different probabilistic distributions. For instance, an Exponential distribution is used to model Web lifetime dyadic data, e.g., user dwell time, and similarly the Poisson distribution is used to model count dyadic data, e.g., click counts.

\begin{figure}[t!]
        \centering
    	\includegraphics[width=0.48\textwidth]{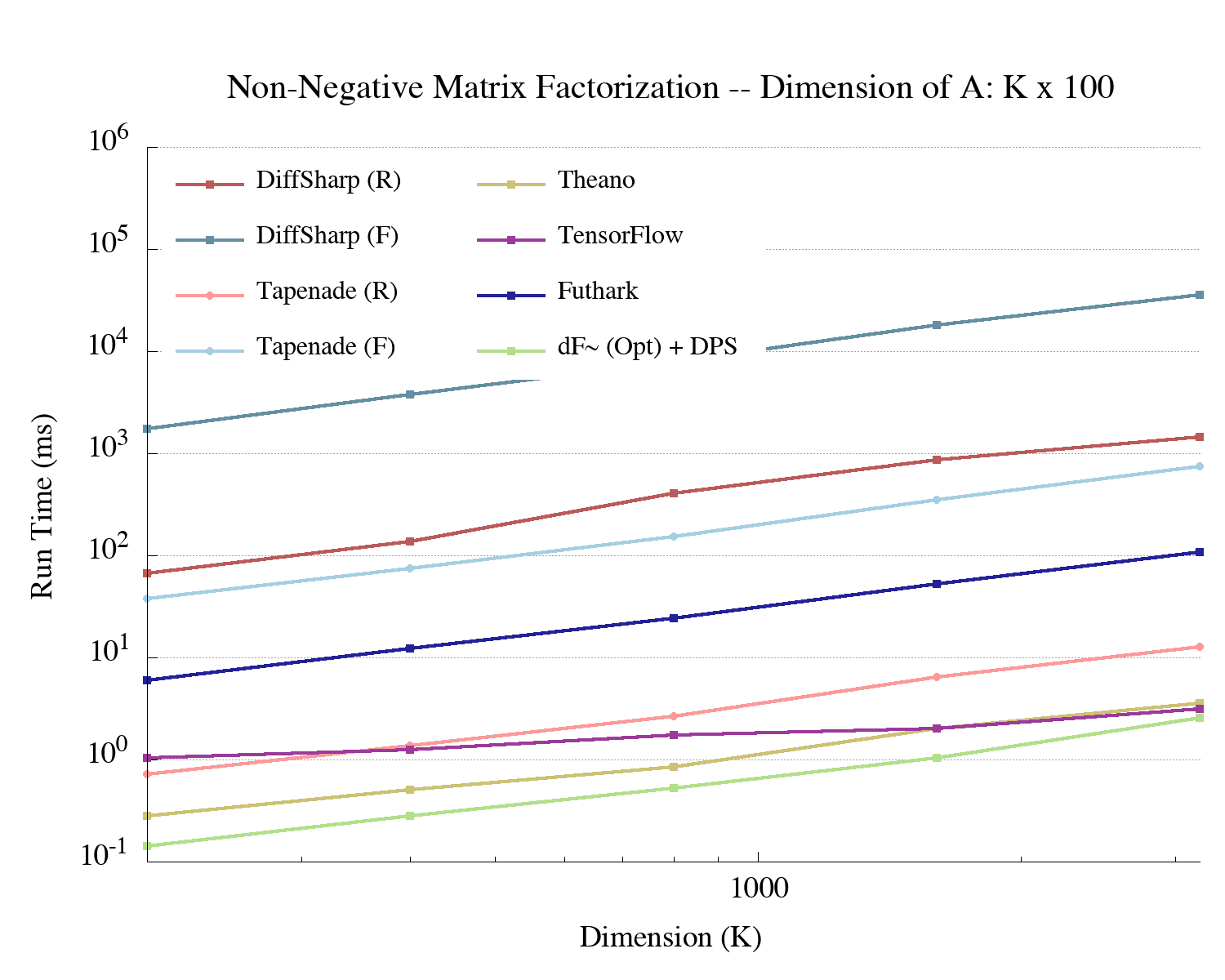}~\includegraphics[width=0.48\textwidth]{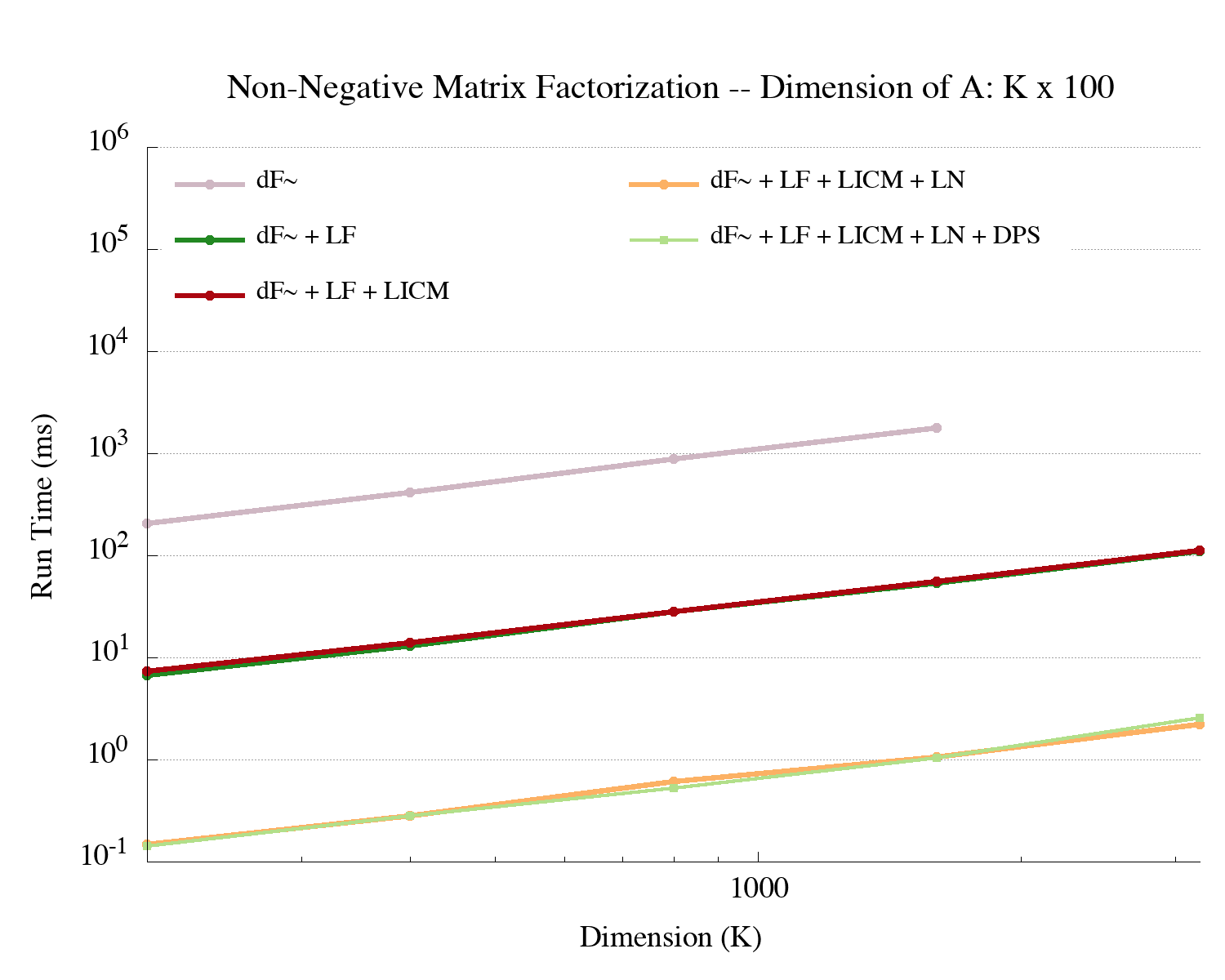}\\
        \includegraphics[width=0.48\textwidth]{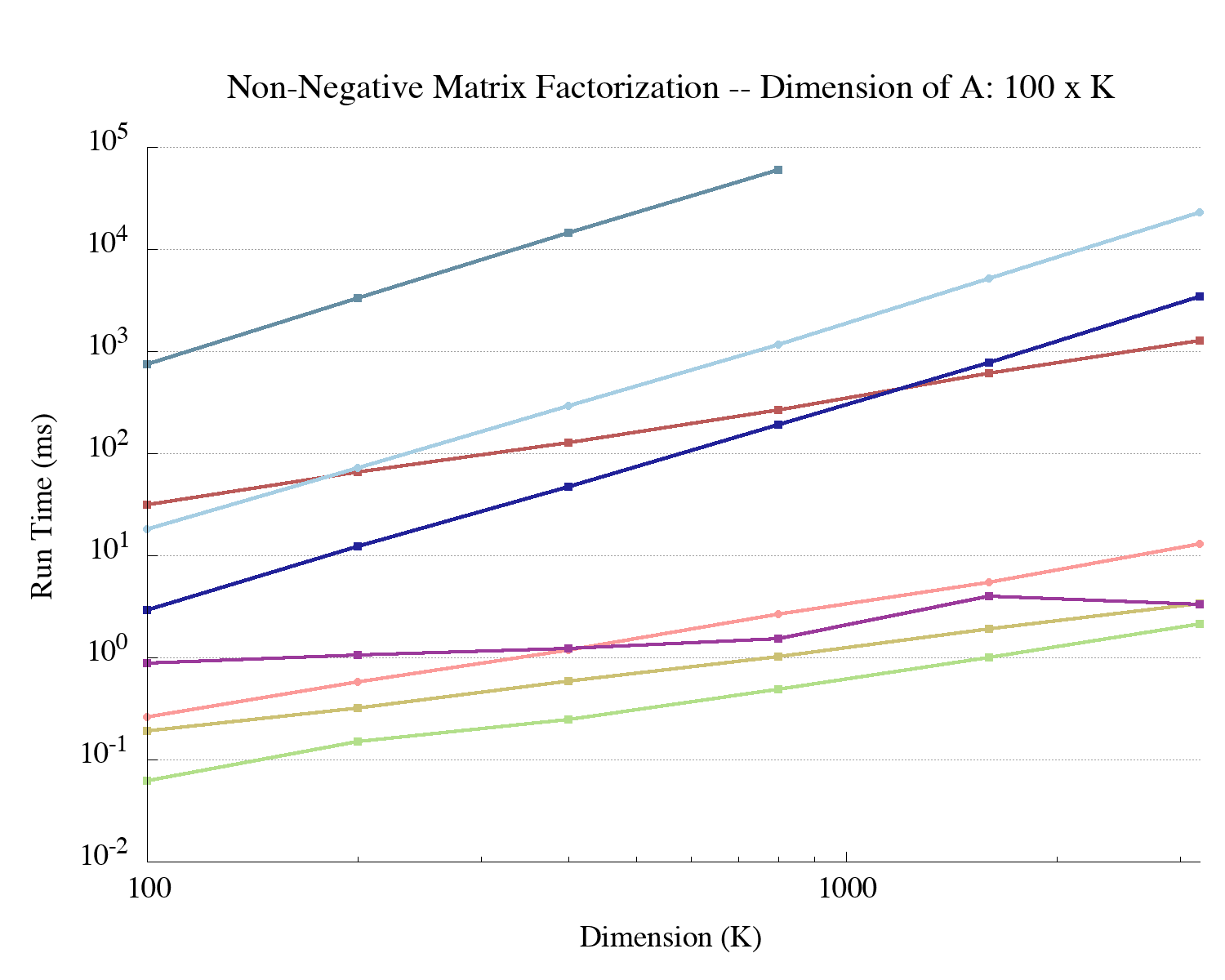}~\includegraphics[width=0.48\textwidth]{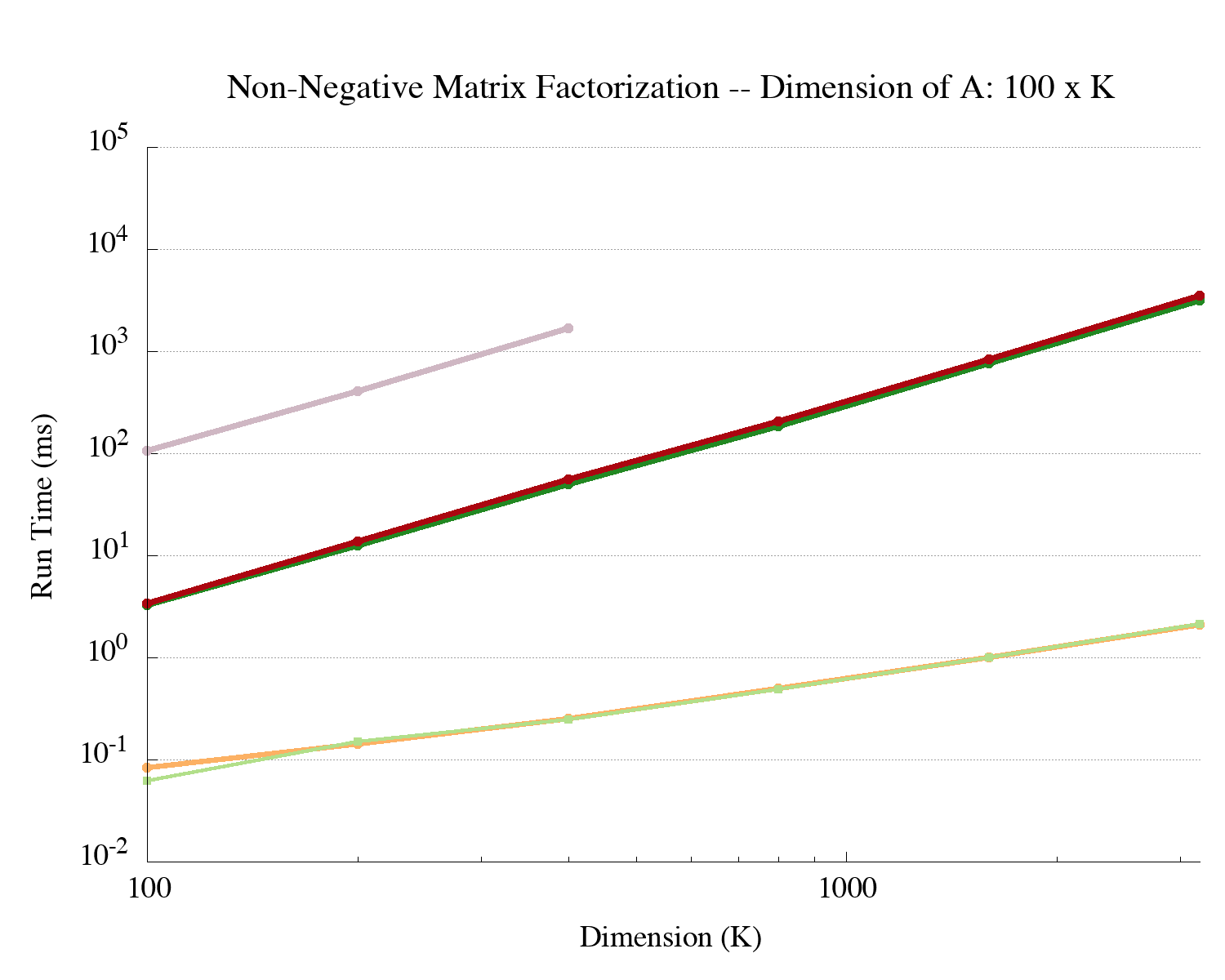}
        \caption{Performance results for NNMF.}
        \label{nmf_perf}
\end{figure}

The iterative algorithm to find $W$ and $H$ depends on the form of the assumed underlying distribution. 
In particular the update formula for gradient descent are derived by computing the gradient of the negative $\log$ of the likelihood function. 
For example, the negative $\log$ of the exponential distribution is represented as follows:

$$\mathcal{D}\big(A||\widetilde{A}\big)=\sum\limits_{(i,j)}\Bigg(\log\big(\widetilde{A}_{i,j}\big) + \frac{A_{i, j}}{\widetilde{A}_{i,j}}\Bigg) ,\hspace{1cm} \widetilde{A} = WH$$

\noindent The update formulas are derived manually, and for each new distribution it is the responsibility of the user to undertake the error prone and laborious task of deriving, optimizing, and implementing the update rules. 
\system automatically derives the gradient of the negative $\log$ of the likelihood function for the exponential distribution. After performing optimizations, \system produces an expression which is equivalent to the following update formula, which is manually derived by hand in~\cite{liu2010distributed}:

$$\frac{\partial \mathcal{D}}{\partial H} = W^T\Bigg(\frac{1}{WH}-\frac{A}{\big(WH\big)^2}\Bigg)$$

\noindent Figure~\ref{nmf_perf} shows the performance results of executing the derived update rule on DiffSharp, Tapenade, Theano, and \system.
For all the experiments, we consider factorizing the matrix $A$ into two vectors $W$ and $H$ (represented as $u$ and $v^T$, respectively).
Comparing Tapenade and \system, we observe that the reverse-mode AD of Tapenade behaves similarly to \system.
This shows that \system successfully generates efficient code for this case, which is an ideal case for the reverse-mode AD (the loss function is a scalar valued function, which should compute the gradient with respect to all elements of the input vector).
Finally, as the dimension of the vectors increases, Theano and TensorFlow converge to the same performance as \system and reverse-mode AD of Tapenade, because of two reasons. 
First, the overhead of invoking C functions from Python becomes negligible as the size of the vector increases. 
Second, Theano and TensorFlow invoke BLAS routines which are highly tuned and vectorised implementations for vector operations.

By comparing different configurations of \system, we observe the following three points.
First, the loop fusion improves the performance by around one order of magnitude. 
The generated C code after applying these optimisations is as efficient as the code generated by Futhark.
Second, the loop normalisation (\textit{LN}) optimisations
(which are the ones shown in Figure~\ref{fig:fsmooth_opt_iteration}) have the most impact by asymptotically improving
the performance from two to three orders of magnitude. As explained in Section~\ref{sec:fsmooth_trans}, thse transformation rules
should be combined with loop fission to become effective. 
Finally, the DPS representation slightly improves the performance
by using the stack-allocation descipline for memory management.

\noindent \textbf{Bundle Adjustment} \cite{triggs1999bundle,agarwal2010bundle,zach2014robust} is a computer vision problem, where the goal is to optimise several parameters in order to have an accurate estimate of the projection of a 3D point by a camera. 
This is achieved by minimizing an objective function representing the reprojection error. 

For the experiments, we compute the Jacobian matrix of the Project function in Bundle Adjustment. 
For a 3D point $X\in \mathbb{R}^3$ and a camera with rotation parameter $r \in \mathbb{R}^3$, center position $C \in \mathbb{R}^3$, focal index $f \in \mathbb{R}$, principal point $x_0 \in \mathbb{R}^2$, and radical distortion $k \in \mathbb{R}^2$, the Project function computes the projected point as follows:
\\ \\
\begin{tabular}{r @{\hskip3pt} c @{\hskip3pt} l}
\setlength\tabcolsep{.5pt}
$\text{project}\big(r, C, f, x_0, k, X\big)$
& = &
$\text{distort}\big(k, \text{p2e}\big(\text{rodrigues}\big(r, X - C\big)\big)\big)f + x_0$
\\
$\text{distort}\big(k, x\big)$
&=&
$x\big(1 + k_1||x||^2+k_2||x||^4\big)$
\\
$\text{p2e}\big(X\big) $
&=&
$X_{1..2} \div X_3$
\\
$\text{rodrigues}\big(r, X\big)$
&=&
$X \cos\big(\theta\big) + \big(v \times X\big)\sin\big(\theta\big) + v\big(v^TX\big)\big(1 - \cos\big(\theta\big)\big),$
\\ & &$\theta=||r||, v=\frac{r}{||r||}$
\end{tabular} 
\\

\noindent In order to better demonstrate the expressibility and conciseness of \ladsl{}, Figure~\ref{fig:ba_code} shows the implementation of these functions in this language.

Consider having $N$ 3D points and a vector of $11$ camera parameters. We are interested in computing a Jacobian matrix with $2N$ rows (the project function produces a 2D output for each input data point) and $11$ columns.

Figure~\ref{ba_perf} shows the performance results for computing the mentioned Jacobian matrix.
For this application, as the Jacobian matrix has more rows than columns, the forward-mode AD
outperforms the reverse-mode. Thus, we observe that the forward-mode AD of all systems
outperforms their reverse-mode.
\system outperforms all its competitors by up to three orders of magnitude.
More specifically, \system outperforms both forward and reverse mode of Tapenade. 
We observe that as opposed to the NNMF application, the loop normalisation rules have
a negligible impact for this application. However, the loop-invariant code motion optimisations result in
up to two times performance improvement by hoisting the shared computation outside 
the loop. In addition, DPS representation leads to an additional 10\% performance improvement.
Finally, Futhark outperforms the generated C code of \system{} mainly thanks to a better
data layout representation for tensors, which we plan to integrate into \system{} in the future. 





\begin{figure}[t!]
        \centering
    	\includegraphics[width=0.48\textwidth]{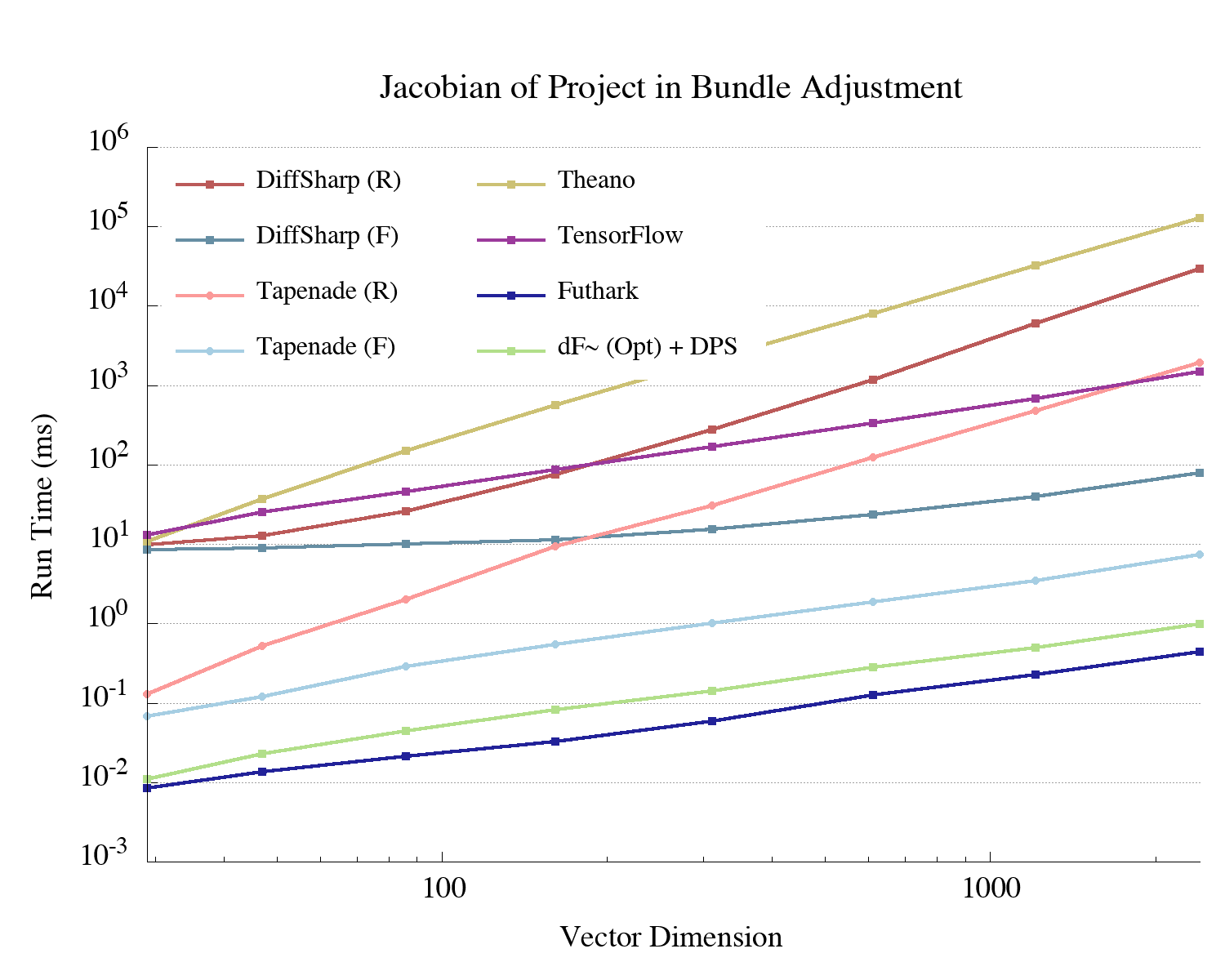}~\includegraphics[width=0.48\textwidth]{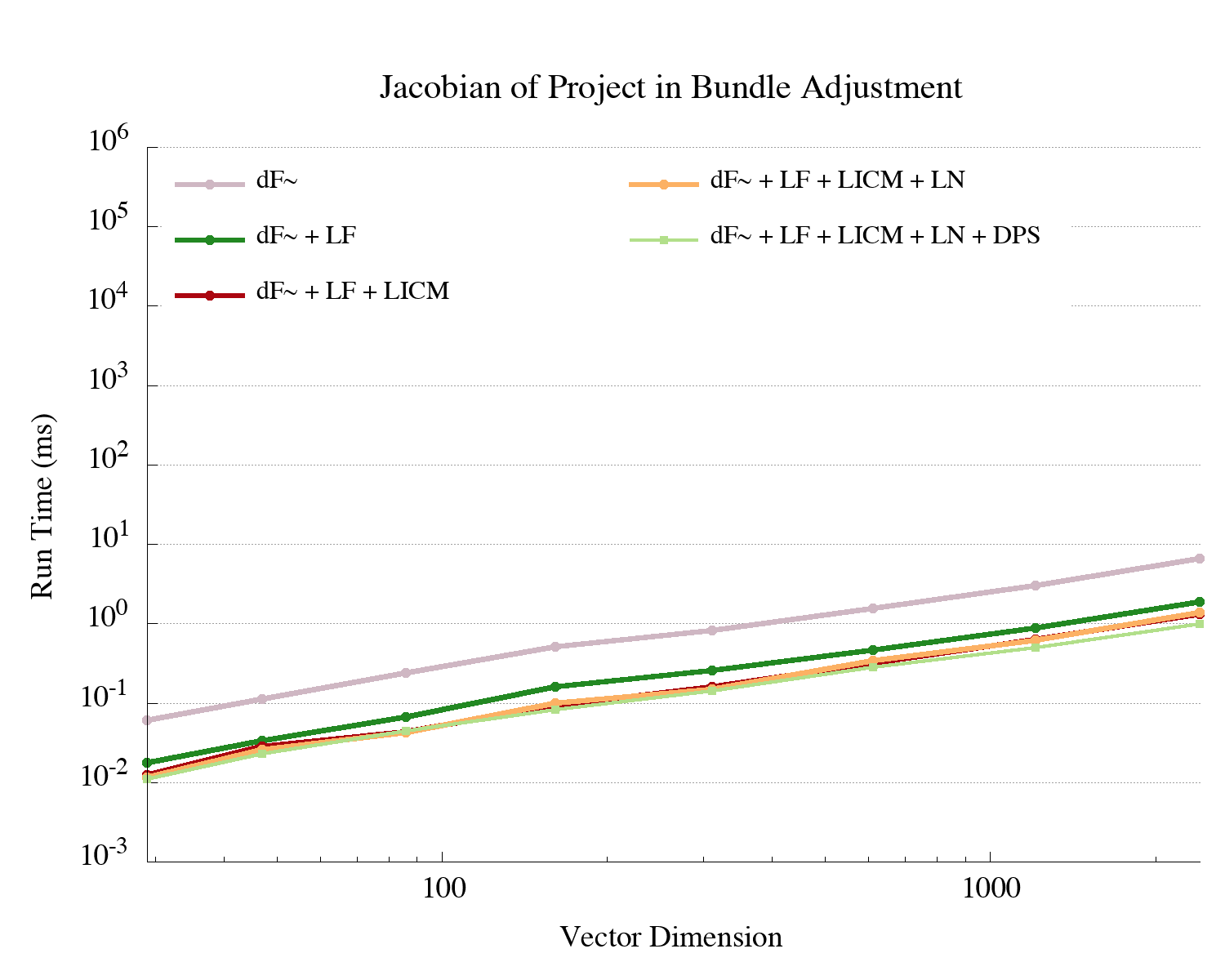}
        \caption{Performance results for Project in Bundle Adjustment.}
        \label{ba_perf}
\end{figure}

\begin{figure}[t]
\begin{fscode}
\lett{} distort = \vabs{(radical: Vector) (proj: Vector)}{}\\
\tabt \lett{} rsq = vectorNorm proj\\
\tabt \lett{} L = 1.0 + radical.[0] * rsq + radical.[1] * rsq * rsq\\
\tabt vectorSMul proj L\\
\lett{} rodrigues = \vabs{(rotation: Vector) (x: Vector)}{}\\
\tabt \lett{} sqtheta = vectorNorm rotation\\
\tabt \lett{} theta = sqrt sqtheta\\
\tabt \lett{} thetaInv = 1.0 / theta\\
\tabt \lett{} w = vectorSMul rotation thetaInv\\
\tabt \lett{} wCrossX = vectorCross w x\\
\tabt \lett{} tmp = (vectorDot w x) * (1.0 - (cos theta))\\
\tabt \lett{} v1 = vectorSMul x (cos theta)\\
\tabt \lett{} v2 = vectorSMul wCrossX (sin theta) \\
\tabt vectorAdd (vectorAdd v1 v2) (vectorSMul w tmp)\\
\lett{} project = \vabs{(cam: Vector) (x: Vector)}{}\\
\tabt\lett{} Xcam = rodrigues (vectorSlice cam 0 2) (\\
\tabt \tabt vectorSub x (vectorSlice cam 3 5) ) \\
\tabt\lett{} distorted =  
distort (vectorSlice cam 9 10) ( \\
\tabt\tabt \tabt vectorSMul (%
vectorSlice Xcam 0 1%
) (1.0/Xcam.[2]) ) \\
\tabt vectorAdd (vectorSlice cam 7 8) ( \\
\tabt\tabt vectorSMul distorted cam.[6]
)
\end{fscode}
\caption{Bundle Adjustment functions in \ladsl{}.}
\label{fig:ba_code}
\end{figure}

 \section{Related Work}
\label{sec_related}
\noindent \textbf{Automatic Differentiation. }
There is a large body of work on automatic differentiation (AD) of imperative programming languages.
Tapenade~\cite{tapenade} performs AD for a subset of C and Fortran, whereas, ADIFOR~\cite{bischof1996adifor} performs AD for Fortran programs.
Adept~\cite{adept} and ADIC~\cite{narayanan2010adic2} perform automatic differentiation for C++ by using expression templates. 
However, as we have seen in our experimental results, an AD tool such as Tapenade misses several optimisation opportunities, mainly due to their limited support for loop fusion and loop-invariant code motion.

ADiMat~\cite{bischof2002combining}, ADiGator~\cite{weinstein2016algorithm}, and Mad~\cite{forth2006efficient} perform AD for MATLAB programs, whereas MuPAD~\cite{hivert2004mupad} computes the derivatives using symbolic differentiation.
AutoGrad~\cite{maclaurin2015autograd} performs AD for Python programs that use NumPy library for array manipulation, whereas
Theano~\cite{bergstra2010theano} uses symbolic differentiation. Tensorflow~\cite{abadi2016tensorflow} performs source-to-source reverse-mode AD, and uses advanced heuristics to solve the memory inefficiencies.
ForwardDiff~\cite{revels2016forward} employs vector forward-mode AD~\cite{khan2015vector} for differentiating Julia programs. This system keeps a vector of derivative values in the dual number instead of only a single derivative value.
All these systems miss important optimisation opportunities such as loop fusion and loop-invariant code motion.

DiffSharp~\cite{baydin2015diffsharp} is an AD library implemented in F\#. 
This library provides both forward-mode and reverse-mode AD techniques.
As DiffSharp is a library implementation of AD (in contrast to \system, which implements AD as source-to-source transformation rules), 
it cannot support the simplification rules such as loop-invariant code motion, loop fusion, and partial evaluation. 
Furthermore, \system can efficiently manage memory by generating C code using DPS, whereas DiffSharp should rely on the garbage collection provided by the .NET framework for memory management. 

Stalingrad~\cite{pearlmutter2008reverse} is an optimising compiler for a dialect of Scheme with a first-class AD operator, with the support for both forward mode and reverse mode of AD. 
One of the key challenges that Stalingrad addresses is perturbation confusion~\cite{siskind2005perturbation}, which occurs for computing the derivative of the functions for which the derivatives are already computed, or the cases where we need the computation of nested differentiation~\cite{pearlmutter2007lazy}. 
We have shown how \system solves the perturbation confusion problem using a static approach thanks to the \derivk{} macro (Section~\ref{sec:nesteddiff}).
One of the main advantages of \system over Stalingrad is its support for loop transformations such as 
loop fusion and loop-invariant code motion.

Karczmarczuk~\cite{karczmarczuk1999functional} presents a Haskell implementation for both forward and reverse mode AD.
Elliott~\cite{elliott2009beautiful} improves this work by giving a more elegant implementation for its forward mode AD. 
Furthermore, Elliott~\cite{Elliott:2018:SEA:3243631.3236765} provides a generalization of AD based on category theory for implementing both forward and reverse-mode AD.
These implementations lack the optimisations offered by transformation rules, espcially loop transformations.

D*~\cite{Guenter:2007:ESD:1276377.1276512} is a symbolic differentiation system, which performs 
the factorisation of the common product terms in sum-of-product expressions. \system also performs
a similar idea by performing common-subexpression elimination and loop-invariant code motion in order to 
eliminate the redundant computation, hence, imporving the performance of forward-mode AD.
One of the key limitations of D* is that its input programs should be fully loop unrolled.
In other words, its differentiation process does not accept programs with loops. 
It would be interesting to see if \system can be 
used to optimise the computer graphics applications handled by D*.

Tensorflow~\cite{abadi2016tensorflow} and Pytorch~\cite{paszke2017automatic} are machine learning libraries 
implemented using Python. These systems are mostly based on tensor abstractions and come with a predefined set of efficient combinators for manipulating tensors.
Furthermore, they can use compilation (e.g., the XLA~\cite{leary2017xla} backend for Tensorflow and the Glow~\cite{rotem2018glow} 
backend for PyTorch) in order to perform further optimisations. 
However, these systems are quite restrictive in what constructs are efficiently supported; additional tensor operations
are not as efficient as the predefined set of tensor operators.

Lantern~\cite{wang2018backpropagation,lantern_icfp} uses the multi-stage programming~\cite{taha00staging} features provided by LMS~\cite{rompf13popl} in order to perform AD for numerical programs written in a subset of Scala. 
A key feature provided by Lantern is supporting reverse-mode AD by using \textit{delimited continuations}~\cite{danvy1990abstracting}.
To the best of our knowledge there is no support for loop fusion nor loop normalisation in Lantern. 
However, there are some imperative forms of fusion implemented for LMS~\cite{rompf13popl} which Lantern can benefit from. 
Furthermore, some form of loop-invariant code motion 
should be achieved thanks to the graph-based intermediate representation~\cite{click1995global} provided by LMS.
We plan to implement the reverse-mode AD using closures and continuations as implemented in Lantern~\cite{wang2018backpropagation} and Stalingrad~\cite{pearlmutter2008reverse}.

Pilatus~\cite{pilatus19ecoop} is a linear algebra language which is also using pull arrays for implementing 
vectors and matrices. Pilatus performs optimisations (e.g., loop fusion and algebraic optimisations) by
using multi-stage programming and rewriting facilities of Squid~\cite{Parreaux:2017:QSR:3136040.3136043,Parreaux:2017:UAS:3177123.3158101}.
However, it does not support loop-invariant code motion and loop normalisation.
Apart from supporting forward-mode AD, Pilatus also supports graph processing algorithms and 
logical probablistic programming, which can be a future direction for \system. 
IFAQ~\cite{ifaq-cgo,ifaq-vldb} and SDQL~\cite{sdql,schleich2022optimizing} are inspired by similar mathematical constructs as Pilatus and support more advanced loop optimizations, however do not support AD.

\noindent \textbf{Array Languages and Fusion.}
There are many array programming languages in the literature, 
APL~\cite{iverson1962programming} being the pioneer among them. 
There are functional array languages such as Futhark~\cite{henriksen2017futhark} and SAC~\cite{Grelck2006} with support for fusion.

In array languages fusion can be achieved by using functional arrays known as \emph{push} and \emph{pull arrays}~\cite{Svensson:2014:DPA:2636228.2636231,edsl-push,Claessen:2012:EAC:2103736.2103740}. A push-array is represented by an effectful function that, given an index and a value, will write the value into the array. 
A pull-array is represented by the length of the array and a function producing an element for a given index, similar to the \vbuildk{} construct in \fsmooth{}~\cite{DBLP:journals/pacmpl/ShaikhhaFVJ19}.
Similarly, functional programming languages use shortcut deforestation for fusing lists either by pulling the stream of data~\cite{Svenningsson:2002:SFA:581478.581491,Coutts07streamfusion} or pushing them~\cite{foldr-fusion-1}, which are implemented in Haskell using the rewrite rule facilities of GHC~\cite{jones2001playing}.
Shortcut deforestation can also be implemented as a library using multi-stage programming~\cite{fold-based-fusion,Kiselyov:2017:SFC:3009837.3009880,jfppushpull}.
It would be interesting to see how the techniques presented in this paper can be implemented on top of other
functional array programming languages (e.g., by using GHC rewrite rules or multi-stage programming).

\noindent \textbf{Numerical DSLs.}
There are many DSLs for numerical workloads. These DSLs can be classified in three categories.
The first category consists of mainstream programming languages used by data analysts such as MATLAB and R. 
These languages offer many toolboxes for performing a wide range of tasks, however, from a performance point of view the focus is only on the efficient implementation of the libraries.
The second category consists of DSLs such as Lift~\cite{Steuwer:2015:GPP:2784731.2784754}, 
Opt~\cite{devito2016opt}, 
Halide~\cite{ragan2013halide}, Diderot~\cite{chiw2012diderot}, and OptiML~\cite{sujeeth2011optiml}, which
generate parallel code from their high-level programs. 
The third category is the DSLs which focus on generating efficient machine code for fixed size linear algbra problems such as Spiral~\cite{spiral} and LGen~\cite{spampinato2016basic}.
These DSLs exploit the memory hierarchy by relying on searching algorithms for making tiling and scheduling decisions.
Except the first category, for which automatic differentiation tools exist, the other DSLs do not have any support for automatic differentiation. 
Moreover, parallel code generation and efficient machine code generation are orthogonal concepts and can be added to \system in the future.

\noindent \textbf{Correctness of AD in functional languages}
There are several recent work about proving correctness for AD in functional languages, either differentiation seen as a macro like in this paper, or as a primitive operator. 
AD is seen as a macro in \cite{barthe2020versatility,huot2020correctness,brunel2019backpropagation} while it is seen as an operator in \cite{abadi2019simple, mak2020differential}. The first three work use some version of logical relations to prove correctness of AD. \cite{brunel2019backpropagation} has a lambda calculus with a linear negation to deal with efficient reverse-mode AD. In \cite{barthe2020versatility} the focus is rather on a slight generalisation of logical relations with applications to correctness of AD. In the work \cite{huot2020correctness}, a non trivial denotational semantics is given to a simply-typed lambda calculus and category theory and correctness of both a forward and reverse-mode AD is proved using category theory. In \cite{abadi2019simple} they consider a first-order language with recursion and reverse-mode AD, seen as a primitive of the language, is shown correct. \cite{mak2020differential} develops a lambda-calculus with pullbacks for proving correct reverse-mode AD as a primitive in a higher-order language.

\section{Conclusions}
In this paper we have demonstrated how to efficiently compute the derivate of a program.
The key idea behind our system is exposing all the constructs used in differentiated programs to the underlying compiler.
As a result, the compiler can apply various loop transformations such as loop-invariant code motion and loop fusion for optimizing differentiated programs.
We have shown how \system outperforms the existing AD tools on micro benchmarks and real-world machine learning and 
computer vision applications.


  \section*{Acknowledgments}
The first author thanks Huawei for their support of the distributed
data management and processing laboratory at the University of Edinburgh.
The second author is supported by a Royal Society University Research Fellowship.

\bibliographystyle{jfp}
\bibliography{ref} 





\end{document}